%% file: paper_dmkd.tex
\documentclass[sn-mathphys,authoryear]{sn-jnl}

\usepackage{graphicx}%
\usepackage{multirow}%
\usepackage{amsmath,amssymb,amsfonts}%
\usepackage{amsthm}%
\usepackage{mathrsfs}%
\usepackage[title]{appendix}%
\usepackage{xcolor, colortbl}%
\usepackage{textcomp}%
\usepackage{manyfoot}%
\usepackage{tabularx, booktabs}%
\usepackage{algorithm}%
\usepackage{algorithmicx}%
\usepackage{algpseudocode}%
\usepackage{listings}%
\usepackage{xspace}%

\usepackage{url}
\usepackage{multicol}
\usepackage[normalem]{ulem}
\usepackage{lipsum}
\usepackage{bm}
\usepackage{stfloats}
\usepackage{caption}
\usepackage[skip=0cm,list=true,labelfont=it]{subcaption}
\usepackage[show]{chato-notes}
\input{macros}

\newtheorem{lemma}{Lemma}
\newtheorem{theorem}{Theorem}
\newtheorem{corollary}{Corollary}
\newtheorem{proposition}[theorem]{Proposition}%

\begin{document}

\title[]{Forming Coordinated Teams that Balance Task Coverage and Expert Workload}


\author*[1]{\fnm{Karan} \sur{Vombatkere}}\email{kvombat@bu.edu}

\author[2]{\fnm{Aristides} \sur{Gionis}}\email{argioni@kth.se}

\author[1]{\fnm{Evimaria} \sur{Terzi}}\email{evimaria@bu.edu}

\affil[1]{\orgdiv{Department of Computer Science}, \orgname{Boston University}, \orgaddress{\city{Boston}, \country{USA}}}

\affil[2]{\orgdiv{Division of Theoretical Computer Science}, \orgname{KTH Royal Institute of Technology}, \orgaddress{\city{Stockholm}, \country{Sweden}}}

\abstract{
We study a new formulation of the team-formation problem, 
where the goal is to form teams to work on a given set of tasks requiring different skills.
Deviating from the classic problem setting
where one is asking to cover all skills of each given task, 
we aim to cover as many skills as possible
while also trying to minimize the maximum workload among the experts. 
We do this by combining penalization terms for the coverage and load constraints into one objective. 
We call the corresponding assignment problem {\bbalance}, and show that it is \NP-hard. 
We also consider a variant of this problem, 
where the experts are organized into a graph, which encodes how well they work together. 
Utilizing such a coordination graph, 
we aim to find teams to assign to tasks such that each team's radius does not exceed a given threshold. We refer to this problem as {\nbalance}.
We develop a generic template algorithm for approximating both problems in polynomial time, 
and we show that our template algorithm for
{\bbalance} has provable guarantees.
We describe a set of computational speedups that we can apply to our algorithms and make them scale for reasonably large datasets. From the practical point of view, we demonstrate how to efficiently tune the two parts of the objective and tailor their importance to a particular application. Our experiments with a variety of real-world datasets demonstrate the utility of our problem formulation as well as the efficiency of our algorithms in practice.}


\maketitle

\section{Introduction}\label{sec:intro}
\input{intro.tex}

\section{Related Work}\label{sec:related}
\input{related.tex}

\section{Problem Definitions}\label{sec:preliminaries}
\input{preliminaries.tex}

\section{Algorithms for {\bbalance}}\label{sec:algos-bbalance}
\input{algos-bbalance.tex}

\section{Algorithms for {\nbalance}}\label{sec:algos-nbalance}
\input{algos-nbalance.tex}

\section{Experiments}\label{sec:experiments}
\input{experiments.tex}

\section{Conclusions}\label{sec:conclusions}
\input{conclusion.tex}

\mpara{Acknowledgments:} 
Evimaria Terzi and Karan Vombatkere are supported by 
NSF grants III 1908510 and III 1813406 as well as a gift from Microsoft.
Aristides Gionis is supported by ERC Advanced Grant REBOUND (834862), 
EC H2020 RIA project SoBigData++ (871042), and 
the Wallenberg AI, Autonomous Systems and Software Program (WASP) 
funded by the Knut and Alice Wallenberg Foundation.

\bibliography{references}

\appendix
\newpage
\section{The {\ntbalance} Problem}\label{sec:appendix-proof}
\input{appendix-proof}

\end{document}

%% file: macros.tex
\newcommand{\bbalance}{\textsc{Bal\-anc\-ed-Cov\-er\-age}}
\newcommand{\nbalance}{\textsc{Net\-work-Bal\-anc\-ed-Cov\-er\-age}}
\newcommand{\ntbalance}{\textsc{Teams-Match\-ing}}

\newcommand{\independentset}{\textsc{Independent Set}}

\newcommand{\NP}{\ensuremath{\mathbf{NP}}\xspace}
\newcommand{\bigO}{\ensuremath{\mathcal{O}}\xspace}

\newcommand{\tasks}{\ensuremath{\mathcal{J}}\xspace}
\newcommand{\task}{\ensuremath{J}\xspace}
\newcommand{\experts}{\ensuremath{\mathcal{X}}\xspace}
\newcommand{\expert}{\ensuremath{X}\xspace}
\newcommand{\skills}{\ensuremath{S}\xspace}
\newcommand{\cov}{\ensuremath{C}\xspace}
\newcommand{\load}{\ensuremath{L}\xspace}
\newcommand{\maxload}{\ensuremath{L_{\max}}\xspace}
\newcommand{\assignment}{\ensuremath{A}\xspace}
\newcommand{\edges}{\ensuremath{E}\xspace}
\newcommand{\threshold}{\ensuremath{\tau}\xspace}
\newcommand{\objective}{\ensuremath{F}\xspace}
\newcommand{\opt}{\ensuremath{\mathit{OPT}}\xspace}
\newcommand{\numtasks}{\ensuremath{m}\xspace}
\newcommand{\numexperts}{\ensuremath{n}\xspace}

\newcommand{\teams}{\ensuremath{\mathcal{T}}\xspace}
\newcommand{\team}{\ensuremath{T}\xspace}
\newcommand{\radius}{\ensuremath{R}\xspace}
\newcommand{\maxradius}{\ensuremath{R_{\max}}\xspace}
\newcommand{\diam}{\ensuremath{\mathit{Diam}}\xspace}

\newcommand{\avgteamsize}{\ensuremath{\overline{S}}\xspace}
\newcommand{\maxteamsize}{\ensuremath{S_{\max}}\xspace}
\newcommand{\avgteamradius}{\ensuremath{\overline{R}}\xspace}
\newcommand{\avgteamdensity}{\ensuremath{\overline{\delta}}\xspace}
\newcommand{\avgteampairwise}{\ensuremath{\overline{\rho}}\xspace}

\newcommand{\bibsonomy}{\textit{Bbsm}\xspace}
\newcommand{\bibsonomyone}{{\bibsonomy}\textit{-1}\xspace}
\newcommand{\bibsonomytwo}{{\bibsonomy}\textit{-2}\xspace}
\newcommand{\bibsonomythree}{{\bibsonomy}\textit{-3}\xspace}

\newcommand{\imdb}{{\textit{IMDB}}\xspace}
\newcommand{\imdbone}{{{\imdb}\textit{-1}}\xspace}
\newcommand{\imdbtwo}{{{\imdb}\textit{-2}}\xspace}
\newcommand{\imdbthree}{{{\imdb}\textit{-3}}\xspace}

\newcommand{\freelancer}{{\textit{Freelancer}}\xspace}
\newcommand{\guru}{{\textit{Guru}}\xspace}

\newcommand{\greedy}{\normalfont{\texttt{Greedy}}\xspace}
\newcommand{\thresholdgreedy}{\normalfont{\texttt{Thres\-hold\-Greedy}}\xspace}
\newcommand{\noupdategreedy}{\normalfont{\texttt{No\-Update\-Greedy}}\xspace}
\newcommand{\taskgreedy}{\normalfont{\texttt{Task\-Greedy}}\xspace}
\newcommand{\smartrandom}{\normalfont{\texttt{Smart\-Random}}\xspace}
\newcommand{\lpsetcover}{\normalfont{\texttt{LP\-Cover}}\xspace}
\newcommand{\networkbalance}{\normalfont{\texttt{NThreshold}}\xspace}
\newcommand{\networkbalancelp}{\normalfont{\texttt{NThreshold-R-LP}}\xspace}
\newcommand{\networkbalancegreedy}{\normalfont{\texttt{NThreshold-R-Greedy}}\xspace}

\newcommand{\teampruning}{\normalfont{\texttt{Team\-Prun\-ing}}\xspace}

\newcommand{\lpallradii}{\normalfont{\texttt{NThres\-hold-All-LP}}\xspace}

\newcommand{\greedyindividual}{\normalfont{\texttt{Greedy\-Individual}}\xspace}

\newcommand{\formcandidateteams}{\normalfont{\texttt{Can\-di\-date\-Teams}}\xspace}

\newcommand{\assignteams}{\normalfont{\texttt{Assign\-Teams}}\xspace}

\newcommand{\formcandidateteamsR}{\normalfont{\texttt{Can\-di\-date\-Teams-R}}\xspace}
\newcommand{\formcandidateteamsallR}{\normalfont{\texttt{Can\-di\-date\-Teams-AllR}}\xspace}

\newcommand{\spara}[1]{\smallbreak\noindent{\bf{#1}}}
\newcommand{\mpara}[1]{\medskip\noindent{\bf{#1}}}

\newcommand*{\belowrulesepcolor}[1]{%
	\noalign{%
		\kern-\belowrulesep
		\begingroup
		\color{#1}%
		\hrule height\belowrulesep
		\endgroup
	}%
}
\newcommand*{\aboverulesepcolor}[1]{%
	\noalign{%
		\begingroup
		\color{#1}%
		\hrule height\aboverulesep
		\endgroup
		\kern-\aboverulesep
	}%
}

\newcommand{\squishlist}{\begin{list}{$\bullet$}
  { \setlength{\itemsep}{0pt}
     \setlength{\parsep}{3pt}
     \setlength{\topsep}{3pt}
     \setlength{\partopsep}{0pt}
     \setlength{\leftmargin}{1.5em}
     \setlength{\labelwidth}{1em}
     \setlength{\labelsep}{0.5em} } }
\newcommand{\squishend}{
  \end{list}  }

\newtheorem{problem}{Problem}

%% file: intro.tex
The abundance of online and offline labor markets (e.g., 
Guru, Freelancer, online scientific collaborations, etc.) 
has motivated a lot of work on the \emph{team-formation} problem.
In the team-formation setting, 
the input consists of 
($i$) a task, or a collection of tasks, 
so that each task requires a set of skills, and 
($ii$) a set of experts, 
where each expert is also associated with a set of skills.
The objective is to identify one team, or one team for every task, 
such that all the skills in every task are covered by at least one team member.  
Notably, the majority of works in team-formation research 
require \emph{complete} coverage of the skills of the input tasks~\citep{anagnostopoulos10power,anagnostopoulos2012online,anagnostopoulos18algorithms,bhowmik2014submodularity,	kargar2013finding,kargar2011discovering,kargar2012efficient,lappas2009finding,majumder2012capacitated,li2015team,li2015replacing,
	li2017enhancing,rangapuram2013towards,yin2018social}.
The differences among existing papers 
lie in the way they define the ``goodness'' of a team. 
For example, in some cases they optimize the communication cost of the team, 
while in other cases they optimize the load of the experts, or their associated cost.

We motivate the inherent trade off between task coverage and expert workload using the example of Fig.~\ref{fig:intro-example}. 
Consider the bipartite graph with experts as one set of nodes and tasks as the other. The edges shown in the graph represent the expert-task assignments.
The assignment on the left, achieves 100\% coverage for all three tasks; 
however Charlie has a workload of 3, Bob and David each have workload of 2 while Alice is not assigned to any task. 
However, the assignment on the right -- which allows for partial coverage -- does not cover any task 100\%, yet it is more balanced in terms of expert workload; all experts now have a workload of 1.

\begin{figure*}[h!]
\centering
\begin{subfigure}[b]{0.49\textwidth}
    \centering
    \includegraphics[width=\textwidth]{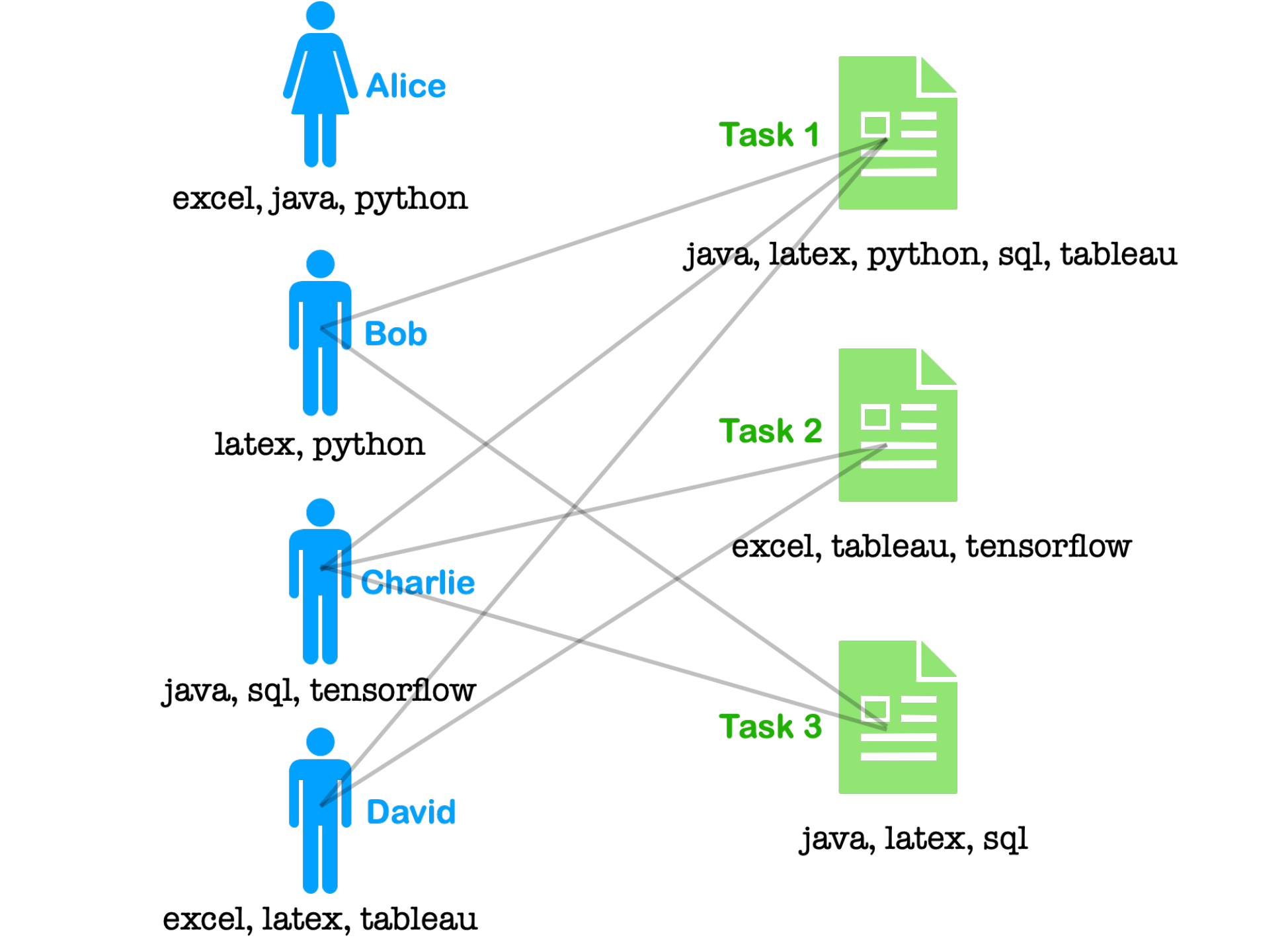}
    \caption{Assignment with full coverage for all tasks and unbalanced expert workload.}
\end{subfigure}
\hfill
\begin{subfigure}[b]{0.49\textwidth}
    \centering
    \includegraphics[width=\textwidth]{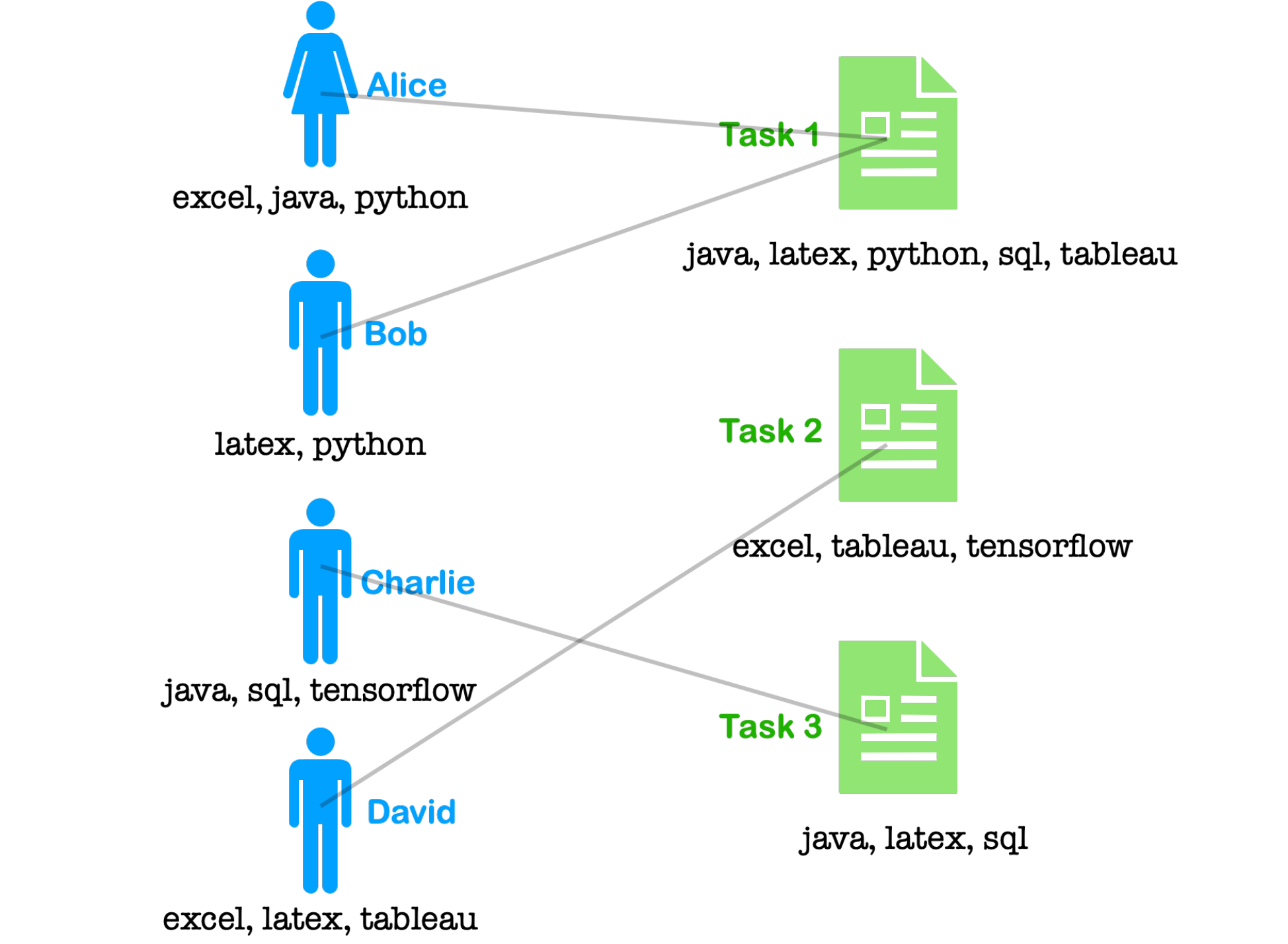}
        \caption{Assignment with partial coverage and balanced expert workload.}
\end{subfigure}
\caption{Motivating example with 4 experts and 3 tasks.}
\label{fig:intro-example}
\end{figure*}

In this paper, we propose team-formation problems
where the goal is to assign experts to a set of input tasks 
such that the task coverage is maximized,
and at the same time, the maximum workload among the experts used is minimized.
This trade-off suggests that we need not always cover the skills of every task completely, since covering a large fraction of their required skills might be sufficient. Also, given that overworked experts do not 
perform well, we penalize expert overloading by minimizing the maximum number of 
tasks assigned to an expert.  Therefore, for an assignment $\assignment$ of experts to tasks,
our goal is to maximize the combined objective:
\begin{equation}\label{eq:intro-objective}
	\objective(\assignment) = \lambda\cov(\assignment)-\maxload(\assignment),
\end{equation}
where $\cov(\assignment)$ is the sum of the fraction of the skills
of the tasks being covered by their assigned
experts and $\maxload(\assignment)$ is the maximum number of tasks assigned to a single expert.

Although we normalize the two terms of the objective (Eq.~\eqref{eq:intro-objective}) 
and make them comparable, 
in certain applications we may want to aim for different trade off between 
the coverage and maximum-load terms.
Thus, we incorporate the 
\emph{balancing coefficient} $\lambda$, 
which enables an effective tuning
of the importance of the two terms.  
We call this problem {\bbalance}.

Often, the experts are organized in a network, which encodes how well 
experts can work with each other. 
In the presence of such information, 
we extend the {\bbalance} problem so that the teams assigned to tasks 
have the property that their radius is not larger than a pre-specified threshold. 
The motivation is for teams to have small coordination cost 
and be able to work well with each other.
We call this version of the problem {\nbalance}.

We show that the two problems we define, 
{\bbalance} and {\nbalance}, are \NP-hard.

From the application point of view, it makes sense to 
relax the hard constraint of full coverage; in practice,
skills in tasks are often overlapping. For example, consider a task
requiring skills: \textit{advertising, internet advertising, 
Facebook advertising, online marketing, social network platforms}. 
Clearly, these are overlapping and not all of them need to 
be covered.  
Additionally, minimizing the maximum expert workload is desirable
for better team performance.

From the algorithmic point of view, 
optimizing the above objective, with or without the radius constraint in the teams, is challenging;   
the function itself may take negative values.  
Therefore, it does not admit
multiplicative approximation guarantees.
Although the coverage part of the objective ($\cov(\cdot)$) is a monotone submodular function, 
the  maximum load part does not have a predictable form (i.e., it is not linear or convex).  
Therefore, recent techniques~\citep{harshaw19submodular,mitra21submodularplus}
on submodularity optimization cannot be applied.  
However, we adopt from these works a weaker notion of approximation and aim to find an assignment $\assignment$ such that:
\begin{equation}\label{eq:intro-approx}
	\lambda \cov(\assignment)-\maxload(\assignment)\geq \alpha \left(\lambda \cov(\opt)\right)-\maxload(\opt),
\end{equation}
where $\opt$ is the optimal solution to the {\bbalance} or the {\nbalance} problems. 
In this case, $\alpha\leq 1$ is an approximation 
guarantee that better fits functions like ours. In this paper, we show that for the {\bbalance} problem, 
we can design a polynomial-time algorithm with  $\alpha = \left(1-1/e\right)$, 
which is probably 
the best we can hope for our objective given that the $\cov(\cdot)$ is monotone and submodular. Unfortunately, the {\nbalance} problem
appears to be significantly harder and for that we only present a heuristic algorithm, which works extremely well in 
our extensive experiments;  
designing an approximation algorithm for {\nbalance} 
is an open problem.  We note however, that both our algorithms follow the same generic design template --- which we  believe is interesting by itself.
We also show that our algorithms admit a lot of practical speedups, which are a consequence of the structure of our objective function.

Our experimental results demonstrate that our algorithms are practical in terms of their running time, 
and they output assignments with high total task coverage and very low maximum load. Comparisons with a number of baselines
inspired by existing works show that our algorithms consistently outperform them.  In our experiments, we also compare the characteristics of the teams found by our algorithms for {\bbalance} and {\nbalance}. Our findings are consistent with our expectation that
the solutions to the {\nbalance} problem are teams that are more
cohesive in the graph that encodes the experts' ability to work together; that is, the teams found as solutions to {\nbalance} have higher density in this graph.

%% file: related.tex
In this section, we highlight some related work in team formation and discuss its relationship to our problem and the algorithmic techniques we propose in this paper. 
To the 
best of our knowledge there is no other paper that addresses  the exact
{\bbalance} and {\nbalance} problems we discuss here.

\spara{Team formation with a single task:}
A large body of work in team formation assumes that there is a single task, which requires a set of skills. 
Additionally, there are experts who possess a subset of skills. The goal is to identify a ``good'' subset of the experts that collectively cover 
\emph{all} the skills required by the task. 
In the majority of this work ~\citep{bhowmik2014submodularity,kargar2013finding,kargar2011discovering,kargar2012efficient,lappas2009finding,majumder2012capacitated,li2015team,li2015replacing,li2017enhancing,rangapuram2013towards,yin2018social, hamidi2023variational, kou2020efficient, berktacs2021branch}, the requirement that all skills of the tasks are covered is a hard constraint.  Different problem formulations arise from the different definitions of the ``goodness'' of a team (i.e., small communication cost). 
The work by \citet{kargar2011discovering} and \citet{rangapuram2013towards} consider different graph communication costs in an offline setting to find a team of experts. However, these works consider single tasks with a complete coverage requirement, and consequently do not consider the trade-off between communication cost and expert workload. A subsequent related work by \citet{kargar2013finding} considers a bi-criteria optimization for complete coverage of a single task, to minimize both the communication cost as well as the personnel cost of the teams formed. 
While this work has a similar flavor to ours, it is important to note that our {\nbalance} problem formulation is a generalization of their work since we relax the complete coverage constraint and extend the offline scenario to forming teams for \emph{multiple} tasks simultaneously.

More recently, there has been some work
aiming to maximize a combined objective of task coverage minus 
the sum of the costs of the experts participating in the team~\citep{nikolakaki21efficient, dorn2010composing}.
In other words, the goal is to maximize a submodular (i.e., coverage) minus a linear function.  
The setting is similar to ours and it could be expanded to consider multiple tasks. 
However, the linear part of the objective is more structured than the maximum load we are considering here.  
As a result, the algorithmic techniques that were developed by~\citet{nikolakaki21efficient} cannot be applied to our setting. 
On the other hand, the work of~\citet{dorn2010composing} balances coverage with the team's communication cost on a graph. 
However, since their work considers only single tasks, their heuristics do not consider the workload of experts.

\spara{Team formation with multiple tasks:}
There is a number of papers that consider multiple tasks~\citep{anagnostopoulos10power,anagnostopoulos2012online,anagnostopoulos18algorithms,nikolakaki20finding, selvarajah2021unified}, 
most of which focus on the \emph{online} version of the problem, 
where tasks arrive in a streaming fashion.
The offline versions of these problems are also \NP-hard. 
Regardless of whether we study the offline or the online version of these problems, 
the setting is to minimize the load of the most loaded expert
while covering completely all the skills in all tasks. 
Our setting is a relaxation of these problems
aiming to maximize a combined objective of coverage minus load.
Also, this line of work considers a minimization problem 
while in this paper we study a maximization problem, 
and therefore, the approximation bounds we seek are different.

\spara{Approximation framework:} One of the 
intricacies of our objective function in the {\bbalance} and the {\nbalance} problems is that it can potentially take negative values.
The approximation of such functions requires a weaker notion of approximation that is different from the multiplicative approximation bounds~\citep{harshaw19submodular,mitra21submodularplus}. 
Although we adopt this framework in our case,
our objective function does not fall into any of the categories that have been studied before.  Therefore, we need to design new algorithms for our setting.

%% file: preliminaries.tex
In this section, we  describe our notation and basic concepts, and formally define the {\bbalance} and {\nbalance} problems.

\subsection{Preliminaries}
\spara{Tasks, Experts and Skills.}
Throughout, we  assume a set of $\numtasks$ tasks
$\tasks = \{\task_1,\allowbreak\ldots,\allowbreak\task_\numtasks\}$ and a set of $\numexperts$ experts
$\experts = \{\expert_1,\allowbreak\ldots,\allowbreak\expert_\numexperts\}$.
We  also assume a set of skills $\skills$ such that every task \emph{requires} a set of skills
and every expert \emph{masters} a set of skills.  That is, for every task $\task_j \subseteq \skills$
and for every expert $\expert_i\subseteq \skills$.
Note that each skill could have an associated weight, but that doesn't change the problem complexity in our setting.

\spara{Assignments.}
An \emph{assignment} of experts to tasks is represented by a binary matrix~$\assignment$, 
such that $\assignment(i,j)=1$ if expert $\expert_i$ is assigned to task $\task_j$; 
otherwise $\assignment(i,j)=0$. Alternatively, one can view an assignment $\assignment$ 
as a bipartite graph with the nodes on one side corresponding to the experts
and the nodes on the other side corresponding to the tasks; edge $(i,j)$ exists if and only if
$\assignment(i,j)=1$.
Finally, we often view an assignment $\assignment$ as a \emph{set} of its $1$-entries.

\spara{Teams.}
Given an assignment $\assignment$, we can find the set of $\numtasks$ teams associated with $\assignment$, denoted by $\teams_\assignment$, such that $\team_j\in \teams_\assignment$ is the team of experts associated with task $\task_j$: i.e., $\team_j=\{X_i\mid A(i,j)=1\}$.
We use the additive skill model~\citep{anagnostopoulos10power} to define the expertise of a team: a skill is covered by the team if there exists at least one member on the team who has that skill. 

\spara{Task coverage.}
Given an assignment $\assignment$, we define the \emph{coverage} of task $\task_j$ as the fraction of the skills in $\task_j$ covered
by the experts assigned to $\task_j$. Formally,
\begin{equation*}
	\cov(\task_j\mid \assignment) = \frac{|(\cup_{i: \assignment(i,j) = 1} \expert_i) \cap \task_j|}{|\task_j|}.
\end{equation*}
Note that $0\leq \cov(\task_j\mid\assignment)\leq 1$.

Given an assignment $\assignment$, and the individual task coverages $\cov(\task_j\mid \assignment)$, we define the
\emph{overall coverage} as the sum of the individual task coverages:

\begin{equation*}
	\cov(\assignment) = \sum_{j=1}^m\cov(\task_j\mid \assignment).
\end{equation*}

\spara{Expert workload.}
Additionally, given an assignment $\assignment$, we define the \emph{load} of expert $\expert_i$ in $\assignment$ as the number of
tasks that $\expert_i$ is assigned to. Formally, 
\begin{equation*}
	\load(\expert_i\mid \assignment) = \sum_{j} \assignment(i,j).
\end{equation*}

Given an assignment $\assignment$, the \emph{maximum load} among all experts is
\begin{equation*}
	\maxload (\assignment) = \max_i\load(\expert_i\mid\assignment).
\end{equation*}

\spara{Coordination costs.}
We represent pairwise (symmetric) coordination costs between individual experts using edge weights on a graph $G = (\experts, \edges)$. The vertices of $G$ correspond to the set of experts, $\experts$ and the edges, $\edges$ are characterized by a metric distance function $d: \edges \to \mathbb{R}_{\geq 0}$. 
Although in the experimental section we discuss how $d(\cdot,\cdot)$ is computed, we point out here that we assume that there is a non-negative distance between any two experts;
that is, $d(\expert_i,\expert_j)\geq 0$ for every $\expert_i\neq\expert_j$. 
We also assume that $d(\cdot,\cdot)$ is a metric.

\spara{Team radius and diameter.}
We first define the \textit{radius} of a team~$\team$ as $\radius(\team) = \min_{\expert_i\in\team} \max_{\expert_j\in\team} d(\expert_i, \expert_j)$. 
The \textit{diameter} $\diam (\team)$ of a team $\team$
corresponds to the longest distance between any two experts on that team $\team$, and is defined as
$\diam (\team) = \max_{\expert_i, \expert_j \in \team} d(\expert_i, \expert_j).$
Since we consider a discrete metric space, it follows that:
$\frac{1}{2}\diam (\team) \le \radius (\team) \le \diam (\team)$.

Given an assignment $\assignment$
and the set of teams $\teams_{\assignment}$ associated with it, 
we define 
$\radius_{\max} (\assignment) = \max_{\team \in \teams_\assignment} \radius (\team)$.

\subsection{The \large{\bbalance} Problem} 
We now define the {\bbalance} problem as follows:
\begin{problem}[{\bbalance}]{\label{problem:balance}}
	Given a set of $\numtasks$ tasks $\tasks=\{\task_1,\ldots , \task_\numtasks \}$ and a set of $\numexperts$ experts $\experts = \{\expert_1,\ldots \expert_\numexperts \}$
	find an assignment $\assignment$ of experts to tasks such that 
	\begin{equation}\label{eq:objective}
		\objective(\assignment) = \lambda \cov( \assignment)  - \maxload (\assignment)
	\end{equation}
	is maximized.
\end{problem} 

The following observations provide some insight on our problem definition. 

\smallskip
\noindent
\emph{Observation 1:} 
The objective function (see Eq.~\eqref{eq:objective}) consists of two terms: the coverage, which we want to maximize, and the maximum load, which we want to minimize.  These two terms act in opposition to one another and a good solution needs to identify a ``balance point" between the experts being used and the coverage being achieved.  
Thus, the  number of experts in the solution is not
constrained in the definition of {\bbalance} itself.

\smallskip
\noindent
\emph{Observation 2:}
The parameter $\lambda$ is referred to as a \emph{balancing coefficient}. 
Depending on the application, one may need to tune the importance of the two parts of the objective. The balancing coefficient $\lambda$ should be thought
of as a factor that adds flexibility to the model and allows for flexibility in the team-construction process. 
A detailed discussion on how we set the value of $\lambda$ in practice is provided in Section~\ref{sec:discussion}. 

\smallskip
\noindent
\emph{Observation 3:}
The objective function $\objective (\cdot)$ is a summation of two quantities: coverage
and maximum load.  The coverage is a sum of normalized coverages multiplied by $\lambda$ and therefore it is a quantity that takes real values between
$[0,\lambda m]$; the value of $0$ is achieved when no task is covered and the value $\lambda m$ is achieved when all tasks are fully covered. 
The maximum load is a term that takes integer values between $\{0,m\}$, as the maximum load of an expert is between $0$ and the total 
number of tasks.  Therefore, the values of the two quantities are comparable and they can be added (or subtracted).

\smallskip
\noindent
\emph{Observation 4:} Finally, it can be shown that 
the first part of the objective, i.e.,  $\cov(\assignment)$, is a monotone and 
submodular function.  
We state this in the following proposition:

\begin{proposition}\label{prop:monotone-submodular}
The overall coverage function:
$\cov(\assignment) =  \sum_{j=1}^m\cov(\task_j\mid \assignment)$ is a monotone and submodular function.
\end{proposition}

The proof of this proposition is omitted as it is relatively simple: $\cov(\cdot)$ is a monotone submodular function as it is a summation of coverage functions that are known to be monotone and submodular~\citep{krause2014submodular}.

\spara{Problem complexity:} 
Clearly, there are cases where our problem is easy to solve: for example, if there is only one task then the best solution is the one
assigning every expert to this one task.  
However, our problem is \NP-hard in general.
Using similar observations as the ones made by~\citet{anagnostopoulos10power}
we can show that the {\bbalance} problem is \NP-hard even when there are only two tasks.

\begin{theorem}\label{thm:bbalance}
The {\bbalance} problem is \NP-hard even for $m=2$.
\end{theorem}

\begin{proof}
We provide a proof of \NP-hardness for $\lambda=1$, via a reduction from the monotone satisfiability or \textsc{MSat} problem. 
The \textsc{MSat} problem is a version of satisfiability where clauses have only positive or only negative literals, and is known to be \NP-hard~\citep{lewis1983michael}.

An instance of \textsc{MSat} is specified by a set of clauses,
each clause being a disjunction of literals that are all positive or all negative.
Given an instance of the \textsc{MSat} problem we create an instance of the {\bbalance} problem, 
as follows.
\squishlist
\item every clause $C_\ell$ in \textsc{MSat} corresponds to a skill in our problem;
\item every literal $x_i$ in \textsc{MSat} corresponds to an expert $\expert_i$ in our problem; $\expert_i$ has skills 
that correspond to the clauses
in which $x_i$ or its negation participates;
\item we create two tasks $\tasks =\{ \task_1,\task_2\}$; $\task_1$ requires the skills that correspond to the clauses with positive literals and $\task_2$ requires the skills that correspond to the clauses with negative literals.
\squishend

We can show that
the instance of the {\bbalance} problem we have created  has 
a solution of value $1$ if and only if the corresponding instance of the
\textsc{MSat} problem has a satisfying assignment. 
For the one direction assume that there is a satisfying assignment in \textsc{MSat}. 
For a literal $x_i$ that is set to \texttt{true} the expert $\expert_i$ is assigned only to $\task_1$.
For a literal $x_i$ that is set to \texttt{false} the expert $\expert_i$ is assigned only to $\task_2$.
All experts are assigned to exactly one task, and thus, $\maxload=1$. 
Furthermore, both tasks are fully covered, and thus, the total coverage is 2. 
Therefore the value of the instance of the {\bbalance} problem is $2-1=1$.

For the other direction assume that the {\bbalance} objective is $1$. 
Notice that the possible values for \maxload are 0, 1, and 2. 
For the {\bbalance} objective to be $1$, the max load \maxload can only be 1. 
Indeed, if $\maxload=0$ or $2$ the value of the objective is less than or equal to ~$0$. 
When $\maxload=1$ then for the objective to be $1$, the total coverage should also be equal to~$2$. 
This only happens if there is an assignment of the experts to the two tasks 
such that each expert is assigned to exactly one task and each task is covered completely, 
which essentially means that there is a satisfying assignment to the \textsc{MSat} problem.
\end{proof}
	
\subsection{The \large{\nbalance} Problem} 
We now define the {\nbalance} problem as follows:

\begin{problem}[{\nbalance}]{\label{problem:nbalance}}
	Given a set of $\numtasks$ tasks 
    $\tasks=\{\task_1,\allowbreak\ldots,\allowbreak\task_\numtasks \}$, 
    a set of $\numexperts$ experts 
    $\experts = \{\expert_1,\allowbreak\ldots,\allowbreak\expert_\numexperts \}$, 
    a distance function $d(\cdot,\cdot)$ between any two experts, 
    and a radius constraint $r$, find an assignment $\assignment$ of experts to tasks such that 
	\begin{equation}\label{eq:objective-nbalance}
		\objective(\assignment) = \lambda \cov( \assignment)  - \maxload (\assignment)
	\end{equation}
	is maximized, and each task has a team of radius at most $r$, i.e.,
     $\radius_{\max} (\assignment) \leq r$.
\end{problem} 

\begin{theorem}\label{thm:nbalance}
The {\nbalance} problem is \NP-hard even for $m=2$ and any radius constraint $r$.
\end{theorem}
The proof of Theorem~\ref{thm:nbalance} follows from the fact that the {\nbalance} problem is a generalization of the {\bbalance} problem. 

%% file: algos-bbalance.tex
The objective function $\objective(\cdot)$ of the {\bbalance} problem is defined as the  
dif\-fer\-ence between a submodular function (\emph{coverage}) and another function (\emph{maximum load}), 
which does not have a concrete form i.e., it is neither linear nor convex.  Therefore, existing results
on optimizing a submodular function~\citep{nemhauser78best} 
or a submodular plus a linear or convex function~\citep{harshaw19submodular,mitra21submodularplus,nikolakaki21efficient} 
are not applicable.

We describe {\thresholdgreedy}, a polynomial-time algorithm 
for the  {\bbalance} problem. 
We show {\thresholdgreedy} outputs an assignment
$\assignment$ such that:
\begin{equation*}
	\cov(\assignment)-\maxload(A) \geq \left(1-\frac{1}{e}\right) \cov(\opt)-\maxload(\opt),
\end{equation*}
or equivalently,
\begin{equation}\label{eq:approximation}
\objective(\assignment) \geq \left(1-\frac{1}{e}\right) \objective( \opt)  - \frac{1}{e}\maxload (\opt).
\end{equation}
where {\opt} is the optimal solution to the {\bbalance} problem.

The approximation guarantee described in Eq.~\eqref{eq:approximation}
is a weaker form of approximation than standard multiplicative approximation guarantees.
However, this is used in cases, like ours, where the objective function is not guaranteed to be 
positive~\citep{harshaw19submodular,mitra21submodularplus,nikolakaki21efficient}.

\subsection{The {\thresholdgreedy} Algorithm}
A key observation that {\thresholdgreedy} exploits is that
the value of $\maxload$ is an integer in  $[0, m]$, where $m$ is the total number
of tasks. 
Therefore, {\thresholdgreedy} proceeds by finding an assignment for each possible 
value of $\maxload$ and then returns the assignment with the best value of $\objective (\cdot)$.
The pseudocode is given in Algorithm~\ref{algo:thresholdgreedy}.

\begin{algorithm}[t]
	\caption{The {\thresholdgreedy} algorithm.}\label{algo:thresholdgreedy}
	\begin{algorithmic}[1]
		\Require  Set of $m$ tasks $\tasks$, $n$ experts $\experts$, and $\lambda$
		\Ensure An assignment of experts to tasks $\assignment$
		\State $\assignment \leftarrow \emptyset, \objective_{\max} = 0$
		\For{$\threshold = 1,...,m$}  \label{ln:outer}
		\State Create the set of experts $\experts_\threshold$, with $\threshold$ copies of each expert
		\State  $A_\threshold = {\greedy} (\experts_\threshold,\tasks)$ \label{ln:coverage}
		\State Compute $\objective_\threshold = \lambda \cov(\assignment_\threshold) - \threshold$\label{ln:objective}
		\If{$\objective_\threshold \geq \objective_{max}$}
		\State $\objective_{max} = \objective_\threshold$
		\State $\assignment \leftarrow \assignment_\threshold$
		\EndIf
		\EndFor\\
		\Return $\assignment$
	\end{algorithmic}
\end{algorithm}

In more detail,
given a threshold $\threshold$ on the value of $\maxload$,  any expert can be used at most $\threshold$ times.
Conceptually, this means that there are $\threshold$ copies of every expert and we find $\assignment_\threshold$ to be the {\greedy}  assignment corresponding to  $\threshold$;  
$\assignment_\threshold$ is found by invoking
the standard {\greedy} algorithm~\citep{vazirani2013approximation} --- for optimizing a monotone submodular function --- 
in order to 
optimize the overall coverage i.e., $\cov(\cdot)$. 
After trying all possible $\numtasks$ values of $\threshold$, we pick the assignment $\assignment_\threshold$ 
that has the maximum value of the objective $\objective(\assignment_{\threshold})$. 

The {\greedy}
algorithm for solving the  coverage problem for input experts $\experts_\threshold$
and tasks $\tasks$ (Line~\ref{ln:coverage} of Algorithm~\ref{algo:thresholdgreedy}) greedily assigns experts in $\experts_\threshold$
to tasks until there are no more experts available.  At step $\ell+1$,  {\greedy}  
finds assignment $A_\threshold^{\ell+1}$ by extending~$\assignment_\threshold^{\ell}$ with the addition 
of expert $i$ assigned to task $j$ so that its \emph{marginal gain}
\begin{equation}\label{eq:marginal-gain}
	\tilde{\cov}((i,j)\mid A^\ell) =  \cov\left(\assignment_\threshold^{\ell}\cup (i,j)\right)-\cov\left(\assignment_\threshold^{\ell}\right)
\end{equation}  
is maximized.  
During this greedy assignment, each one of the $\threshold$ copies of every expert is considered as a
different expert and once a copy is assigned to a task the copy is removed from the candidate experts. 

\subsection{Approximation}
Here, we prove our approximation result for {\thresholdgreedy}, as outlined 
already in Eq.~\eqref{eq:approximation}. Before proving the main theorem 
we need the following lemma: 

\smallskip
\begin{lemma}\label{lm:fixedthreshold}
	Let $\assignment_\threshold$ be the assignment of experts to tasks returned by  
	{\greedy}  (Line~\ref{ln:coverage} of Alg.~\ref{algo:thresholdgreedy}) for fixed threshold
	workload $\threshold$. Let $\opt_\threshold$ be the optimal assignment of experts $\experts_\threshold$ to tasks $\tasks$
	with respect to the coverage objective $\cov(\opt_\threshold)$. Then, it holds that:
	\[
	\cov\left(\assignment_\threshold\right)\geq \left(1-\frac{1}{e}\right)\cov\left(\opt_\threshold\right).
	\]
\end{lemma}

The proof of this lemma is similar to the proof that {\greedy} is an $\left(1-
\frac{1}{e}\right)$-approximation algorithm to the coverage problem~\citep{vazirani2013approximation} and is thus omitted.

The above lemma states that for every threshold $\threshold$ (i.e., for every iteration of  {\thresholdgreedy}), the {\greedy} subroutine is guaranteed to return a solution that has good coverage
with respect to the optimal solution for the coverage problem for this threshold $\threshold$.
The lemma does not state anything about the final solution returned by
{\thresholdgreedy}, or about the approximation with respect to the objective 
function~$\objective(\cdot)$. 
We build upon the lemma and state the following theorem.

\smallskip
\begin{theorem}\label{thm:maintheorem}
	Let $\assignment$ be the assignment returned by {\textrm{\thresholdgreedy}}
	and let $\opt$ be the optimal assignment for the {\bbalance} problem. Then we have the following approximation:
	$$\lambda \cov( \assignment)  - \maxload (\assignment) \geq \left(1-\frac{1}{e}\right) \lambda \cov( \opt)  - \maxload (\opt).$$
\end{theorem}

\begin{proof}
	Let us  assume that
	$\maxload(\opt)=\threshold^*$. 
	Note that $\maxload(\assignment)$ may or may not be equal to 
	$\threshold^*$. Then, we have the following: 

  \begin{center}
	\begin{align*}
 		\objective(\assignment) & \geq \objective(\assignment_{\threshold^*}) & \hfill\text{(True for any $\threshold$)}\\
		& = \lambda \cov(\assignment_{\threshold^*})-\threshold^*\\                                               
		&\geq \left(1-\frac{1}{e}\right) \lambda \cov(\opt_{\threshold^*})-\threshold^*  & \text{(Lemma~\ref{lm:fixedthreshold})}\\
		&\geq \left(1-\frac{1}{e}\right) \lambda \cov(\opt)-\threshold^* & \text{($\opt_{\threshold^*}$ is optimal for threshold $\threshold^*$)}\\
		& = \left(1-\frac{1}{e}\right)\lambda \cov(\opt)-\maxload(\opt).
	\end{align*}
 \end{center}
\end{proof}

\subsection{Running Time and Speedup}
\label{sec:speedups}

A naive implementation of {\thresholdgreedy} has running time $\bigO(m^2n^2)$. 
It requires~$m$ calls to the {\greedy} routine in Line~\ref{ln:coverage}, 
which if implemented naively,
takes time~$\bigO(mn^2)$.  Such a running time would make {\thresholdgreedy} impractical.  
Below, we discuss three methods that significantly improve the running time of our
algorithm and allow us to experiment with reasonably large datasets.

\spara{Lazy greedy instead of greedy:}
First, instead of using the naive implementation of {\greedy}, we deploy the lazy-evaluation
technique introduced by~\cite{minoux1978accelerated}. 
The lazy-evaluation technique utilizes a maximum priority queue to exploit the diminishing returns of the submodular function $\cov()$ to avoid re-evaluating candidate elements with low marginal gain, and performs very well in practice.
In our experiments, we only use this lazy-evaluation version of {\greedy}. 

\spara{Early termination of {\thresholdgreedy}:} A computational bottleneck for {\thresholdgreedy} is its outer loop
(line~\ref{ln:outer} in Algorithm~\ref{algo:thresholdgreedy}), which needs to be repeated 
$\numtasks$ times, where $\numtasks$ is the total number of tasks.   Here we show that not all 
$\numtasks$ values of $\threshold$ need to 
be considered. This is because the value of the objective function as computed by 
{\thresholdgreedy} for the different values of $\threshold$ is a unimodal function, which initially increases and then starts decreasing.  Therefore, once a  maximum is found for some value
of $\threshold$, the algorithm
can safely terminate as the value of the objective will not improve for larger values of $\threshold$.

If we denote by $\assignment_\threshold$
the assignment produced at the $\threshold$-th iteration of  {\thresholdgreedy} and by  $C_\threshold = \cov(\assignment_\threshold)$, then $\objective_\threshold = C_\threshold -\threshold$. Using this notation, we have the following theorem.

\smallskip
\begin{theorem}\label{thm:unimodality}
	If there is a value of the threshold $\threshold^\ast$, such that
	$F_{\threshold^\ast}\geq F_{\threshold^\ast-1}$ and
	$F_{\threshold^\ast}\geq F_{\threshold^\ast+1}$, then
	the values of the objective function $\objective_\threshold=\objective(\assignment_\threshold)$ as computed by {\thresholdgreedy} (line~\ref{ln:objective})
	for $\threshold = 1,\ldots , m$ are unimodal. That is, 
	$\objective_1\leq \objective_2\leq\ldots\leq \objective_{\threshold^\ast}$ and
	$\objective_{\threshold^\ast}\geq \objective_{\threshold^\ast+1}\geq \ldots\geq \objective_m$.
\end{theorem}
\smallskip

In order to prove Theorem~\ref{thm:unimodality}, we rely on the properties of {\thresholdgreedy} as well as on the fact 
that the coverage function $\cov(\cdot)$ is monotone and submodular (Proposition~\ref{prop:monotone-submodular}).  
Recall that $\assignment_\threshold$
is the assignment produced at the $\threshold$-th iteration of  {\thresholdgreedy} and
$C_\threshold = \cov(\assignment_\threshold)$.
Then, by definition $\objective_\threshold = C_\threshold -\threshold$.  
Moreover, the monotonicity and submodularity
of the coverage function imply the following:\footnote{$C_0=0$ 
since it is the coverage of the empty assignment.}

\smallskip
\begin{proposition}\label{prop:monotonicity}
	The monotonicity of the overall coverage function implies that for every $\threshold\in\{1,\ldots,\numtasks\}$: $C_\threshold\geq C_{\threshold-1}$.
\end{proposition}

\smallskip
\begin{proposition}\label{prop:submodularity}
	The submodularity of the overall coverage function implies that for every 
	$\threshold\in\{1,\ldots,\numtasks -1 \}$: $C_\threshold-C_{\threshold-1}\geq C_{\threshold+1}-C_\threshold$.
\end{proposition}
\smallskip

These propositions
rely on the fact that  in every iteration $\threshold$, {\thresholdgreedy} produces
assignment $\assignment_\threshold$, which has the property that $\assignment_\threshold\subseteq \assignment_{\threshold+1}$. That is, 
the $1$-entries in $\assignment_\threshold$ are a super\-set of the $1$-entries in 
$\assignment_{\threshold+1}$.

We are now ready to prove Theorem~\ref{thm:unimodality}.
\begin{proof}
	Let us assume that there is a threshold $\threshold^\ast$ such that $\objective_{\threshold^\ast}\geq \objective_{\threshold^\ast-1}$ and $\objective_{\threshold^\ast}\geq \objective_{\threshold^\ast+1}$.
	Since $\objective_{\threshold^\ast}\geq \objective_{\threshold^\ast-1}$ , we have
	\begin{align}\label{align:cond1}
		C_{\threshold^\ast} - \threshold^\ast & \geq  C_{\threshold^\ast-1} - (\threshold^\ast-1)\nonumber\\
		(C_{\threshold^\ast} - C_{\threshold^\ast-1}) & \geq 1. 
	\end{align}
	Using Inequality~\eqref{align:cond1} and Proposition~\ref{prop:submodularity}, we have 
	\[
	C_1-C_0\geq C_2-C_1\geq\ldots\geq C_{\threshold^\ast}-C_{\threshold^\ast-1}\geq 1.
	\]
	Thus, for every  $\threshold\leq \threshold^\ast$ it holds that
	\begin{align*}
		C_\threshold-C_{\threshold-1} & \geq 1\\
		C_\threshold -\threshold &\geq C_{\threshold-1}-(\threshold-1)\\
		F_\threshold & \geq F_{\threshold-1}.
	\end{align*}
	The proof is symmetric for the values of $\threshold>\threshold^\ast$. That is, 
	since $\objective_{\threshold^\ast}\geq \objective_{\threshold^\ast+1}$ , we have
	\begin{align}\label{align:cond2}
		C_{\threshold^\ast} - \threshold^\ast & \geq  C_{\threshold^\ast+1} - (\threshold^\ast+1)\nonumber\\
		(C_{\threshold^\ast+1} - C_{\threshold^\ast}) & \leq 1. 
	\end{align}
	Using Inequality~\eqref{align:cond2} and Proposition~\ref{prop:submodularity}, we have 
	\[
	C_\numtasks-C_{\numtasks-1}\leq C_{\numtasks-1}-C_{\numtasks-2}\leq\ldots\leq C_{\threshold^\ast+1}-C_{\threshold^\ast}\leq 1.
	\]
	Thus, for every  $\threshold> \threshold^\ast$ it holds that
	\begin{align*}
		C_{\threshold+1}-C_{\threshold} & \leq 1\\
		C_{\threshold+1} -(\threshold-1) &\leq C_{\threshold}-\threshold\\
		F_{\threshold+1} & \leq F_{\threshold}.
	\end{align*}
\end{proof}

We will call the value of $\tau$ for which $\objective(\cdot)$ gets maximized in the iterations of the {\thresholdgreedy} algorithm the \emph{best-greedy workload} 
and the corresponding value of the objective the \emph{best-greedy objective}.

\spara{Improving on linear search over workload values:} 
The unimodality of the objective function as computed by {\thresholdgreedy} for the different values of $\threshold$, 
clearly allows us to try all possible values of $\threshold$ starting from $1$ until the value of $\objective_\threshold$
stops increasing.  This is a \emph{linear search} over the 
different thresholds. We speedup this linear search by \emph{combining an exponential with a linear search}.
That is, we search over an exponentially increasing range of values of $\threshold = 2^i$, for $i \geq 0$; once the objective function decreases for some $i$, we then perform a linear search over the range of workload values, $\threshold \in [2^{i-1}, 2^{i}]$. In practice we observe that this technique significantly improves over the simple linear search.

Note that the unimodality of the objective function as computed by {\thresholdgreedy} for the different values of $\threshold$, would suggest a binary search over the values of $\threshold$. This type of search does not work well in practice because the running time of every iteration of {\thresholdgreedy} increases with the value of $\threshold$ and the binary search requires trying (at least some) large values of $\threshold$.  Thus in our experiments, we only use the combination of exponential and linear search we described above.

\subsection{Tuning Coverage vs.\ Workload Importance}
\label{sec:discussion}

One must choose an appropriate value of the balancing coefficient, $\lambda$ for each application, such that it tunes the relative importance of task coverage and expert workload as desired. 
In practice, we achieve this by 
examining different values of $\lambda$ and then picking the one that gives the most intuitive
trade-off between the coverage and the load of the corresponding solutions.  
There are two naive ways of implementing such a search process:   The first is to run  {\thresholdgreedy}
(with all the speedup ideas we proposed in Section~\ref{sec:speedups}) for the different values of 
$\lambda$. The second is to run {\thresholdgreedy} \emph{without} the early termination technique we discussed
in Section~\ref{sec:speedups} and for $\lambda=1$. This would mean that we would have to go over all possible values
of $\threshold$, and for each threshold $\threshold$ store independently the value of 
the coverage $C_\threshold$ for this threshold; then make a pass over all these values 
and weigh them appropriately with different $\lambda$s.  The first solution requires running
{\thresholdgreedy} as many times as the different $\lambda$s.  The second solution
requires running {\thresholdgreedy} once, but for \emph{all} possible values of threshold $\threshold=m$.  Both these solutions are infeasible in practice even for datasets of moderate size. However, we make a key observation in Proposition~\ref{prop:lambdas}, that enables us to efficiently search for an appropriate value for $\lambda$.

\smallskip
\begin{proposition}\label{prop:lambdas}
	Assume that $\lambda_1> \lambda_2$ and let 
	the best-greedy objectives achieved for those values be $\objective^{\lambda_1}_{\threshold_1}$ and 
	$\objective^{\lambda_2}_{\threshold_2}$, respectively.  
	Then, for the corresponding best-greedy workloads
	we have that $\threshold_1\geq \threshold_2$.
\end{proposition}

\begin{proof}
	Since $\lambda_1 > \lambda_2$, there exists an $\alpha>1$ such that 
	$\lambda_1 = \alpha \lambda_2$. 
	Our proof will be by contradiction: suppose that $\threshold_1<\threshold_2$.
	By Proposition~\ref{prop:monotonicity} we have that  $C_{\threshold_2} \geq C_{\threshold_1}$. 
	Since~$\threshold_2$ corresponds to the best-greedy workload for $\objective^{\lambda_2}$
	we have $\objective^{\lambda_2}_{\threshold_2}\geq\objective^{\lambda_2}_{\threshold_1}$
	and~thus:
	\begin{align*}
		\lambda_2 C_{\threshold_2} - \threshold_2 &\geq \lambda_2 C_{\threshold_1} - \threshold_1\\
		\lambda_2 (C_{\threshold_2} - C_{\threshold_1}) &\geq (\threshold_2-\threshold_1).
	\end{align*}
	Since $\threshold_1$ corresponds to the best-greedy workload for $\objective^{\lambda_1}$ we have  that $\objective^{\lambda_1}_{\threshold_2}\leq \objective^{\lambda_1}_{\threshold_1} $
	\begin{align*}
		\lambda_1 C_{\threshold_2} - \threshold_2 &\leq \lambda_1 C_{\threshold_1} - \threshold_1\\
		\lambda_1 (C_{\threshold_2} - C_{\threshold_1}) &\leq (\threshold_2-\threshold_1) \\
		\alpha \lambda_2 (C_{\threshold_2} - C_{\threshold_1}) &\leq(\threshold_2-\threshold_1)
	\end{align*}
	Combining these two results we get
	\begin{align*}
		\alpha \lambda_2 (C_{\threshold_2} - C_{\threshold_1}) &\leq(\threshold_2-\threshold_1) \leq \lambda_2 (C_{\threshold_2} - C_{\threshold_1}),
	\end{align*}
	which implies that $\alpha\leq 1$, which is a contradiction.
\end{proof}
\smallskip

\spara{An efficient search on the values of $\lambda$:} Using Proposition~\ref{prop:lambdas} we can explore the solutions of {\thresholdgreedy} for different values of $\lambda\in\Lambda\subseteq\mathbb{R}_{+}$ efficiently, by
running {\thresholdgreedy} only once and -- at the same time -- exploiting the early termination trick we discussed in Section~\ref{sec:speedups}.

We first run {\thresholdgreedy} with a large value of $\lambda$, and determine the \emph{best-greedy workload} and the corresponding value of the \emph{best-greedy objective}. We then compute the \textit{best-greedy} values for smaller values of $\lambda$, 
and plot the corresponding values of $\cov(\assignment)$ and $\maxload(\assignment)$ for each $\lambda$ value. 
Graphically, the best $\lambda$ value for each dataset corresponds to the $\lambda$ value observed at the \textit{elbow} of the plot, where further increase of $\lambda$ does not result in a significant increase in coverage. 
Thus, a suitable value of $\lambda$ can be identified by visual inspection, such that the best-greedy workload and best-greedy objective values yield a high value for the overall coverage, $\cov(\assignment)$ while simultaneously giving a reasonably low value for the  $\maxload(\assignment)$.
Note that the $\lambda$ value can be adjusted as needed, as per the requirements of the application domain. 

%% file: algos-nbalance.tex
In this section, we introduce {\networkbalance}, our algorithm  for solving the {\nbalance} problem. 
The pseudo-code of the algorithm is shown in Algorithm~\ref{algo:networkbalance}.
Conceptually, the algorithm is similar to {\thresholdgreedy}. 
More specifically, {\networkbalance} considers all values of load $\tau = 1,\ldots,\numtasks$. 
For each value~$\tau$, the algorithm forms candidate teams ({\formcandidateteams})
that satisfy the radius constraint and then it assigns teams to tasks ({\assignteams}). 
This assignment may cause some experts to violate the load  constraint imposed by $\tau$, 
thus, an additional pruning step ({\teampruning}) is needed to ensure that the load constraint is not violated. 
Finally, {\networkbalance} returns the assignment corresponding to the best objective found across the different workload values $\tau$.

In the rest of the section, we describe each one of the steps of {\networkbalance} in detail and discuss all computational issues that arise.

\begin{algorithm}[t]
    \caption{The {\networkbalance} algorithm.}\label{algo:networkbalance}
    \begin{algorithmic}[1]
        \Require  Set of $m$ tasks $\tasks$, $n$ experts $\experts$, graph $G = (\experts, \mathbf{E})$ with coordination costs, radius constraint $r$, and $\lambda$.
        \Ensure An assignment $\assignment$ of experts to tasks.
        \State $\assignment \leftarrow \emptyset, \objective_{\max} = 0$
        \State $\mathcal{T} \leftarrow \formcandidateteams(\experts, G, r)$ \label{ln:candidate-teams}
        \For{$\threshold = 1,...,m$}  \label{ln:forloop-threshold}
        \State $\assignment_\threshold \leftarrow \assignteams(\tasks, \mathcal{T}_{\threshold}, \threshold)$ \label{ln:assign-teams}
        \State $\assignment'_{\threshold} \leftarrow \teampruning(\assignment_{\threshold}, \threshold)$  \label{ln:prune-teams}
        \State  $\objective_{\threshold} \leftarrow \lambda C(\assignment'_{\threshold}) - \threshold$
        \If{$\objective_{\threshold} \geq \objective_{max}$}
        \State $\objective_{max} = \objective_{\threshold}$
        \State $\assignment \leftarrow \assignment'_{\threshold}$
        \EndIf
        \EndFor\\
        \Return $\assignment$
    \end{algorithmic}
\end{algorithm}

\subsection{Forming Candidate Teams}
First, we form a set of candidate teams $\mathcal{T}$ such that the each team in $\mathcal{T}$ has a radius that satisfies the specified radius constraint $r$;
this is done in 
Line~\ref{ln:candidate-teams} of Algorithm~\ref{algo:networkbalance}.
We pursue two alternatives for forming candidate teams, which we call 
{\formcandidateteamsR} and {\formcandidateteamsallR} and which we describe below.

\spara{{\formcandidateteamsR}}:
Given a set of $\numexperts$ experts $\experts$, a graph $G = (\experts, E)$ with their coordination costs, 
and a radius constraint $r$, $\formcandidateteamsR(\experts, G, r)$ forms $\numexperts$ teams, 
one team $\team_i$ for each expert $\expert_i$.
Team $T_i$ consists of expert $\expert_i$ 
and all other experts $\expert_j$ with $d(\expert_i,\expert_j) \leq r$. 
That is, $\team_i = \expert_i\cup\{\expert_j\mid d(\expert_i,\expert_j) \leq r\}$.
This method runs in time $\bigO(\numexperts^2)$ and creates $\numexperts$ candidate teams.

\spara{{\formcandidateteamsallR}}: Here, 
we consider several different radii $0 < r' \leq r$; for each~$r'$ we invoke {\formcandidateteamsR} and form $\numexperts$ teams corresponding to radius constraint~$r'$.
In practice, we form teams of varying sizes by splitting the interval $(0,r]$ into $k$ parts of size $r/k$, and choosing $k$ different values for $r' \in \{r/k, 2r/k,\ldots , r\}$. {\formcandidateteamsallR} returns $k\numexperts$ candidate teams, and its running time is $\bigO(k \numexperts^2)$.

\subsection{Assigning Teams to Tasks}
\label{sec:assignteamstotasks}
Before we describe our 
general algorithm for assigning teams to tasks, we consider a special case, where every team 
consists of one expert and the task is to assign experts to tasks.
In this case, the team-assignment problem can be written as 
a linear program as follows: 
let $x_{ij}=1$ if expert $i$ is assigned to task $j$, 
and let $C_{ij}$ denote the fraction of skills required by task $\task_j$ covered by expert $i$.
The linear program (LP) is the following:
\begin{align*}
	\text{maximize}~~  & \sum_{i=1}^n \sum_{j=1}^m C_{ij}x_{ij},\\
	\text{such that}~~ &\sum_{i=1}^n x_{ij} \leq 1, \quad \text{ for all }  1 \leq j \leq m ,\\
	&\sum_{j=1}^m x_{ij} \leq \threshold \quad \text{ for all } 1 \leq i \leq n , \text{ and}\\
    & 0\leq x_{ij}\leq 1.
\end{align*}
Note that due to the \emph{unimodular} nature of the constraints the above LP only has integer solutions, 
i.e., in the optimal solution it is $x_{ij}\in\{0,1\}$, for all $(i,j)$~\citep{papadimitriou1998combinatorial}.

Therefore, when teams consist of one expert, the team-assignment problem can be solved optimally in polynomial time.  
Additionally, the above LP works in cases when there is a pre-specified set of teams
$\mathcal{T} = \{\team_1,\ldots , \team_\ell \}$. 
The solution obtained by the LP in this case guarantees that 
each task is assigned to at most one team, 
and each team is assigned to at most $\threshold$ tasks. 
However, since the teams may have arbitrary overlap among their experts, 
there is no guarantee for the number of tasks assigned to a single expert. 
We consider the solution of the above LP for teams, even if it
violates the per-expert load constraint. 
To ensure compatibility with the load constraints, 
we then prune the teams so that each expert has load at most $\threshold$ (see next section).

In practice, we solve the {\assignteams} task shown in Algorithm~\ref{algo:networkbalance}
either by solving the LP we described above using a readily-available solver like Gurobi~\citep{gurobi}, 
or by a greedy algorithm that greedily matches a team to a task that maximizes the objective and does not violate any of the constraints. 
Such a greedy assignment is a $2$-approximation algorithm to the problem described by the LP~\citep{khan2016efficient}
and it runs in time~$\bigO(\numtasks^2)$.
We note though that the Gurobi solver works extremely well in practice.

Clearly, given an assignment of teams to tasks, we can generate a corresponding assignment $\assignment$ of experts to tasks as follows: for each task a team is assigned to, all experts on that team have a 1-entry in the corresponding column in $\assignment$. 

\subsection{Pruning Teams}
As the assignment $\assignment$ returned by {\assignteams} 
may violate the load constraint $\tau$ for individual experts, 
we  prune the assignment by removing experts from teams in order to guarantee that the load of each individual is $\tau$ or less. 
For this, 
we invoke the following {\teampruning} step in Line~\ref{ln:prune-teams} of Algorithm~\ref{algo:networkbalance}.

The pseudocode for the {\teampruning} routine is presented in Algorithm~\ref{algo:teampruning}. 
The pruning algorithm takes as input an assignment $\assignment$ of experts to tasks, and the load constraint $\threshold$. It then removes (or un-assigns) experts from tasks until all experts satisfy the workload constraint $\threshold$.

In order to explain {\teampruning,}
we introduce the idea of \emph{coverage loss}, which we define to be the amount of coverage of a task that is lost when an expert is removed from the team assigned to that task.
First, we obtain the set of all overloaded experts that need to be pruned. Then for each task that the expert is assigned to, we compute the loss in coverage  by removing the expert from that team. We add these coverage-loss values to a priority queue. Subsequently, we prune experts from tasks in order of increasing coverage loss from the priority queue, until all experts satisfy the workload constraint, $\tau$. 
Every time we remove an expert from a task, we recompute the coverage losses of all other experts that were assigned to that task.

The worst-case running time of  {\teampruning} is $\bigO(\numexperts^2 m)$; in practice, this is significantly faster as it is not usually necessary to prune the entire priority queue.

\begin{algorithm}[t]
    \caption{The {\teampruning} algorithm}\label{algo:teampruning}
    \hspace*{\algorithmicindent} \textbf{Input:} Assignment $\assignment$ and workload constraint $\threshold$ \\
    \hspace*{\algorithmicindent} \textbf{Output:} Pruned Assignment $A'$
    \begin{algorithmic}[1]
        \State $\assignment' \leftarrow \assignment$
        \State $\experts_{\threshold} \leftarrow$ Set of experts in $\assignment$ with workload greater than $\threshold$ 
        \State Initialize a priority queue to store coverage losses for expert-task pairs
        \For{each expert $\expert$ in $\experts_{\threshold}$}
        \For{each team $\team$ expert $\expert$ is on}
        \For{each task $\task$ team $\team$ is assigned to}
        \State Compute loss in coverage of task $\task$ by removing expert $\expert$ from team $\team$.
        \State Insert expert-task coverage loss into priority queue.
        \EndFor
        \EndFor
        \EndFor
        \While{Any expert $\expert$ in $\experts_{\threshold}$ violates workload constraint $\threshold$}
        \State  $\assignment' \leftarrow$ Prune expert-task pair from $\assignment$ using priority queue. 
        \State Recompute coverage losses of experts on pruned team.
        \EndWhile\\
        \Return{$\assignment'$}
    \end{algorithmic}
\end{algorithm}

\subsection{Approximation}
Although {\networkbalance}  performs well in practice, we have no formal approximation guarantees for its performance.  Part of the reason for this is that the sub\-problem of assigning a set of pre-formed teams (i.e., the ones formed by  {\formcandidateteams}) to tasks such that the coverage is maximized, while the load of each individual expert is below a threshold $\tau$ is an \NP-hard problem itself. We prove this in  Appendix~\ref{sec:appendix-proof}. 

This observation does not mean that {\nbalance} cannot be approximated; it simply means that {\networkbalance} as it is designed in Algorithm~\ref{algo:networkbalance} cannot have provable approximation bounds.

\subsection{Running Time and Speedups}\label{sec:speedups2}
In this section, we discuss the running time of {\networkbalance} and propose some practical speedups. Note that a naive implementation of the {\networkbalance} algorithm would have a running time $\bigO(m^2 n^2)$. Since the {\networkbalance} algorithm computes the same objective as {\thresholdgreedy}, we can exploit some of the speedup techniques from Sec.~\ref{sec:speedups}.

\spara{Early Termination of {\networkbalance}:} We make use of Theorem~\ref{thm:unimodality}, and do not consider all $\numtasks$ values of $\threshold$.
The value of the objective function $\objective_{\threshold}$ as computed by 
{\networkbalance} for the different values of $\threshold$ is a unimodal function, and once a maximum is found for some value of $\threshold$, the algorithm can safely terminate as the value of the objective will not improve for larger values of $\threshold$.

\spara{Improving on Linear search over workload values:} As in {\thresholdgreedy}, in Line~\ref{ln:forloop-threshold} of Algorithm~\ref{algo:networkbalance} we search over an exponentially increasing range of values of $\threshold = 2^i$, for $i \geq 0$; once the objective function decreases for some $i$, we then perform a linear search in the range  $\threshold \in [2^{i-1}, 2^{i}]$. In practice, this technique significantly improves the performance of the method, over the simple linear search.

\subsection{Instantiating the {\networkbalance} Algorithm}
\label{sec:nbalance-algo-names}
We specify here the naming convention we use for different variants of the {\networkbalance} algorithm, depending on how we choose to implement the  subroutines: {\formcandidateteams} (i.e., {\formcandidateteams}\texttt{-R} or {\formcandidateteams}\texttt{-All}) and {\assignteams} (i.e., {\assignteams}\texttt{-LP} or {\assignteams}\texttt{-Greedy}), we call the corresponding versions of {\networkbalance}: {\networkbalance}\texttt{-R-LP},
{\networkbalance}\texttt{-R-Greedy},
{\networkbalance}\texttt{-All-LP} and 
{\networkbalance}\texttt{-All-Greedy} respectively;
{\teampruning} is always invoked.

\subsection{Tuning Coverage vs.\ Workload Importance}
\label{sec:nbalance-tuning}

Similar to the technique used for the {\thresholdgreedy} algorithm in Section~\ref{sec:discussion}, depending on the application, we choose an appropriate value of the balancing coefficient, $\lambda$ such that it balances the relative importance of task coverage and expert workload. 
We call the value of $\tau$ for which $\objective(\cdot)$ gets maximized in the iterations of the {\networkbalance} algorithm the \emph{best-network workload} 
and the corresponding value of the objective the \emph{best-network objective}.
We can then make use of Proposition~\ref{prop:lambdas}, but modified with the best-network workload and best-network objective, and follow the technique in Section~\ref{sec:discussion} to graphically select an appropriate $\lambda$ value that gives the most desirable
trade-off between the coverage and the workload.

%% file: experiments.tex
We experimentally evaluate 
our algorithms for both {\bbalance} and  {\nbalance} 
using real-world datasets. 
We compare our algorithms with other heuristics, inspired by related work.
In the end of the section, we also compare 
the solutions obtained by {\thresholdgreedy} and {\networkbalance}, aiming to provide additional 
insight on the differences and the similarities of the two methods.

Our implementation is in Python and available online.\footnote{\url{https://github.com/kvombatkere/Team-Formation-Code}}
For all our experiments
we use single-process implementation on a 64-bit MacBookPro with an AppleM1Pro CPU and 16GB RAM. 

\subsection{Experiments for {\bbalance}}
\label{sec:exp}
In this section we first introduce our datasets and baselines, and then discuss our experiments for the {\bbalance} problem. We show how we choose the balancing coefficient $\lambda$ for each dataset, and then evaluate the performance of {\thresholdgreedy} and baselines in terms of the objective, expert load and running time.

\subsubsection{Datasets}
\label{sec:datasets1}
We evaluate our methods on several real-world datasets; 
some of these datasets have been used in past team-formation papers~\citep{anagnostopoulos10power,nikolakaki20finding,nikolakaki21efficient}.
A short description of the datasets follows, while 
their statistics are shown in Table~\ref{tab:summarystats}.

\smallskip
\noindent
\emph{IMDB:} The data is obtained from the International Movie Database.\footnote{https://www.imdb.com/interfaces/}
We simulate a team-formation setting where movie directors conduct auditions for movie actors: we assume that movie genres correspond to skills, movie directors to experts, and actors to tasks. 
The set of skills possessed by a director or actor is the union of genres of the movies they have participated in.
In order to experiment with datasets of different sizes, we create 
three data instances by selecting all movies created since 2020, 2018 and 2015. 
From these movies we select the directors that have at least one actor in common with at least one other director, 
and then randomly sample 1000, 3000 and 4000 directors, to form the set of experts in the 3 datasets. 
Then we randomly sample 4000, 10000 and 12000 actors, to form the set of tasks.
We refer to these datasets as {\imdbone}, {\imdbtwo} and {\imdbthree}, respectively. 

\smallskip
\noindent
\emph{Bibsonomy:} 
This dataset comes from a social bookmark and publication sharing system with a large number of publications, each of which is written by a set of authors~\citep{benz2010social}. 
Each publication is associated with a set of \emph{tags}; we filter tags for stopwords and use the 1000 most common tags as skills.
We simulate a setting where certain prolific authors (experts) conduct interviews for other less prolific authors (tasks). 
An author's skills are the union of the tags associated with their publications. 
Upon inspection of the distribution of skills among all authors we determine prolific authors to be those with at least 12 skills. 
We create three datasets by selecting all publications since 2020, 2015 and 2010. 
From these publications we select the prolific authors that have at least one paper in common with at least one other prolific author, and then randomly sample 500, 1500 and 2500 prolific authors to form the set of experts in the 3 datasets. Then we randomly sample 1000, 5000 and 9000 non-prolific authors, to form the set of tasks.
We refer to these datasets as {\bibsonomyone}, {\bibsonomytwo}, {\bibsonomythree}, respectively.

\smallskip
\noindent
\emph{Freelancer} and \emph{Guru:} These two datasets consist of random samples of real jobs that are posted by users online, and a random sample of real freelancers, in the \textit{Freelancer}\footnote{\url{freelancer.com}} and \textit{Guru} \footnote{\url{guru.com}} online labor marketplaces respectively.
The data consists of tasks that require certain discrete skills, and experts who possess discrete skills.
The Freelancer data we use consists of 993 jobs (i.e. tasks) that require skills and 1212 freelancers (i.e. experts) that have skills; we refer to this dataset as~{\freelancer}. Similarly, the Guru data we use consists of 3195 tasks that require skills and 6120 experts that have certain skills; we refer to this dataset as {\guru}.

\begin{table}
    \caption{Summary statistics of our datasets.}
    \label{tab:summarystats}
    {\small
        \begin{tabular}{lrrrrr}
            \toprule
            Dataset&Experts&Tasks&Skills&skills/ &skills/\\
            &             &         &        & expert&task\\
            \midrule
            {\imdbone}&1000 &4000 &24 &2.2 &2.0\\
            {\imdbtwo}&3000 &10000 &25 &2.4 &2.2\\
            {\imdbthree}&4000 &12000 &26 &2.8 &2.8\\
            {\bibsonomyone}&500 &1000 &957 &13.0 &4.8\\
            {\bibsonomytwo}&1500 &5000 &997 &13.6 &4.9\\
            {\bibsonomythree}&2500 &9000 &997 &13.6 &4.9\\
            {\freelancer}&1212&993&175&1.5&2.9\\
            {\guru}&6120&3195&1639&13.1&5.2\\
            \bottomrule
        \end{tabular}
    }
\end{table}

\subsubsection{Baselines}
Motivated by existing work, we use the following three algorithms as baselines:

\spara{{\lpsetcover}:} This algorithm is an application of the offline Linear Programming rounding (LP-rounding) algorithm discussed by \citet{anagnostopoulos10power}.  
Using their LP formulation, the goal is to obtain a fractional assignment of experts to tasks
such that every task is fully covered and the maximum load is minimized. 
Once a fractional assignment  is obtained 
(let $X_{ij}$ be the fractional assignment of expert $i$ to task $j$), 
a rounding scheme is provided that operates in logarithmic number of rounds; 
in each round we independently assign expert $i$ to task $j$ with probability $X_{ij}$. 
It can be shown that at the end of rounding each
task is fully covered with high probability and the load achieved is a logarithmic approximation to the optimal load.  
In our case, we proceed with the same LP, but in every iteration
of the rounding phase, we check the value of our objective and we only keep the solution that has the best value. 
Our LP has $mn$ variables and $\bigO(mn)$ constraints. 
If $T$ is the running time for the LP then the overall running time of {\lpsetcover} is $\bigO(T+mn)$. 
For our experiments we use Gurobi~\citep{gurobi} and we observe that {\lpsetcover} 
is significantly slower than the other baselines.


\spara{{\taskgreedy}:} 
This algorithm is inspired by the previous work of \citet{nikolakaki20finding}. 
{\taskgreedy} iterates over all tasks sequentially and 
for each task it greedily assigns experts to maximize the task's coverage. 
To balance the maximum workload with the total task coverage successfully, 
we implement two heuristics. 
First we randomize the order in which experts are greedily assigned to tasks in each iteration. 
This ensures an even distribution of experts in a setting in which several experts might be equivalently good for a task. Second, we only assign experts if they yield a significant increase in the task coverage. We quantify this coverage amount by a hyperparameter, $\beta$, which we specifically grid search and optimize for each dataset. Excluding the grid search, the {\taskgreedy} algorithm has a running time of $\bigO(\numtasks\numexperts)$ since there are $\numexperts$ experts available for each of the $\numtasks$ tasks.

\spara{{\noupdategreedy}:} 
This algorithm is a simple modification  of {\thresholdgreedy}: 
for each expert--task pair $(i,j)$, we initialize the keys in the priority queue to 
$v(i,j) = \tilde{\cov}((i,j)\mid \assignment^0)$, 
where $\assignment^0$ is the assignment with all entries equal to $0$. 
We then use these initial marginal-gain values to iteratively add expert-task edges $(i,j)$ 
in decreasing order of their $v(i,j)$ values, without ever updating them.
In order to improve the performance
of {\noupdategreedy}, we only use an expert  if $v(i,j) > \beta$, where $\beta$ is a hyperparameter.
{\noupdategreedy} has a running time of $\bigO(\numtasks\numexperts \log(\numtasks\numexperts))$, since there are $\numtasks \numexperts$ total expert-task edges, and sorting these edges takes time $\bigO(\log(\numtasks\numexperts))$.

\smallskip
In all cases, we perform a grid search over the values of all hyperparameters and
we report the best results for each algorithm and each dataset.

\begin{figure*}
	\centering
	\begin{subfigure}[b]{0.24\textwidth}
		\centering
		\includegraphics[width=\textwidth]{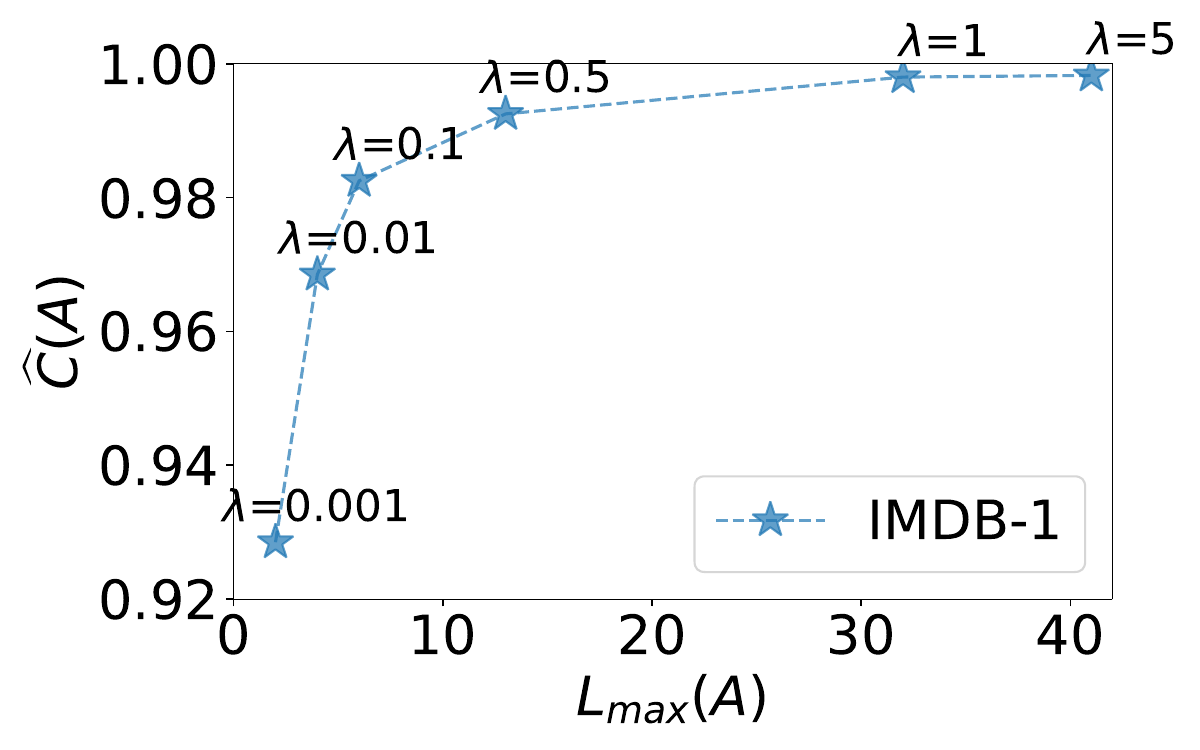}
	\end{subfigure}
	\hfill
	\begin{subfigure}[b]{0.24\textwidth}
		\centering
		\includegraphics[width=\textwidth]{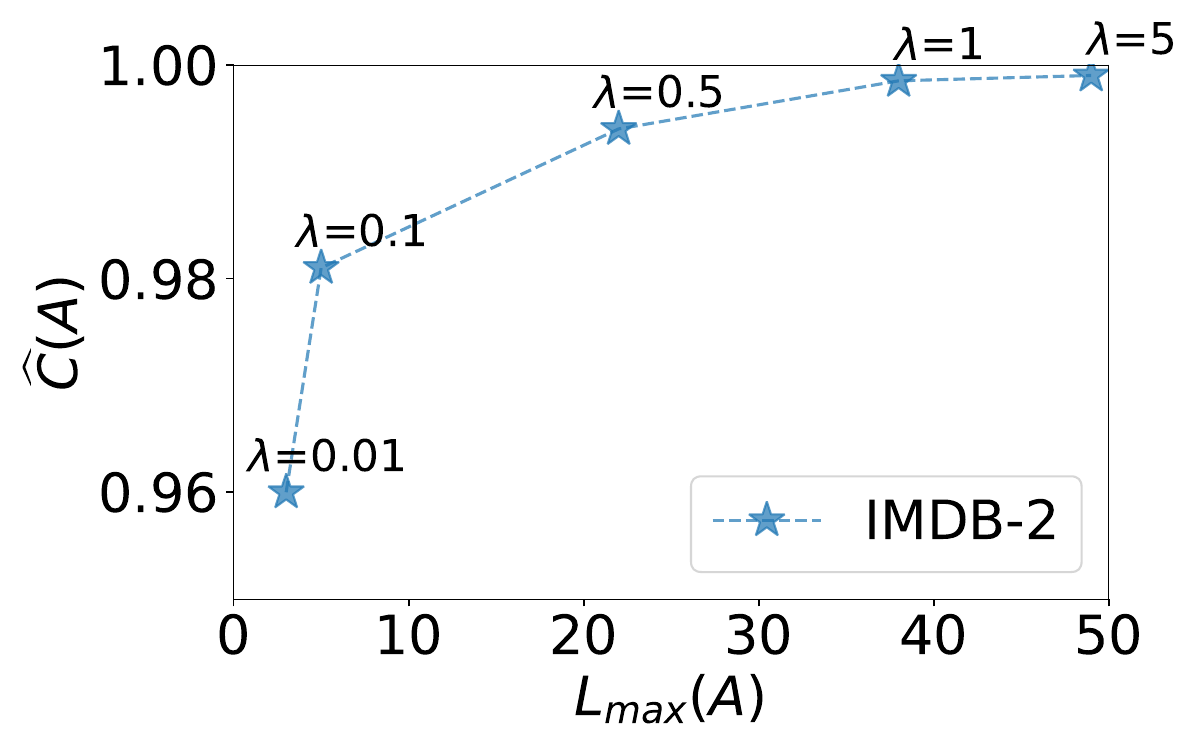}
	\end{subfigure}
	\hfill
	\begin{subfigure}[b]{0.24\textwidth}
		\centering
		\includegraphics[width=\textwidth]{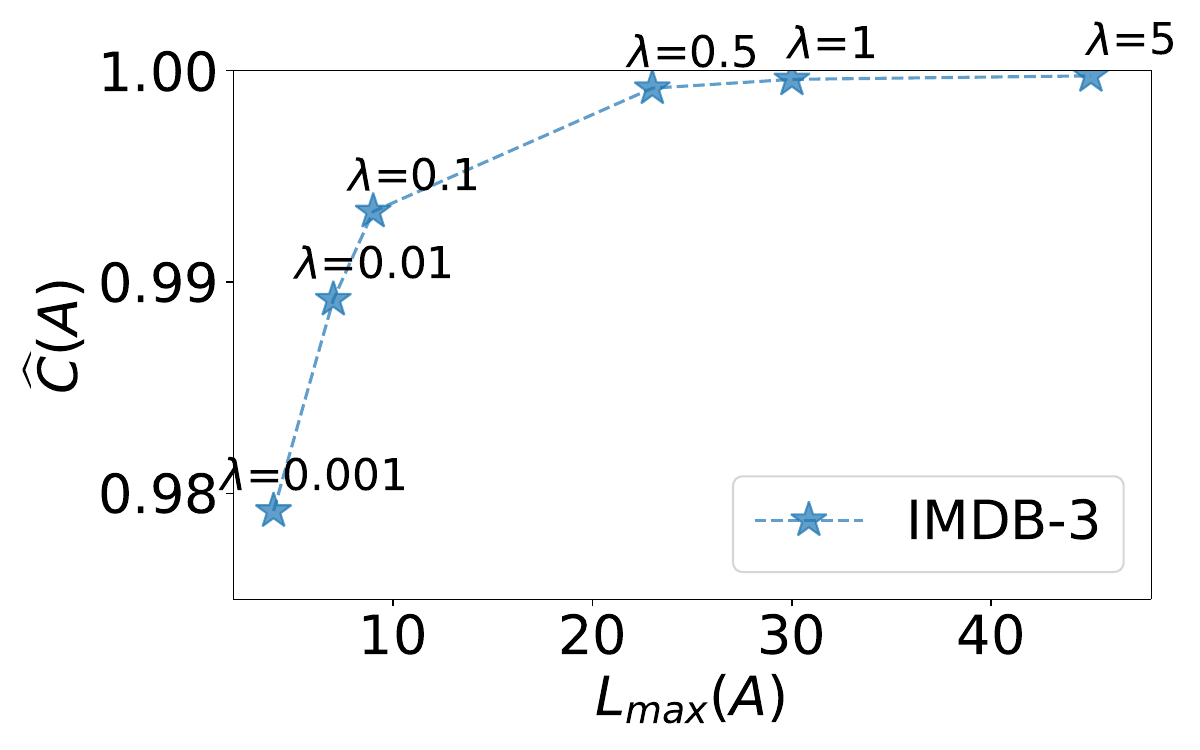}
	\end{subfigure}
	\hfill
	\begin{subfigure}[b]{0.24\textwidth}
		\centering
		\includegraphics[width=\textwidth]{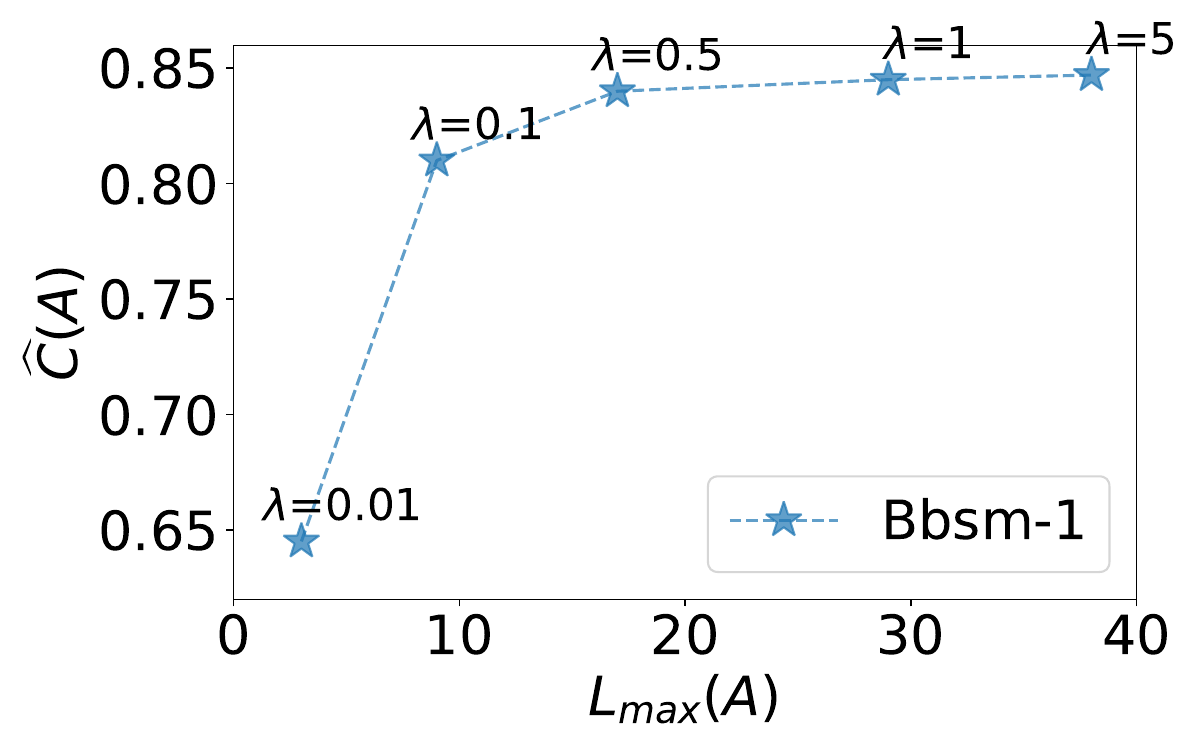}
	\end{subfigure}
	\begin{subfigure}[b]{0.24\textwidth}
		\centering
		\includegraphics[width=\textwidth]{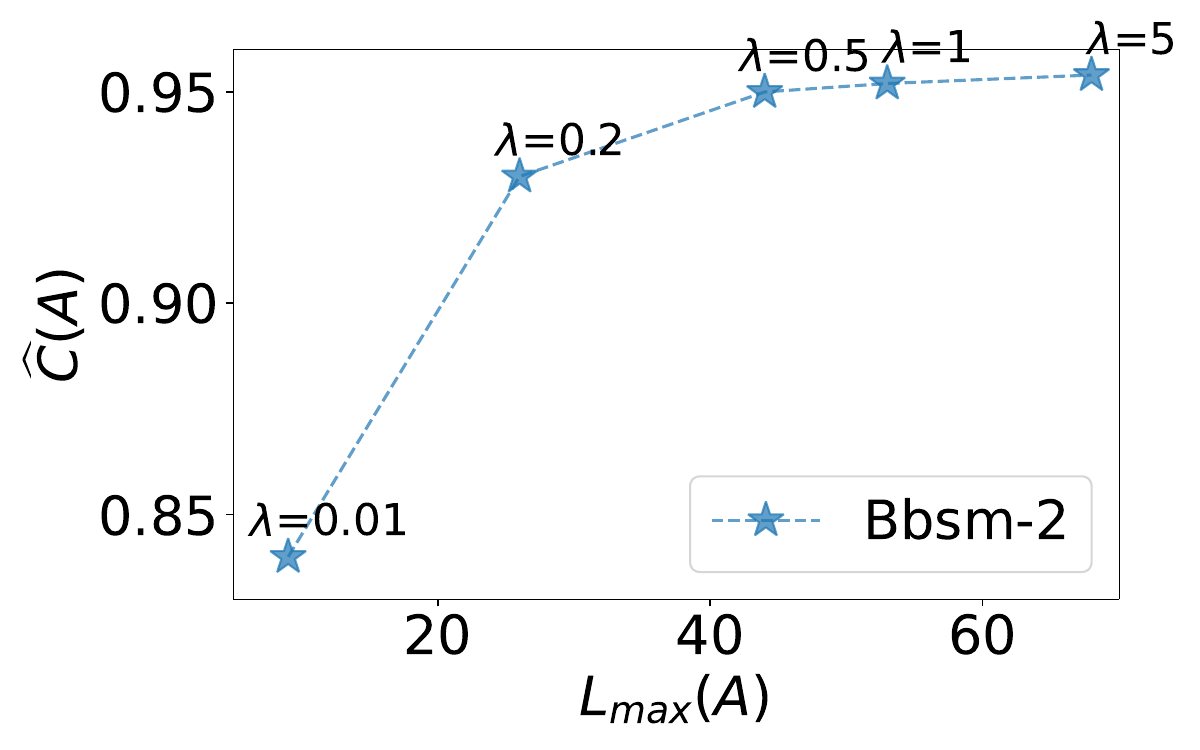}
	\end{subfigure}
	\hfill
	\begin{subfigure}[b]{0.24\textwidth}
		\centering
		\includegraphics[width=\textwidth]{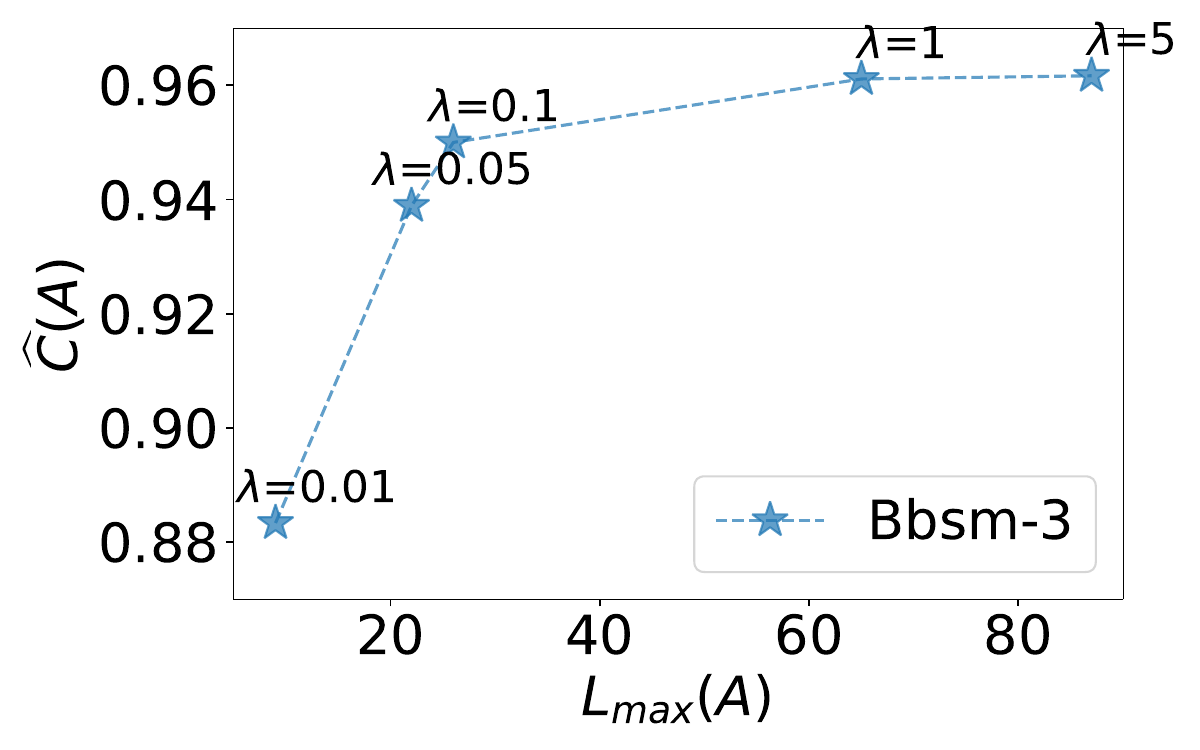}
	\end{subfigure}
	\hfill
	\begin{subfigure}[b]{0.24\textwidth}
		\centering
		\includegraphics[width=\textwidth]{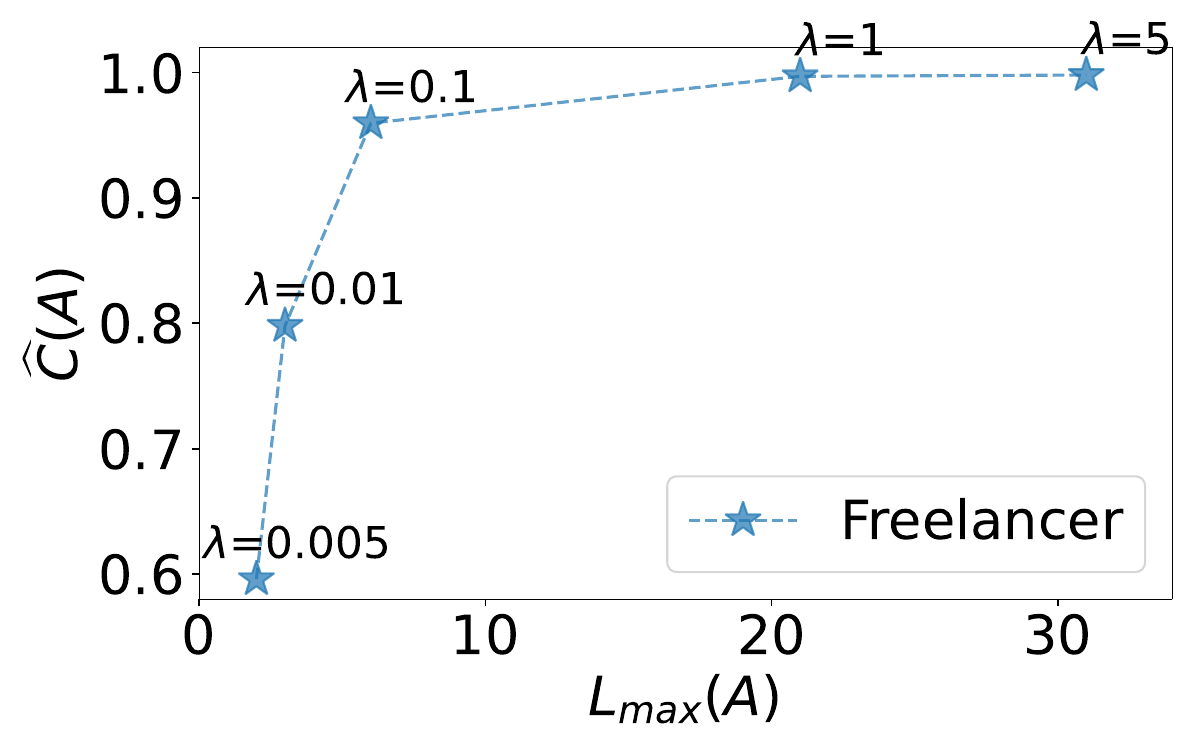}
	\end{subfigure}
	\hfill
	\begin{subfigure}[b]{0.24\textwidth}
		\centering
		\includegraphics[width=\textwidth]{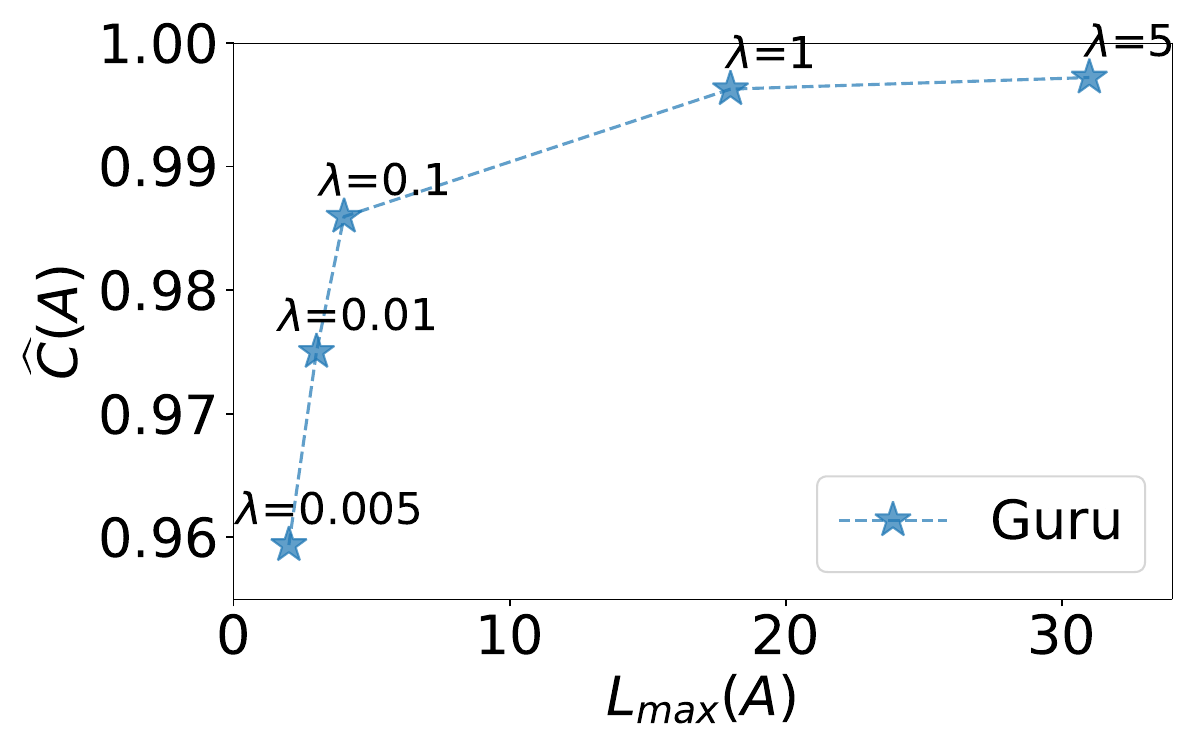}
	\end{subfigure}
	\caption{\label{fig:lambdaSearch}The {best-greedy} workload $\maxload(\assignment)$ value and 
	the coverage $\cov(\assignment)$ corresponding to 
	the best-greedy objective $\objective^{\lambda}(\assignment)$ computed by {\thresholdgreedy}. 
	Each subplot shows a range of values of the balancing coefficient $\lambda$ for each dataset.
	}
\end{figure*}

\subsubsection{Tuning Coverage and Workload Importance}
\label{sec:tuning}

Before showing our experimental results, we discuss how we set the balancing coefficient $\lambda$, 
following the techniques described in Section~\ref{sec:discussion}.
We first run  {\thresholdgreedy}  with a large value of $\lambda$, 
and determine the \emph{best-greedy workload} and the corresponding value of the \emph{best-greedy objective}. 
We then compute the \textit{best-greedy} values for smaller values of $\lambda$, 
and plot the corresponding values of $\cov(\assignment)$ and $\maxload(\assignment)$ for each $\lambda$ value. 
Fig.~\ref{fig:lambdaSearch} shows these scatter plots for each dataset. 
In most of our datasets we experimented with relatively small values of $\lambda\in(0,5]$.
We then visually inspect these plots to identify a suitable value of $\lambda$ such that the best-greedy workload and best-greedy objective values yield a high value for the overall coverage, $\cov(\assignment)$ while simultaneously giving a reasonably low value for the  $\maxload(\assignment)$. 
The values of $\lambda$ we picked for the different datasets are shown besides the 
dataset name in Table~\ref{tab:bbalance-algo-results}.

\begin{table*}
	\caption{Experimental performance of {\thresholdgreedy} and baseline algorithms in terms of the objective $\objective^{\lambda}$, the maximum load $\maxload$ and the average task coverage $\widehat{\cov}= \frac{1}{\numtasks}\cov$. The best values for each dataset are in bold.}
	\label{tab:bbalance-algo-results}
	\tiny
	\begin{tabular}{ l  rrr  rrr  rrr rrr }
		\toprule
		\multirow{2}{*}{Dataset $(\lambda)$} &
		\multicolumn{3}{c}{\thresholdgreedy} &
		\multicolumn{3}{c}{\lpsetcover} &
		\multicolumn{3}{c}{\taskgreedy} &
		\multicolumn{3}{c}{\noupdategreedy}  \\
		\cmidrule(lr){2-4} \cmidrule(lr){5-7} \cmidrule(lr){8-10} \cmidrule(lr){11-13}
		& $\objective^{\lambda}$& $\maxload$ & $\widehat{\cov}$ & $\objective^{\lambda}$& $\maxload$ & $\widehat{\cov}$ & $\objective^{\lambda}$& $\maxload$ & $\widehat{\cov}$ & $\objective^{\lambda}$& $\maxload$ & $\widehat{\cov}$  \\
		\midrule
		{\imdbone} $(0.1)$ & \textbf{388} & \textbf{6} & \textbf{0.98} & 295 & 72 & 0.92 & 318 & 45 & 0.91 & 191 & 150 & 0.85 \\
		{\imdbtwo} $(0.1)$ & \textbf{972} & \textbf{5} & \textbf{0.98} & 845 & 123 & 0.97 & 852 & 95 & 0.94 & 636 & 298 & 0.94 \\
		{\imdbthree} $(0.1)$& \textbf{1184} & \textbf{9} & \textbf{0.99} & 1099 & 89 & 0.99 & 922 & 222 & 0.95 & 957 & 200 & 0.96 \\
		
		{\bibsonomyone} $(0.1)$& \textbf{72} & \textbf{9} & \textbf{0.81} & 65 & 16 & 0.81 & 23 & 12 & 0.31 & 24 & 18 & 0.3 \\
		{\bibsonomytwo} $(0.2)$& \textbf{900} & \textbf{29} & \textbf{0.93} & 848 & 65 & 0.91 & 350 & 33 & 0.39 & 330 & 67 & 0.4 \\
		{\bibsonomythree} $(0.1)$& \textbf{827} & \textbf{27} & \textbf{0.95} & 723 & 97 & 0.91& 330 & 91 & 0.47 & 323 & 109 & 0.48 \\
		
		{\freelancer} $(0.1)$	& \textbf{88} & \textbf{6} & 0.95 & 59 & 32 & 0.92 & 63 & 36 & \textbf{0.99} & 25 & 50 & 0.76 \\
		{\guru} $(0.1)$ & \textbf{311} & \textbf{4} & \textbf{0.99} & 287 & 25 & 0.98 & 225 & 30 & 0.80 & 17 & 33 & 0.16 \\
		\bottomrule
	\end{tabular}
\end{table*}

\begin{figure*}
	\centering
	\includegraphics[width=\textwidth]{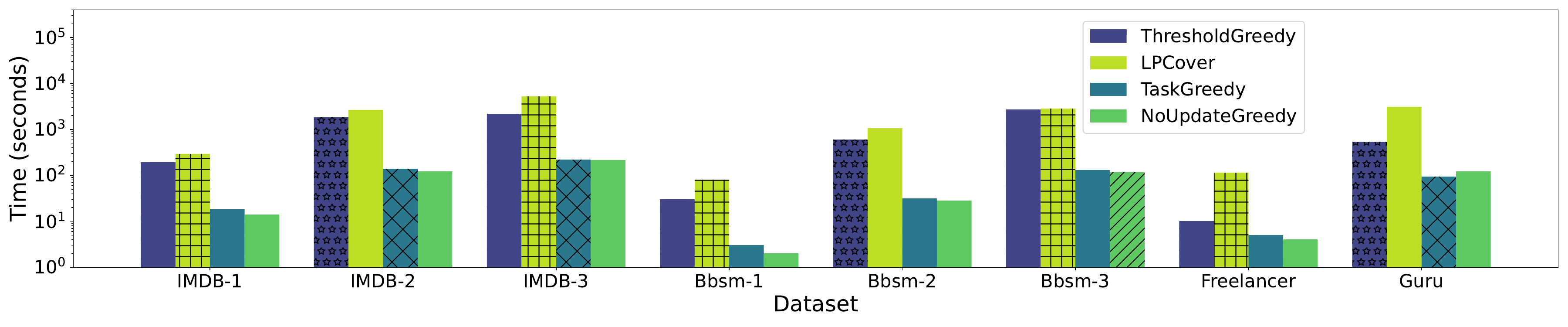}
	\caption{Running time (in seconds) of {\thresholdgreedy} and baseline algorithms, in logarithmic scale.}
	\label{fig:runtimes}
\end{figure*}

\subsubsection{Evaluation}
\label{sec:bbalance-eval}

We show the comparative performance of all four algorithms, 
in terms of the objective function ($\objective^{\lambda}$), 
the average coverage $\widehat{\cov} = \frac{1}{\numtasks}\cov$, and the maximum load $\maxload$, 
in Table~\ref{tab:bbalance-algo-results}.
Intuitively, a \textit{good} solution to an instance of the {\bbalance} problem is an assignment $\assignment$ that not only maximizes the overall task coverage but also minimizes the maximum load of the assignment. 
Our experiments for {\thresholdgreedy} show that it  performs the best, compared to all our baselines, in terms of the objective across all datasets. Additionally, it finds assignments with a low maximum workload and it runs in a reasonable amount of time, even for datasets with several thousand experts and tasks. 
Note that for different datasets we use different values of $\lambda$;
however, {\thresholdgreedy} finds the highest overall task coverage \textit{independently} of the value of $\lambda$, and consequently would also outperform the baselines for other $\lambda$ values as well.

\spara{Objective values $\objective$ and workload $\maxload$:}
As we can observe in Table~\ref{tab:bbalance-algo-results}, 
{\thresholdgreedy}  consistently finds the assignment with the best objective value. 
On average, across all datasets {\thresholdgreedy} performs about 15\% better than {\lpsetcover} 
and 55\% better than {\taskgreedy} and {\noupdategreedy}. 
As the datasets get larger, the superior performance of {\thresholdgreedy} becomes more evident. 
This behavior may be attributed to our algorithm finding solutions with significantly lower $\maxload$.

{\lpsetcover} is consistently the second-best algorithm in terms of the objective function. 
It also performs particularly well on the {\imdbtwo}, {\imdbthree} and {\guru} datasets --- 
it returns objective values that are comparable (but lower) to those returned by {\thresholdgreedy}. 
{\taskgreedy} and {\noupdategreedy} perform relatively well on the {\imdb} and {\freelancer} datasets --- 
they return objective values that are within 20\% of the objective value of {\thresholdgreedy}. 
In general, we observe that these baselines perform reasonably well on smaller datasets:
one explanation is that 
the pool of suitable experts available to {\taskgreedy} is small and 
the initial marginal-gain values used by {\noupdategreedy} are good estimators 
of the true marginal-gain values in subsequent iterations. 
However, while the baselines often achieve an overall task coverage of~90\%,  
{\thresholdgreedy} achieves superior task coverage in the majority of the cases.

In terms of maximum workload,  {\thresholdgreedy} consistently finds the assignment with the lowest maximum workload value across all our experiments;  the baselines return maximum load values that are significantly larger than those returned by {\thresholdgreedy}. On average across all datasets {\thresholdgreedy} finds a maximum load value that is 80\% smaller than the maximum workload values returned by the baselines.
This is because, in an attempt to maximize the overall task coverage, the baselines  make costly assignments of experts to tasks.  While we do see some examples of reasonable workload values (e.g., for the {\guru} dataset), in most cases the workload values returned by the baselines would be infeasible in practice.

\spara{Running time:}
While  {\thresholdgreedy}  has a theoretical running time of $\bigO(m^2n^2)$, the speed-up techniques discussed in Section~\ref{sec:speedups} and Section~\ref{sec:discussion} lead to significantly lower running time in practice. Fig.~\ref{fig:runtimes} shows a bar plot with the running time of all algorithms for each dataset in logscale.
For the smaller datasets (e.g., {\freelancer} and {\bibsonomytwo}), we observe 
that the running time of {\thresholdgreedy} is on
the order of a few seconds. Even for the largest datasets (e.g., {\bibsonomythree} and  {\imdbthree}) the running time of our algorithm is within a few hours. 
We also observe that {\taskgreedy} and {\noupdategreedy} are faster than our algorithm, but {\lpsetcover} is slower,  
due to the computational bottleneck of solving an LP with a large number of variables. 
Note that the running time of the baselines as we report them here do not include 
the grid search we performed in order to tune their hyperparameters.

\subsection{Experiments for {\nbalance}}
\label{sec:exp-nb}
We start by explaining our datasets, introducing a baseline algorithm and showing how we choose the balancing coefficient $\lambda$ for each dataset. We then empirically evaluate the performance of  {\networkbalance}  in terms of the objective, expert load, radius constraint and running time. We also compare its performance with {\thresholdgreedy}.

\subsubsection{Datasets}
\label{sec:datasets2}
We follow the method of 
\citet{anagnostopoulos2012online} and create social graphs with expert coordination costs for our
datasets,  
{\imdb}, {\bibsonomy}, {\freelancer}, and {\guru}. 

For the {\imdb} dataset, we create a social graph among the directors, who form the vertices in the graph. We connect directors using actors as intermediaries: we form and edge between two directors if they have directed at least two distinct actors in common. The cost of the edge is set to $e^{-fD}$, where $D$ is the number of distinct actors directed by the two directors. The distance function  $e^{-fD}$ takes values between 0 and~1, and we note that it quickly converges to the value 0 as the number of common actors $D$ between two directors increases. As in \citet{anagnostopoulos2012online}, we set the value of the parameter $f = \frac{1}{10}$ since this value of $f$ yields a reasonable edge-weight distribution of coordination costs in the social graph for our {\imdb} dataset.

For the {\bibsonomy} dataset, we create a social graph among authors using co-authorship to define the strength of social connection. Two authors are connected with an edge if they have written at least one paper together. Again the cost of the edge is set to $e^{-fD}$, where $D$ is the number of distinct papers coauthored by the two authors. Similar to \citet{anagnostopoulos2012online}, we set the value of the parameter $f = \frac{1}{10}$, so as to obtain a reasonable distribution of edge-weights in our {\bibsonomy} social graph.

For the {\freelancer} and {\guru} datasets, we use the following heuristic to create a social graph among the experts in each dataset: experts with similar, overlapping sets of skills have a lower coordination cost since they are ``closer" to each other in terms of their ability to perform tasks well together.
To create the expert social graphs, we consider each pair of experts, and compute the Jaccard distance between the sets of skills of the pair of experts. The cost of the edge between each pair of experts is then represented by the Jaccard distance between their skill sets. We note that the Jaccard distance takes values between 0 and 1, and is 0 if two experts have identical skill sets, and 1 if their skills are mutually exclusive.

For all datasets, we keep the same names as before and we
present the summary graph statistics of these datasets in Table~\ref{tab:summarystats-graph}. The average path length corresponds to the average shortest path length between all pairs of nodes in the graph, and the average degree is the average of the \textit{unweighted} degrees of all nodes in the graph. 

\begin{table}
\caption{Summary statistics of our graph datasets.}
\label{tab:summarystats-graph}
{\small
\begin{tabular}{lrrr}
    \toprule
    Dataset & \multicolumn{1}{c}{Number} & \multicolumn{1}{c}{Average}  & \multicolumn{1}{c}{Average} \\
     & \multicolumn{1}{c}{of nodes}  & \multicolumn{1}{c}{path length} &  \multicolumn{1}{c}{degree}\\
    \midrule
    {\imdbone}& 1000 & 7.6 & 1.4\\
    {\imdbtwo}& 3000 & 4.2 & 4.5\\
    {\imdbthree}& 4000 & 3.4 & 8.0\\
    {\bibsonomyone}& 500 & 2.6 & 3.1\\
    {\bibsonomytwo}& 1500 & 1.4 & 25.8\\
    {\bibsonomythree}& 2500 & 1.4 & 29.1\\
    {\freelancer}& 1212 & 1.2 & 19.2\\
    {\guru}& 6120 & 1.1 & 42.0\\
    \bottomrule
\end{tabular}
}
\end{table}

\subsubsection{Baseline}
\label{sec:baselines-nb}
We use the following greedy variant of the {\networkbalance} algorithm as a baseline.

\spara{{\greedyindividual}:} This algorithm has a similar logic as {\networkbalance} as it iterates over different workloads. However, {\greedyindividual} does not create candidate teams. 
The algorithm assigns individual experts to tasks in a greedy manner: for each of the $\numexperts\numtasks$ expert-task pairs, we consider the task coverage the expert provides for that task.
We then greedily assign experts to tasks by selecting experts in order of decreasing coverage they provide for tasks. 
As we assign experts to tasks, we also ensure that each expert satisfies the workload constraint $\threshold$.
Note this baseline has a computational overhead of checking that every new expert assigned to a task satisfies the radius constraint $r$ with respect to all other experts already assigned to that task. {\greedyindividual} has a running time of $\bigO(\numtasks\numexperts^2)$.

\subsubsection{Tuning Coverage and Workload Importance}
\label{sec:tuning-nb}
In this section we discuss how the value of $\lambda$ is selected. 
We follow a similar technique as the previous section and determine the \emph{best-network workload} and the corresponding value of the \emph{best-network objective} for a large value of lambda, $\lambda = 5$.
We then compute the \textit{best-network} values for smaller values of $\lambda \in (0, 5]$, 
and plot the corresponding values of $\cov(\assignment)$ and $\maxload(\assignment)$ for each $\lambda$ value. The scatter plots for $r = 0.3$ and $r = 0.7$ are visualized in Figures~\ref{fig:lambda-0.3} and~\ref{fig:lambda-0.7}, respectively. 

We visually inspect these plots to identify a suitable value of $\lambda$ such that the best-network workload and best-network objective values yield a high value for  $\cov(\assignment)$, while simultaneously giving a reasonably low value for $\maxload(\assignment)$. The final $\lambda$ values we selected are shown besides the different datasets in Table~\ref{tab:nbalance-algo-results}. 

\begin{figure*}
    \centering
    \begin{subfigure}[b]{0.24\textwidth}
        \centering
        \includegraphics[width=\textwidth]{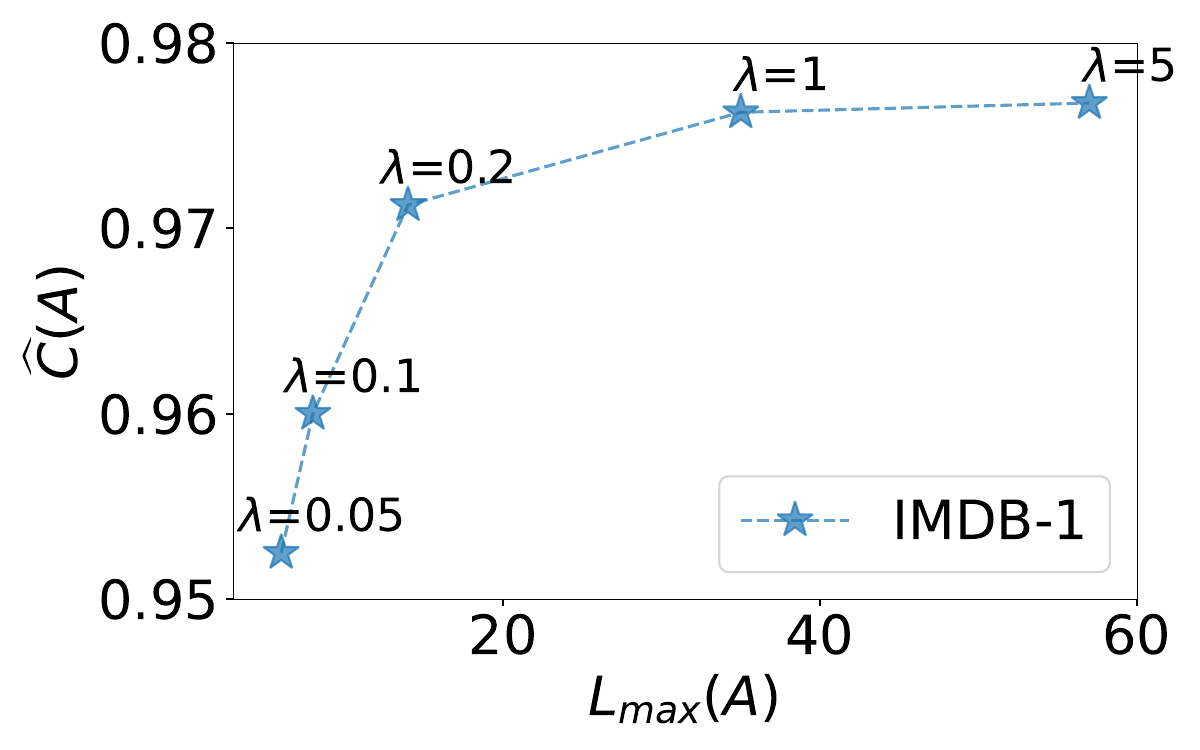}
    \end{subfigure}
    \hfill
    \begin{subfigure}[b]{0.24\textwidth}
        \centering
        \includegraphics[width=\textwidth]{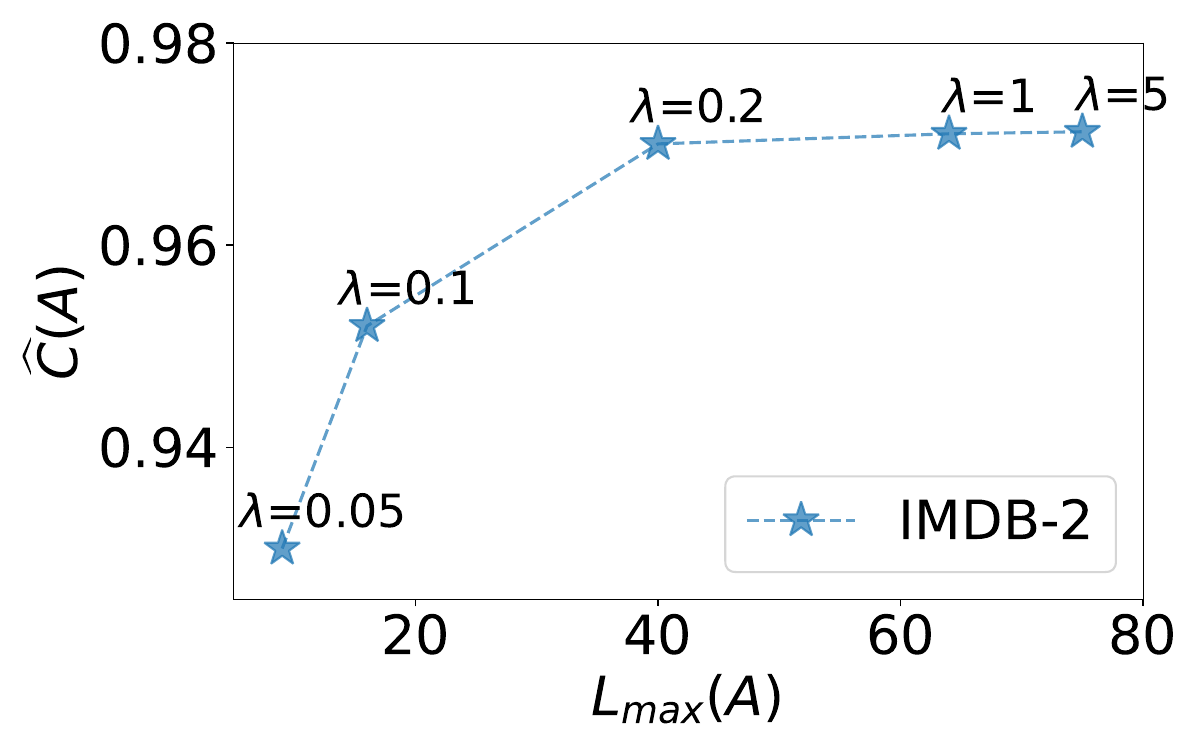}
    \end{subfigure}
    \hfill
    \begin{subfigure}[b]{0.24\textwidth}
        \centering
        \includegraphics[width=\textwidth]{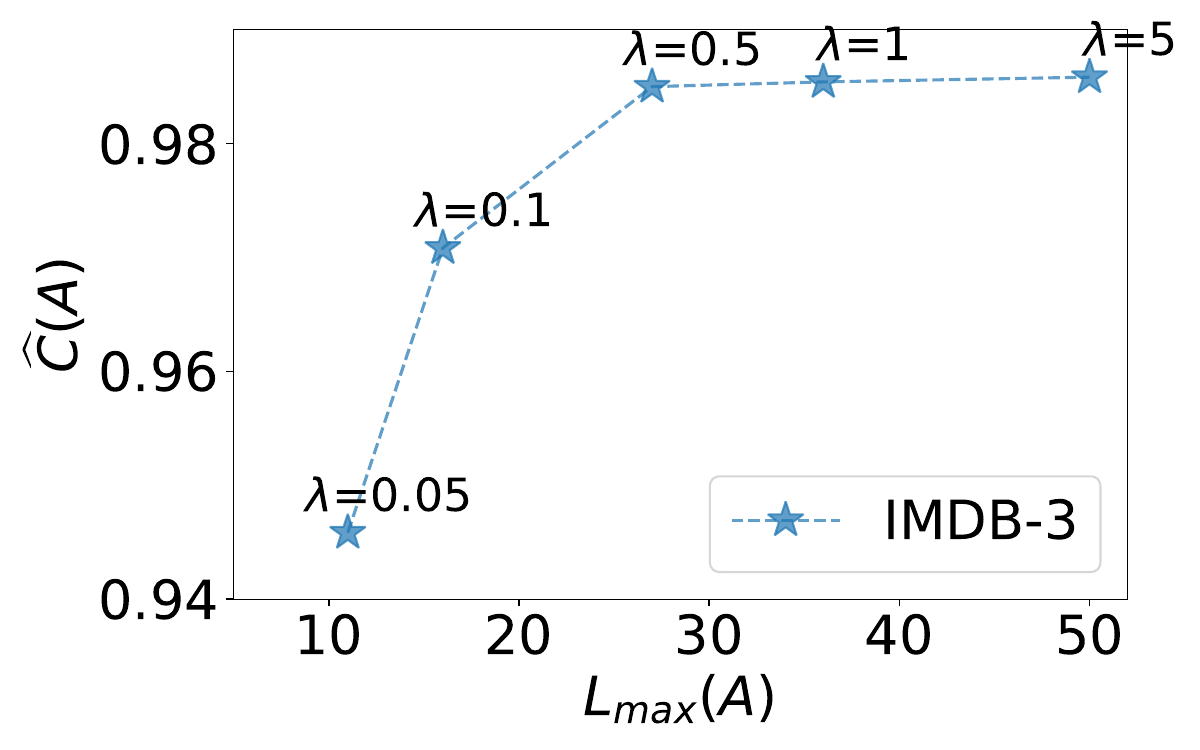}
    \end{subfigure}
    \hfill
    \begin{subfigure}[b]{0.24\textwidth}
        \centering
        \includegraphics[width=\textwidth]{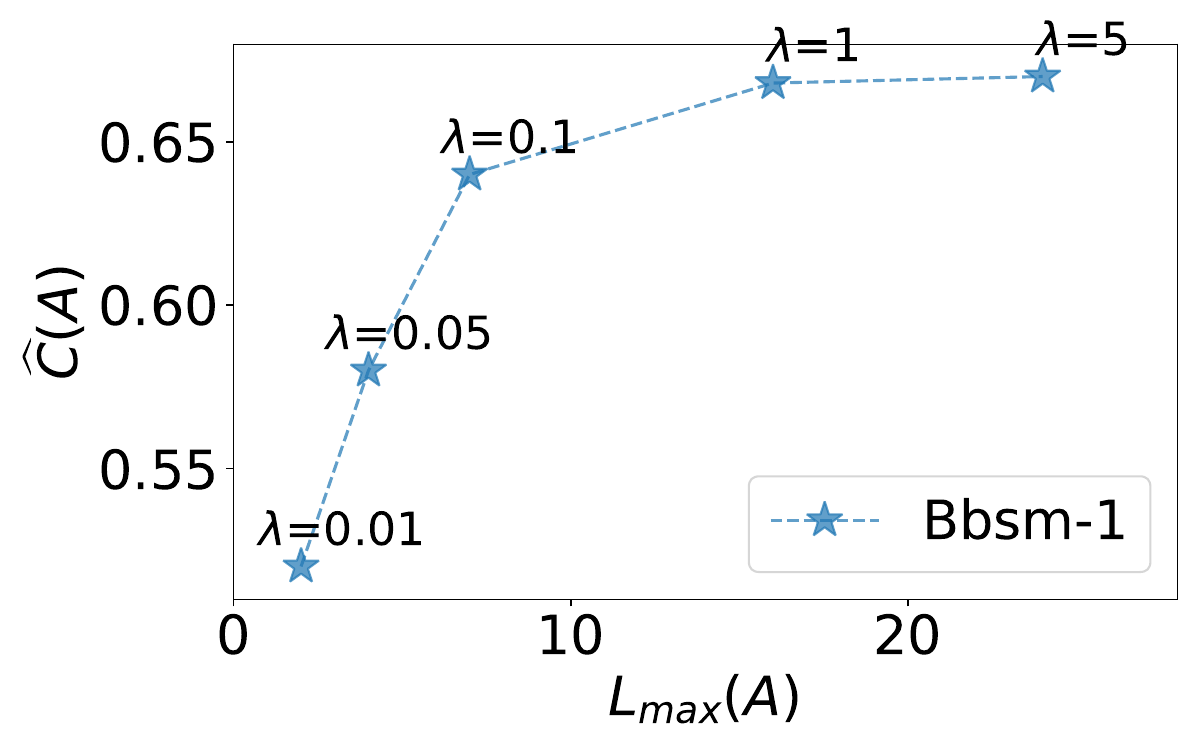}
    \end{subfigure}
    \begin{subfigure}[b]{0.24\textwidth}
        \centering
        \includegraphics[width=\textwidth]{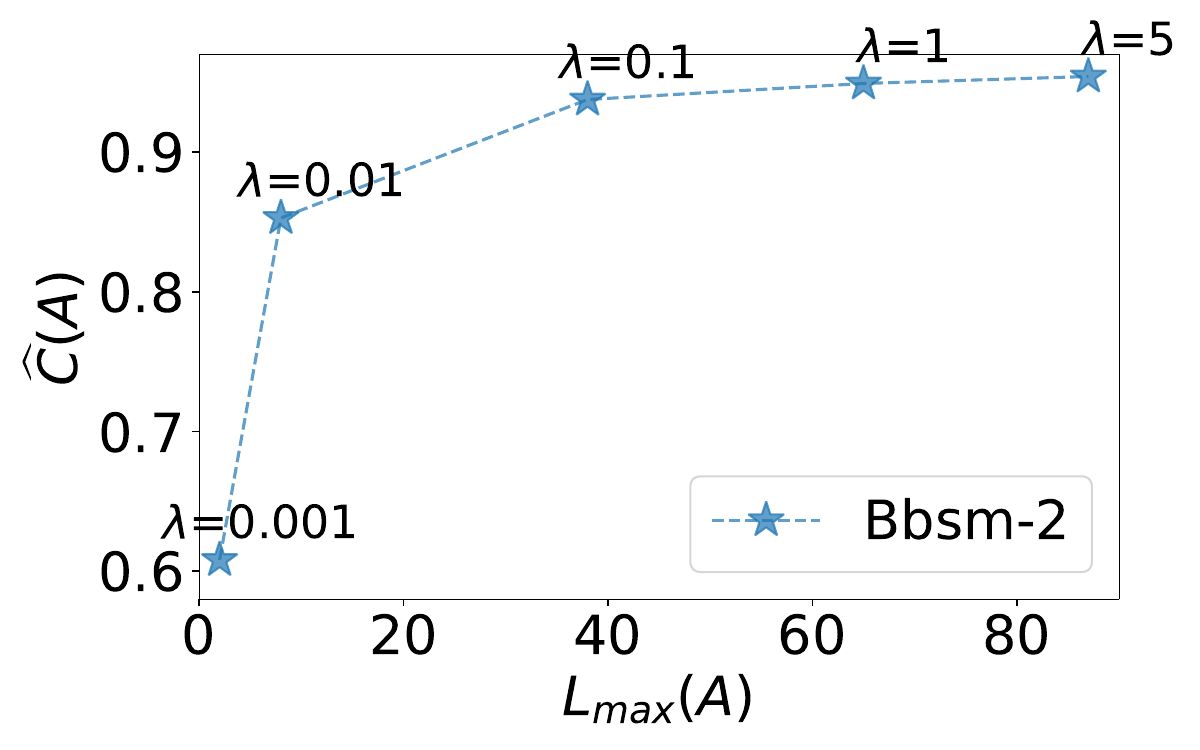}
    \end{subfigure}
    \hfill
    \begin{subfigure}[b]{0.24\textwidth}
        \centering
        \includegraphics[width=\textwidth]{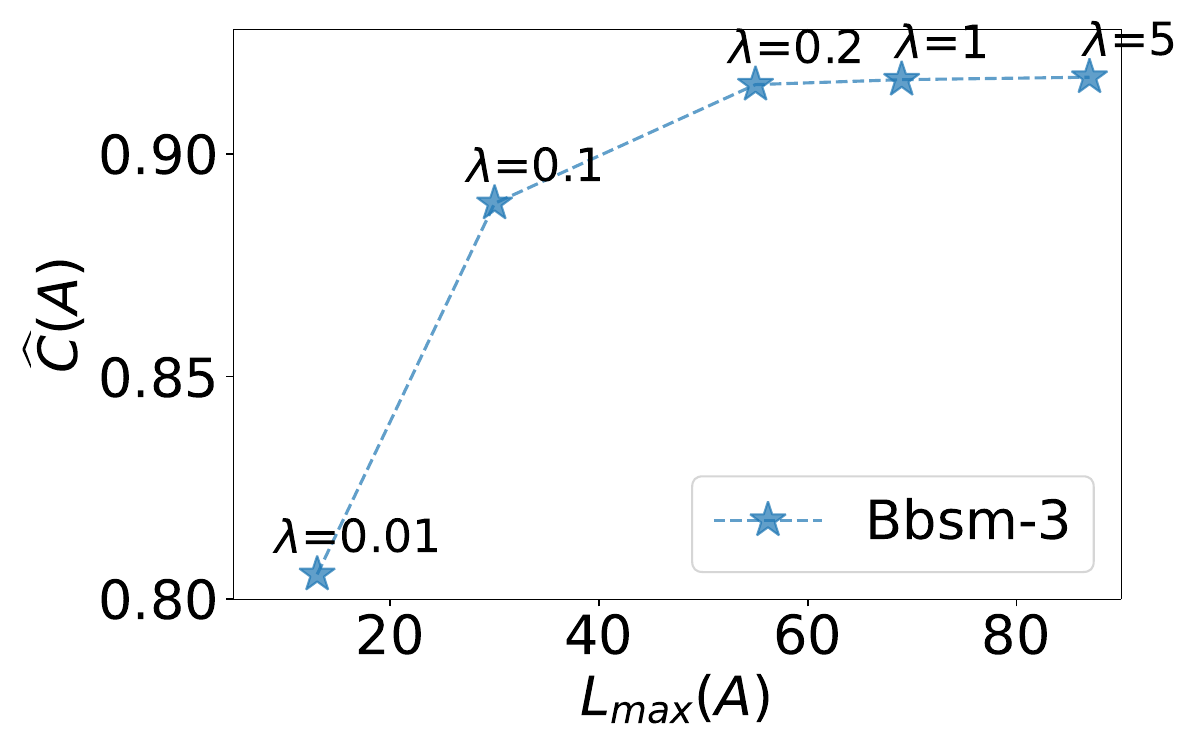}
    \end{subfigure}
    \hfill
    \begin{subfigure}[b]{0.24\textwidth}
        \centering
        \includegraphics[width=\textwidth]{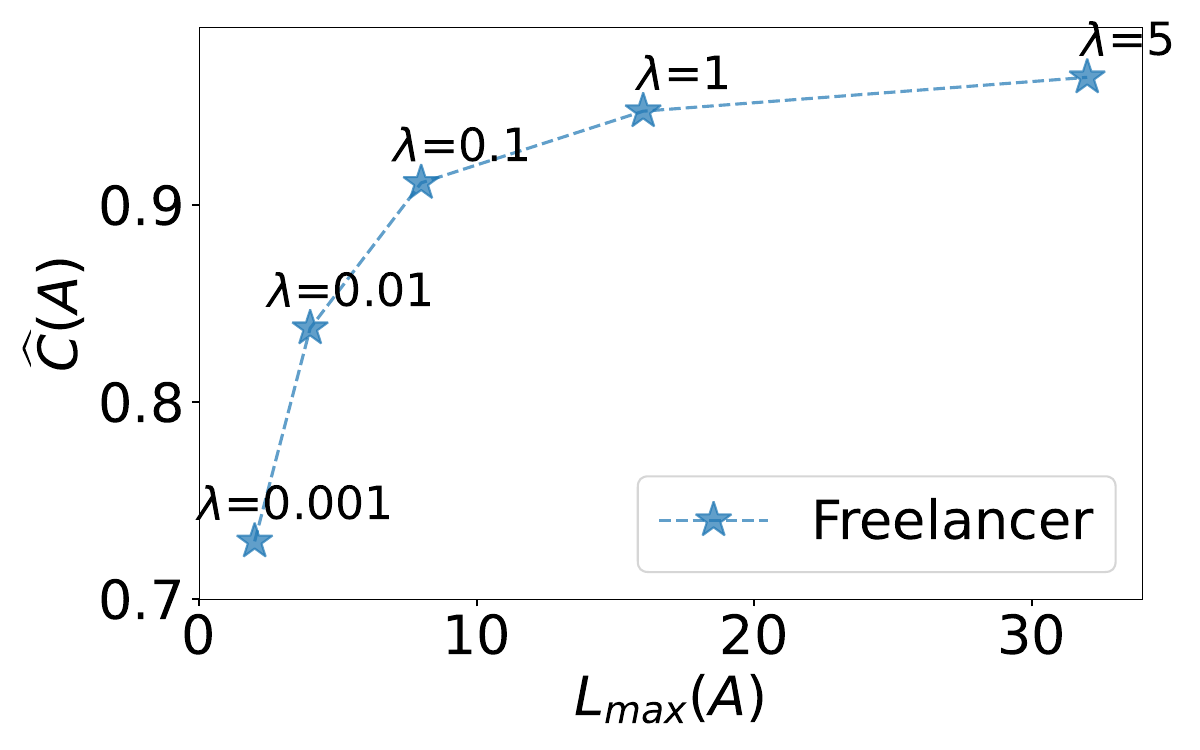}
    \end{subfigure}
    \hfill
    \begin{subfigure}[b]{0.24\textwidth}
        \centering
        \includegraphics[width=\textwidth]{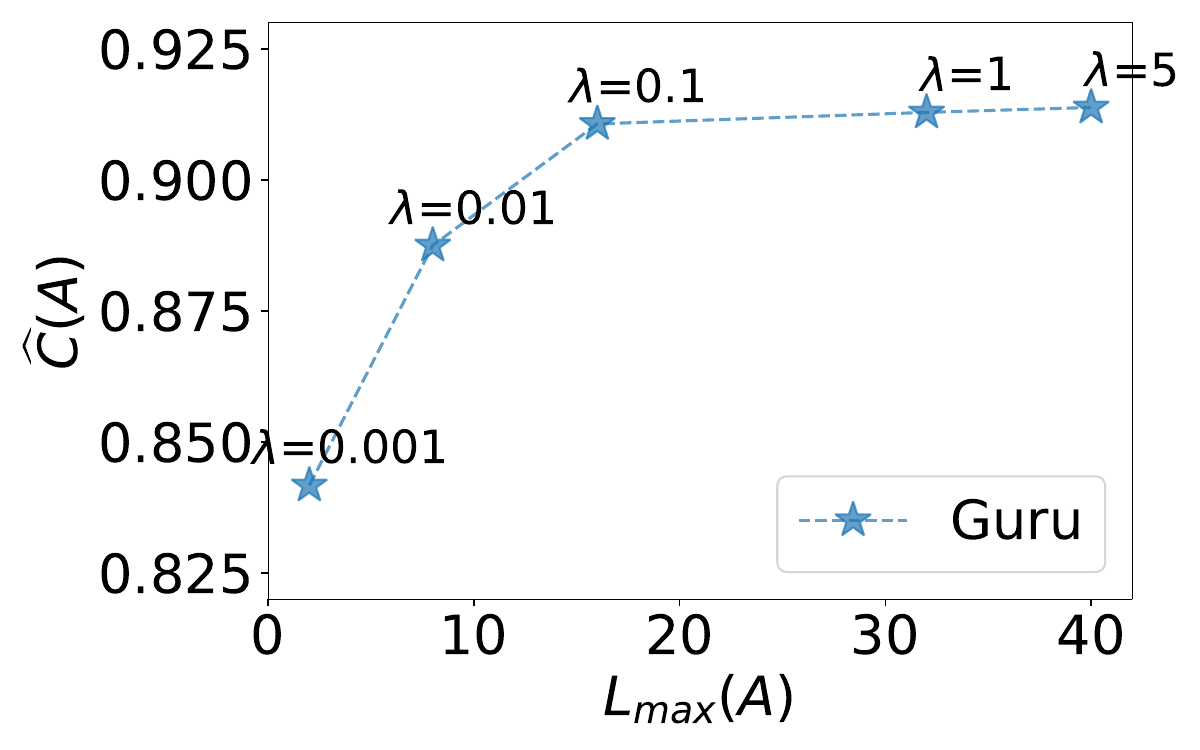}
    \end{subfigure}
    \caption{The best-greedy workload $\maxload(\assignment)$ value and the coverage $\cov(\assignment)$ corresponding to the best-greedy objective $\objective^{\lambda}(\assignment)$ computed by {\networkbalancelp} for $r = 0.3$. Each subplot shows a range of values of the balancing coefficient $\lambda$ for each dataset. }
    \label{fig:lambda-0.3}
\end{figure*}

\begin{figure*}
    \centering
    \begin{subfigure}[b]{0.24\textwidth}
        \centering
        \includegraphics[width=\textwidth]{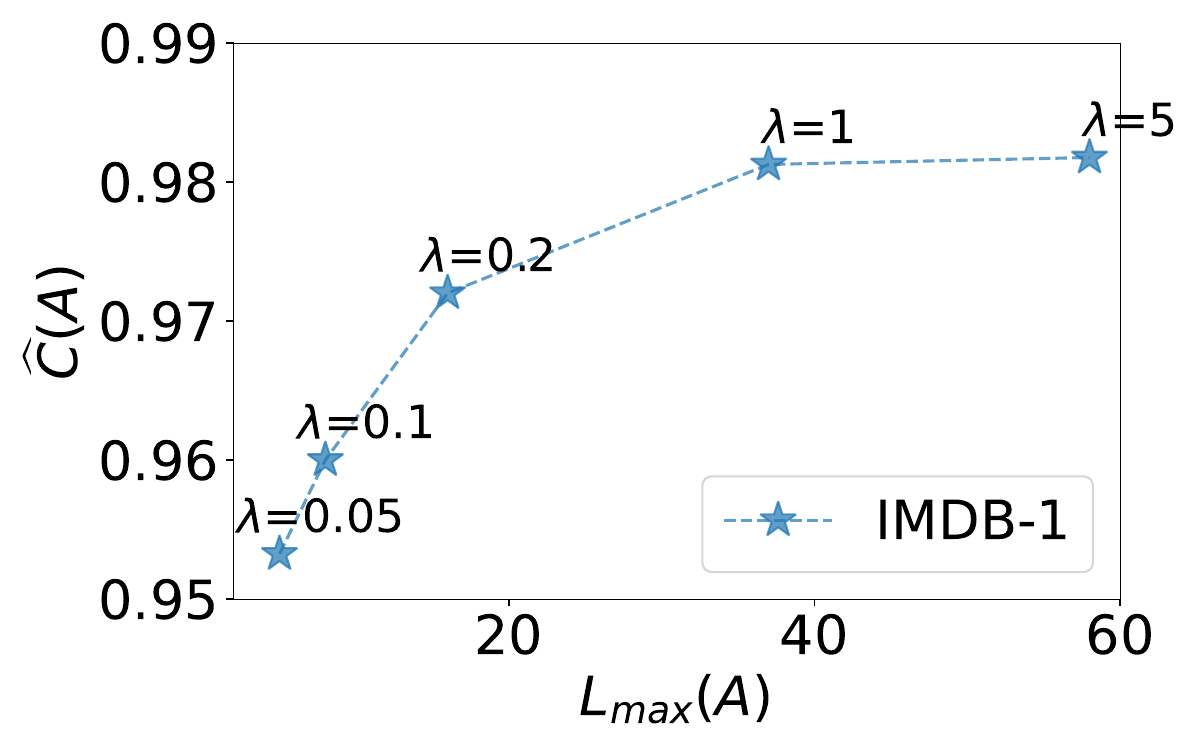}
    \end{subfigure}
    \hfill
    \begin{subfigure}[b]{0.24\textwidth}
        \centering
        \includegraphics[width=\textwidth]{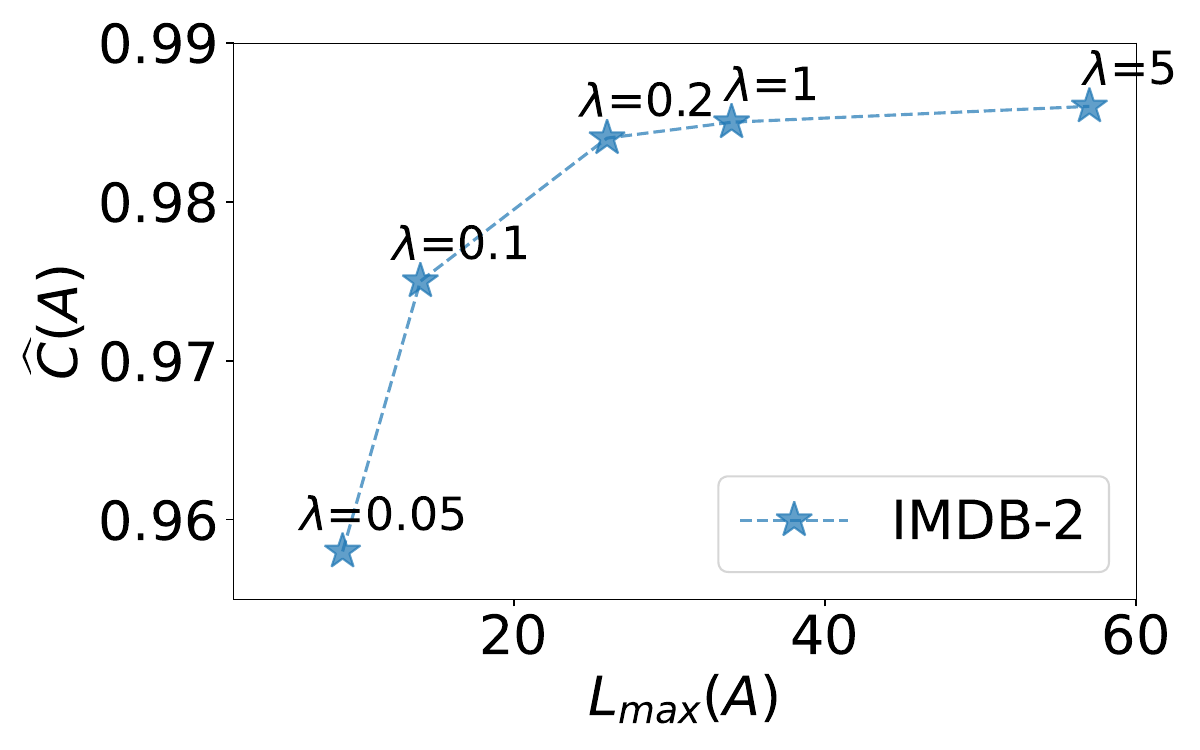}
    \end{subfigure}
    \hfill
    \begin{subfigure}[b]{0.24\textwidth}
        \centering
        \includegraphics[width=\textwidth]{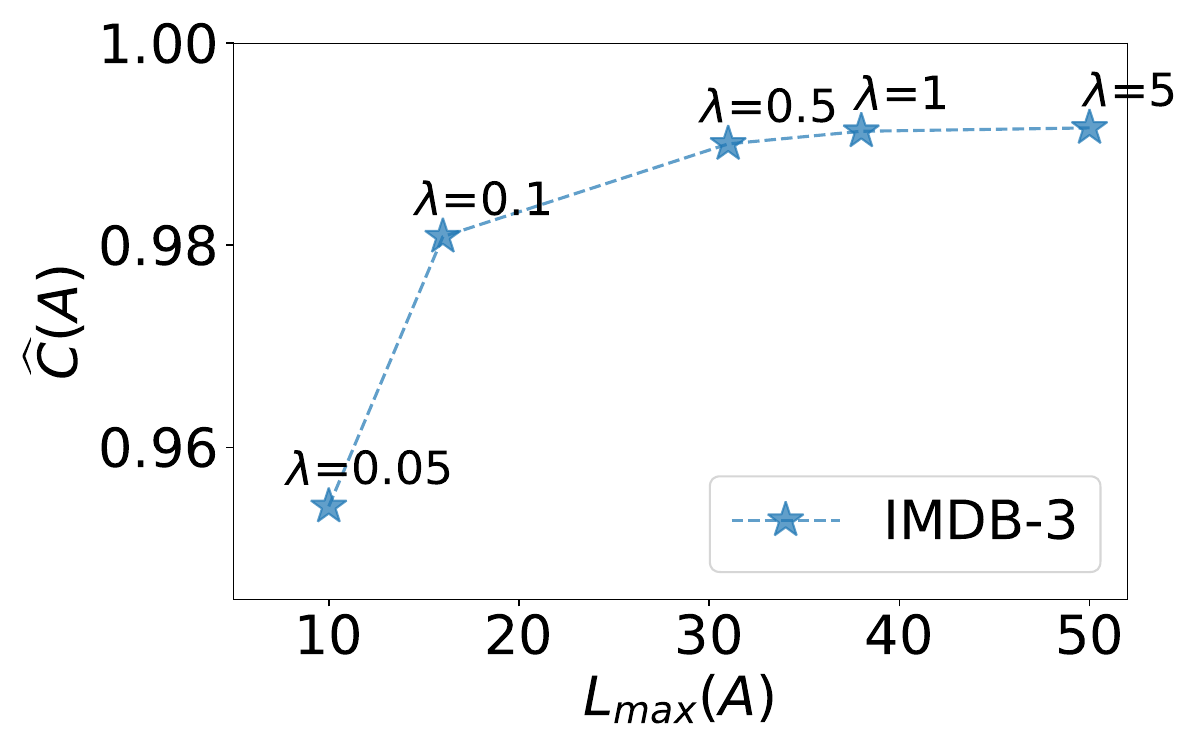}
    \end{subfigure}
    \hfill
    \begin{subfigure}[b]{0.24\textwidth}
        \centering
        \includegraphics[width=\textwidth]{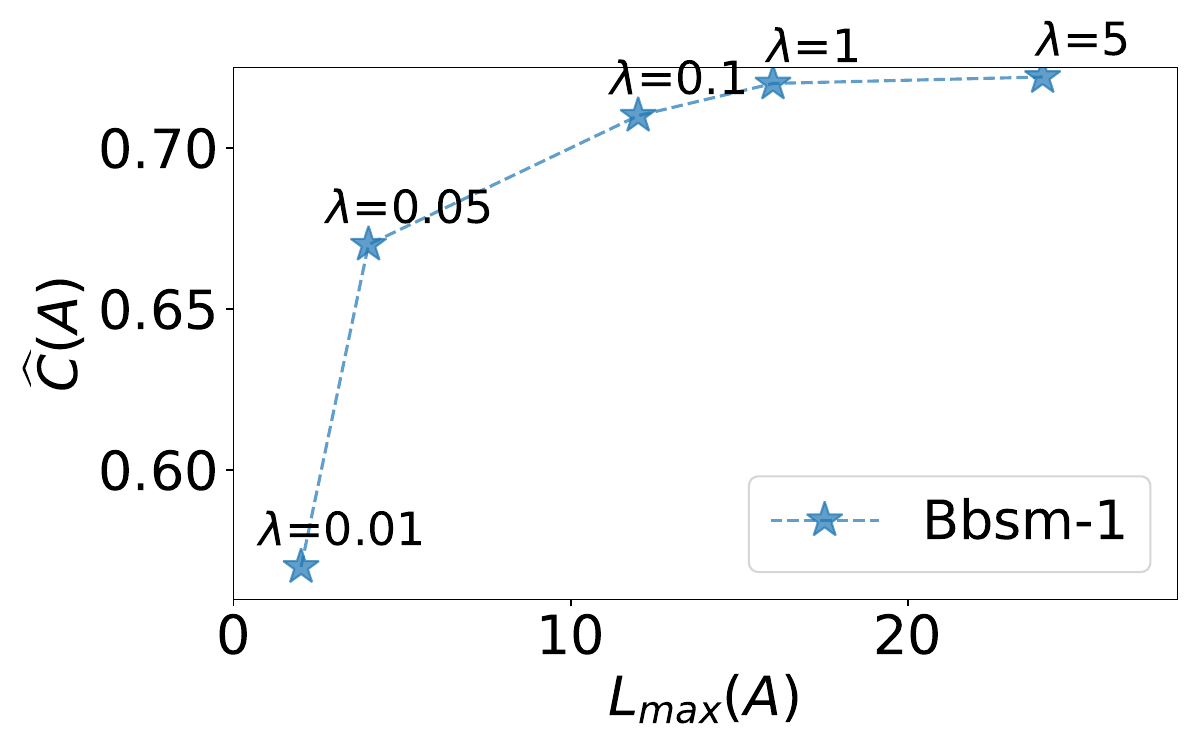}
    \end{subfigure}
    \begin{subfigure}[b]{0.24\textwidth}
        \centering
        \includegraphics[width=\textwidth]{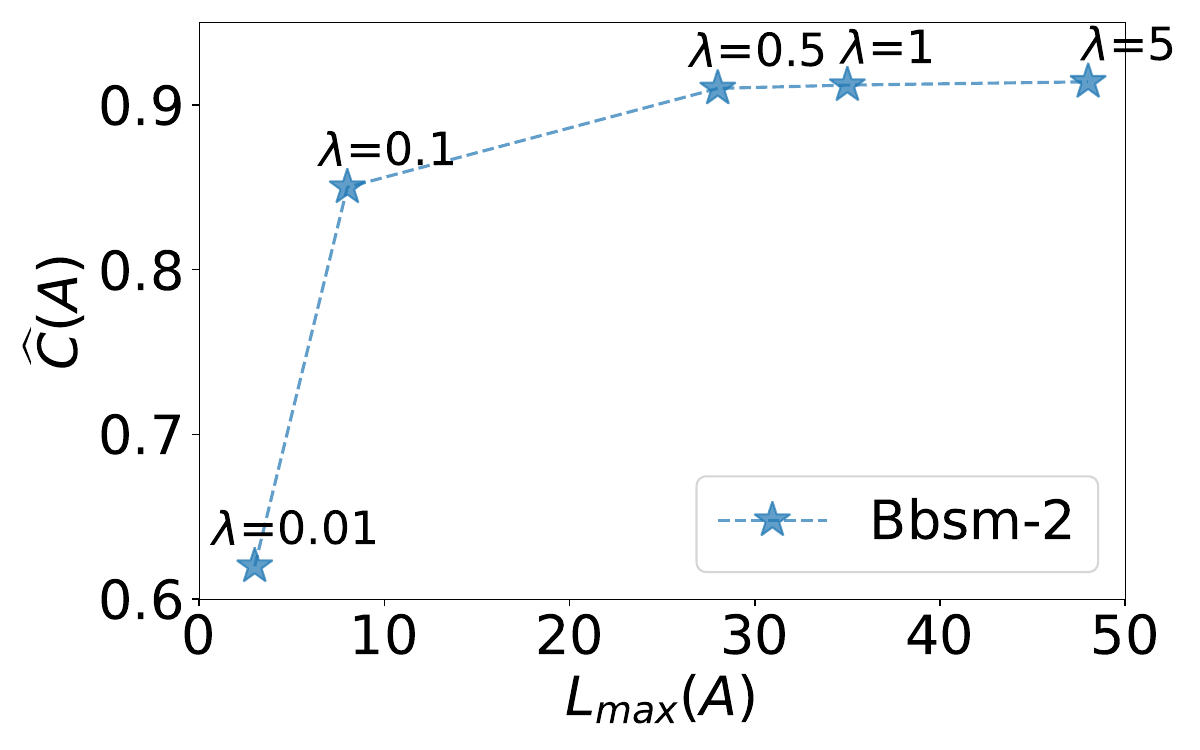}
    \end{subfigure}
    \hfill
    \begin{subfigure}[b]{0.24\textwidth}
        \centering
        \includegraphics[width=\textwidth]{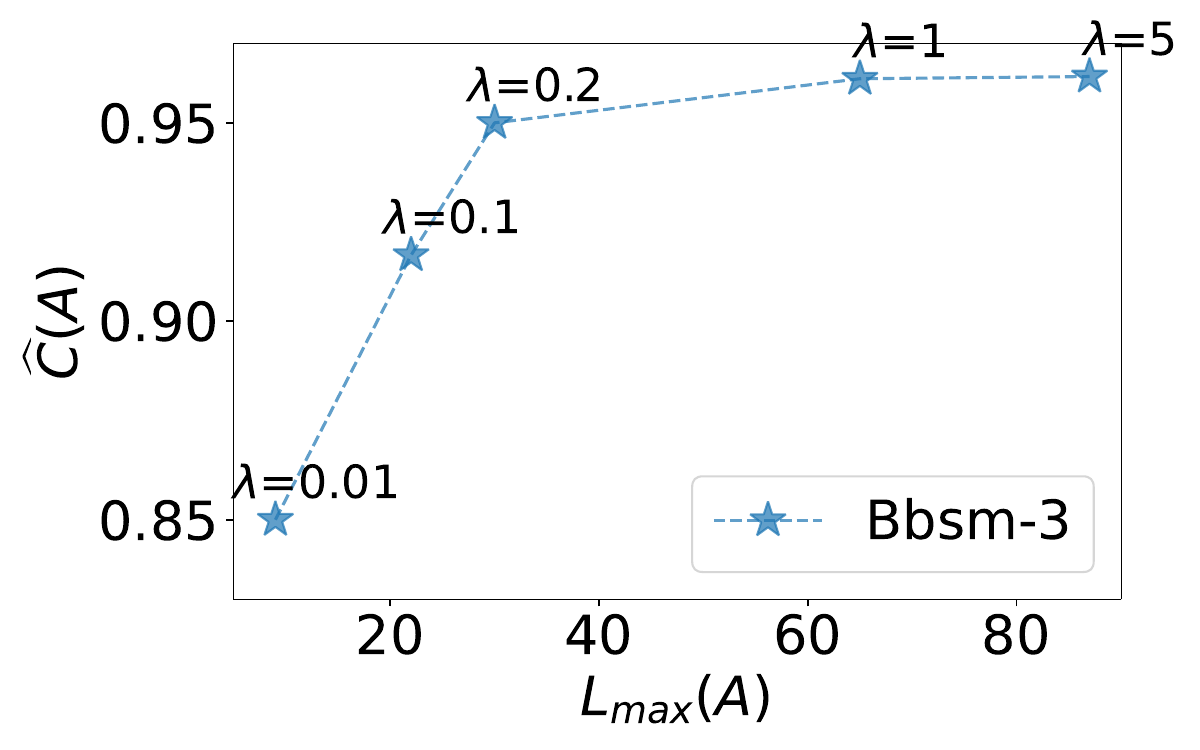}
    \end{subfigure}
    \hfill
    \begin{subfigure}[b]{0.24\textwidth}
        \centering
        \includegraphics[width=\textwidth]{figures/nbalance/Freelancer-lambda-r0.7.pdf}
    \end{subfigure}
    \hfill
    \begin{subfigure}[b]{0.24\textwidth}
        \centering
        \includegraphics[width=\textwidth]{figures/nbalance/Guru-lambda-r0.7.pdf}
    \end{subfigure}
    \caption{The best-greedy workload $\maxload(\assignment)$ value and the coverage $\cov(\assignment)$ corresponding to the best-greedy objective $\objective^{\lambda}(\assignment)$ computed by {\networkbalancelp} for $r = 0.7$. Each subplot shows a range of values of the balancing coefficient $\lambda$ for each dataset. }
    \label{fig:lambda-0.7}
\end{figure*}

\subsubsection{Evaluation}
In this section, we evaluate the performance of the different instantiations of {\networkbalance} we described in Section~\ref{sec:nbalance-algo-names}. Specifically, we evaluate {\networkbalancelp}, {\networkbalancegreedy} and {\lpallradii} and compare their performance with each other, and with the {\greedyindividual} baseline. 
 We compare the algorithms using the objective function ($\objective^{\lambda}$), the average coverage per skill $\widehat{\cov} = \frac{1}{\numtasks}\cov$, and the maximum load $\maxload$.
We omit the results for {\networkbalance}\texttt{-All-Greedy} since {\lpallradii} outperformed it in all aspects.

Since the coordination costs of our datasets have values between 0 and 1 in our datasets, we ran the algorithms for several values of the radius constraint $r \in \{0.1, 0.3, 0.5, 0.7, 0.9\}$.
We observed that $r \in \{0.1, 0.3\}$ yielded similar objective values, coverages, and workloads, as did $r \in \{0.5, 0.7, 0.9\}$. Consequently, we only report results in Table~\ref{tab:nbalance-algo-results} for $r = 0.3$ and $r = 0.7$.

\spara{Objective values $\objective$ and workload $\maxload$:} From Table~\ref{tab:nbalance-algo-results}, we observe that both {\networkbalancelp} and {\networkbalancegreedy} perform very well for all datasets in terms of the average task coverage $\widehat{\cov}$. For {\imdb} (for both $r = 0.3$ and $r = 0.7$) we observe high coverage values greater than or equal to $0.95$. Additionally, these algorithms find reasonably low expert workloads of $\maxload \in [13,17]$.

We observe that {\lpallradii} and {\greedyindividual} also perform well on {\imdb} in terms of average coverage, with coverage values greater than $0.92$ for $r = 0.3$, and coverage values greater than $0.95$ for $r = 0.7$. However we observe that {\greedyindividual} returns significantly higher workload values of $\maxload \geq 31$. Similarly, for {\imdbtwo} and {\imdbthree}, we observe that {\lpallradii} also returns higher workload values of $\maxload \geq 23$.

For {\bibsonomy}, we observe that the {\networkbalance} algorithms have the highest $\objective^{\lambda}$ and $\widehat{\cov}$ values and also the lowest $\maxload$ values (for both $r = 0.3$ and $r = 0.7$).  {\greedyindividual} yields a significantly lower coverage with a much higher expert workload. We note that for {\bibsonomytwo}, {\lpallradii} gives the best results in terms of the objective $\objective^{\lambda}$ and coverage values. However the {\networkbalance} algorithms only perform marginally worse in terms of the objective and have similar workload values.

For the {\freelancer} and {\guru} datasets, we observe that {\networkbalancelp} yields the best $\objective^{\lambda}$ and $\widehat{\cov}$, with workload values $\maxload \leq 15$.

\spara{Effect of radius constraint:} We observe that $\widehat{\cov}$ decreases slightly for all algorithms, across our datasets as the radius constraint decreases from $r = 0.7$ to $r = 0.3$. This is expected since a smaller radius implies that the potential teams of experts available is also smaller. We observe, however, that the difference is marginal, with a decrease in coverage of less than $3\%$.
We observe that the maximum workload values returned by the {\networkbalance} algorithms are also comparable for the different radius constraints. These observations lead us to conclude that the increase in coverage due to increasing team radius could be attributed to the availability of new experts that are within the new, larger team radius.

\spara{Mean expert workloads:} We examine the teams formed by our algorithm in terms of the \emph{mean} of the expert workloads. For {\imdb}, we have that  $\maxload \in [13,17]$, yet the mean mean expert load of the {\networkbalance} solutions is in the range $[3.9, 4]$. This indicates that while there are a few experts who are heavily loaded, on average the {\networkbalance} algorithms find good load-balancing solutions. In contrast, we observe that the baseline {\greedyindividual} has a higher mean expert load for the {\imdb} datasets, in the range $[5.5, 6]$.

Similarly, we observe that for  {\bibsonomyone} and {\bibsonomytwo} the mean expert load of the {\networkbalance} algorithms is in the range $[1.8, 2]$ for both datasets (and for both radius constraints). On the other hand, the mean expert load of {\greedyindividual} is higher for these datasets in the range $[2.7, 3]$. A similar pattern was observed for the {\freelancer} and {\guru} datasets as well.

\spara{Comparison with {\thresholdgreedy}:} We compare the performance of  {\networkbalance} with {\thresholdgreedy} by comparing values in Tables~\ref{tab:bbalance-algo-results} and~\ref{tab:nbalance-algo-results}. While $\widehat{\cov}$ returned by both algorithms is comparable, we see that {\thresholdgreedy}  finds a slightly higher coverage across all the datasets. Additionally, the maximum workload values achieved by  {\networkbalance}  is higher than those achieved by {\thresholdgreedy}. This is because the problem solved by the former algorithms is harder than the one solved by the latter; there are more constraints in terms of how experts can be combined into teams.

\spara{Running time:} We record the total running time of all algorithms for $r = 0.7$ (we observed similar patterns for the other radii values) and illustrate them in Fig.~\ref{fig:runtimes-nbalance-0.7}.
We observe that {\networkbalancelp} has the best running time of all algorithms, and this is closely followed by {\networkbalancegreedy}. While {\lpallradii} does perform well on some of the datasets, we observe that it has the maximum running time of all algorithms, across all datasets. This is an expected result since {\lpallradii} considers many more candidate teams than the other algorithms. 

\begin{table*}
\caption{Experimental performance of {\networkbalance} and {\greedyindividual} in terms of the objective $\objective^{\lambda}$, the maximum load $\maxload$ and the average task coverage $\widehat{\cov}= \frac{1}{\numtasks}\cov$. The best values for each dataset are in bold.}
\label{tab:nbalance-algo-results}
\tiny
\begin{tabular}{c l  rrr  rrr  rrr  rrr }
    \toprule
    &\multirow{2}{*}{Dataset $(\lambda)$} &
    \multicolumn{3}{c}{\networkbalancegreedy} &
    \multicolumn{3}{c}{\networkbalancelp} &
    \multicolumn{3}{c}{\lpallradii} &
    \multicolumn{3}{c}{\greedyindividual}  \\
	\cmidrule(lr){3-5} \cmidrule(lr){6-8} \cmidrule(lr){9-11} \cmidrule(lr){12-14}
    & & $\objective^{\lambda}$& $\maxload$ & $\widehat{\cov}$ & $\objective^{\lambda}$& $\maxload$ & $\widehat{\cov}$ & $\objective^{\lambda}$& $L$ & $\widehat{\cov}$ & $\objective^{\lambda}$& $L$ & $\widehat{\cov}$  \\
    \midrule
        \parbox[t]{2mm}{\multirow{8}{*}{\rotatebox[origin=c]{90}{$r = 0.3$}}}
    &{\imdbone} $(0.2)$ & 752 & \textbf{14} & \textbf{0.97} & \textbf{752} & \textbf{14} & 0.96 & 756 & 8 & 0.92 & 742 & 34 & 0.93 \\
    &{\imdbtwo} $(0.1)$ & 960 & 15 & 0.95 & \textbf{961} & \textbf{13} & \textbf{0.96} & 942 & 26 & 0.94 & 938 & 33 & 0.95 \\
    &{\imdbthree} $(0.1)$& 1149 & 17 & \textbf{0.97} & \textbf{1153} & \textbf{15} & 0.95 & 1143 & 23 & 0.95 & 1138 & 22 & 0.94 \\
    
    &{\bibsonomyone} $(0.1)$& \textbf{56} & \textbf{7} & \textbf{0.64} & \textbf{56} & 8 & \textbf{0.64} & 54 & 9 & 0.63 & 46 & 16 & 0.62 \\
    &{\bibsonomytwo} $(0.01)$ & 26 & 9 & 0.86 & 27 & \textbf{7} & 0.86 & \textbf{29} & 8 & \textbf{0.89} & 25 & 9 & 0.84 \\
    &{\bibsonomythree} $(0.1)$ & 767 & 34 & \textbf{0.89} & \textbf{785} & \textbf{32} & \textbf{0.89} & 748 & 34 & 0.87 & 738 & 37 & 0.86 \\
      &{\freelancer} $(0.1)$	& 80 & 9 & 0.89 & \textbf{82} & \textbf{8} & \textbf{0.9} & 78 & 11 & 0.87 & 74 & 14 & 0.80\\
   &{\guru} $(0.1)$ & 268 & \textbf{11} & 0.86 & \textbf{272} & 15 & \textbf{0.9} & 256 & 28 & 0.88 & 241 & 15 & 0.76 \\

    \midrule
    \parbox[t]{2mm}{\multirow{8}{*}{\rotatebox[origin=c]{90}{$r = 0.7$}}}
    &{\imdbone} $(0.2)$ & 761 & 16 & \textbf{0.97} & 762 & 15 & 0.96 & \textbf{766} & \textbf{8} & \textbf{0.97} & 748 & 31 & 0.97 \\
    &{\imdbtwo} $(0.1)$ & 960 & 15 & \textbf{0.97} & \textbf{963} & \textbf{14} & \textbf{0.97} & 943 & 30 & 0.96 & 941 & 31 & 0.97 \\
    &{\imdbthree} $(0.1)$& 1159 & 17 & \textbf{0.98} & \textbf{1161} & \textbf{15} & 0.97 & 1147 & 28 & 0.97 & 1143 & 35 & 0.96 \\
    
    &{\bibsonomyone} $(0.05)$& \textbf{30} & \textbf{4} & 0.67 & \textbf{30} & 6 & \textbf{0.68} & 29 & 4 & 0.65 & 15 & 16 & 0.63 \\
    &{\bibsonomytwo} $(0.1)$ & 426 & 8 & 0.87 & 428 & 8 & 0.87 & \textbf{438} & \textbf{8} & \textbf{0.90} & 332 & 32 & 0.72 \\
    &{\bibsonomythree} $(0.2)$ & 1677 & 32 & \textbf{0.95} & \textbf{1685} & \textbf{31} & \textbf{0.95} & 1638 & 37 & 0.93 & 1360 & 61 & 0.79 \\
        &{\freelancer} $(0.1)$	& 82 & 9 & \textbf{0.91} & \textbf{83} & \textbf{8} & \textbf{0.91} & 81 & 8 & 0.90& 77 & 17 & 0.84 \\
   &{\guru} $(0.1)$ & 271 & \textbf{12} & 0.89 & \textbf{275} & 15 & \textbf{0.91} & 260 & 30 & 0.90& 242 & 14 & 0.78 \\
    \bottomrule
\end{tabular}
\end{table*}

\begin{figure*}
    \centering
    \includegraphics[width=\textwidth]{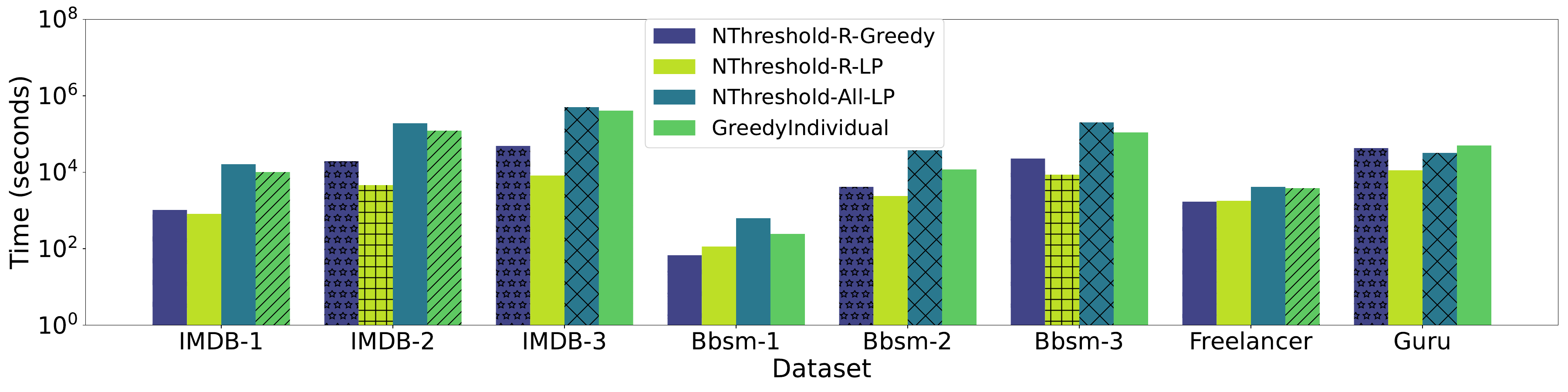}
    \caption{Running time (in seconds) of {\networkbalance} and {\greedyindividual}, in logarithmic scale for radius constraint $r=0.7$.}
    \label{fig:runtimes-nbalance-0.7}
\end{figure*}

\subsection{Team Characteristics}
In this section, we investigate the characteristics of the teams formed by  
{\thresholdgreedy}, {\networkbalance} and {\greedyindividual}. We examine four characteristics: team size, team radii, within-team degree distributions and average pairwise distance of the experts in the formed teams
In the remainder of the section, we discuss the team characteristics in detail. For this analysis, we consider $r = 0.7$, but similar characteristics were observed for other radii as well. We report the average values of the different characteristics for each algorithm and dataset in Tables~\ref{tab:team-characteristics-size} and~\ref{tab:team-characteristics-connectivity}. 

\begin{table*}
\centering
\caption{Team characteristics of {\thresholdgreedy}, {\networkbalance} and {\greedyindividual} algorithms in terms of the average team size $\avgteamsize$, maximum team size $\maxteamsize$ and average team radius $\avgteamradius$.}
\label{tab:team-characteristics-size}
\tiny
\begin{tabular}{ l  rrr  rrr  rrr  rrr }
    \toprule
    \multirow{2}{*}{Dataset} &
    \multicolumn{3}{c}{\thresholdgreedy} &
    \multicolumn{3}{c}{\networkbalancegreedy} &
    \multicolumn{3}{c}{\networkbalancelp} &
    \multicolumn{3}{c}{\greedyindividual}  \\
    \cmidrule(lr){2-4} \cmidrule(lr){5-7} \cmidrule(lr){8-10} \cmidrule(lr){11-13}

    & $\avgteamsize$ & $\maxteamsize$ & $\avgteamradius$ & $\avgteamsize$ & $\maxteamsize$ & $\avgteamradius$ & $\avgteamsize$ & $\maxteamsize$ & $\avgteamradius$ & $\avgteamsize$ & $\maxteamsize$ & $\avgteamradius$  \\
    \midrule

    {\imdbone} & 1.09 & 4 & 0.99 & 1.04 & 3 & 0.68 & 1.04 & 3 & 0.66 & 1.02 & 2 & 0.67 \\
    {\imdbtwo} & 1.06 & 3 & 0.99 & 1.03 & 3 & 0.65 & 1.05 & 4 & 0.64 & 1.02 & 2 & 0.62 \\ 
    {\imdbthree} & 1.06 & 4 & 0.99 & 1.04 & 4 & 0.64 & 1.04 & 3 & 0.64 & 1.03 & 2 & 0.62 \\ 
    
    {\bibsonomyone} & 2.05 & 6 & 0.99 & 1.71 & 11 & 0.64 & 2.15 & 13 & 0.66 & 1.42 & 12 & 0.59 \\
    {\bibsonomytwo} & 1.94 & 8 & 0.99 & 2.89 & 48 & 0.69 & 3.05 & 53 & 0.68 & 1.77 & 13 & 0.63 \\
    {\bibsonomythree} & 1.79 & 8 & 0.99 & 8.68 & 109 & 0.69 & 9.70 & 167 & 0.69 & 2.07 & 19 & 0.68 \\
    
    {\freelancer} & 1.87 & 5 & 0.98 & 8.57 & 88 & 0.69 & 8.59 & 88 & 0.69 & 2.21 & 6 & 0.63\\
    {\guru} & 2.06 & 12 & 0.98 & 14.06 & 85 & 0.69 & 13.21 & 84 & 0.69 & 1.56 & 9 & 0.65 \\
    \bottomrule
\end{tabular}
\end{table*}

\begin{table*}
\centering
\caption{Team characteristics of {\thresholdgreedy}, {\networkbalance} and {\greedyindividual} algorithms in terms of the average team density $\avgteamdensity$, and average team pairwise distance $\avgteampairwise$.}
\label{tab:team-characteristics-connectivity}
\tiny
\begin{tabular}{ l  rr  rr  rr  rr }
    \toprule
    \multirow{2}{*}{Dataset} &
    \multicolumn{2}{c}{\thresholdgreedy} &
    \multicolumn{2}{c}{\networkbalancegreedy} &
    \multicolumn{2}{c}{\networkbalancelp} &
    \multicolumn{2}{c}{\greedyindividual}  \\
    \cmidrule(lr){2-3} \cmidrule(lr){4-5} \cmidrule(lr){6-7} \cmidrule(lr){8-9}
    & $\avgteamdensity$ & $\avgteampairwise$ & $\avgteamdensity$ & $\avgteampairwise$ & $\avgteamdensity$ & $\avgteampairwise$ & $\avgteamdensity$ & $\avgteampairwise$ \\
    \midrule
    {\imdbone} & 1.02 & 0.52 & 1.07 & 0.41 & 1.07 & 0.38 & 1.05 & 0.34 \\
    {\imdbtwo} & 1.01 & 0.51 & 1.06 & 0.34 & 1.08 & 0.35 & 1.03 & 0.31 \\ 
    {\imdbthree} & 1.01 & 0.51 & 1.08 & 0.34 & 1.07 & 0.35 & 1.06 & 0.31 \\ 
    
    {\bibsonomyone} & 1.02 & 0.59 & 1.40 & 0.48 & 1.51 & 0.52 & 1.37 & 0.38 \\
    {\bibsonomytwo} & 1.04 & 0.58 & 1.13 & 0.65 & 1.35 & 0.61 & 1.52 & 0.45 \\
    {\bibsonomythree} & 1.04 & 0.57 & 1.52 & 0.74 & 1.38 & 0.75 & 1.29 & 0.56 \\
    
    {\freelancer} & 1.15 & 0.56 & 2.53 & 0.76 & 2.49 & 0.76 & 2.89 & 0.40 \\
    {\guru} & 2.80 & 0.58 & 9.76 & 0.85 & 10.92 & 0.78 & 2.82 & 0.34 \\
    \bottomrule
\end{tabular}
\end{table*}

\spara{Team size:} We characterize the size of a team (for a task) by the total number of experts assigned to that task. 
In Table~\ref{tab:team-characteristics-size}, we report the average team size formed by the different algorithms for the different datasets.

Overall, we observe that across all datasets, {\thresholdgreedy} consistently finds teams with the smallest sizes and highest task coverage. This is intuitive, since this algorithm doesn't have any graph constraints to satisfy. We also observe that {\thresholdgreedy} has a smaller variance in team sizes, since even the largest teams formed are significantly smaller than those formed by the {\networkbalance} algorithms.

For the {\imdb} datasets, all algorithms yield relatively small teams, with an average team size of a little over 1 expert. 
The smaller team sizes in the {\imdb} datasets could be attributed to the fact that there are relatively few skills in this dataset; often a single director is able to cover the skills of the tasks -- which typically have fewer skills.

For the {\bibsonomy} datasets, particularly {\bibsonomythree}, {\thresholdgreedy} and {\greedyindividual} have smaller average team sizes than the {\networkbalance} algorithms. While most teams formed by the {\networkbalance} algorithms are relatively small with under 10 experts, we observe that there are some teams that are much larger. 
We observe similar patterns for  {\freelancer} and {\guru}, where {\thresholdgreedy} finds teams that are smaller on average, with lower variance in the size than the {\networkbalance} algorithms.

\spara{Team radii:} 
We observe that {\thresholdgreedy} has teams with a much larger radius, of almost 1, since there is no radius constraint for {\thresholdgreedy}. 
For {\imdb} and {\bibsonomy}, we observe that the {\networkbalance} algorithms form most of their teams with radii that are just below the $r = 0.7$ constraint.
We observe that for  {\bibsonomy}, the {\networkbalance} and {\greedyindividual} algorithms have mean team radii of about 0.6, and several teams  with radii less than 0.5; this is not the case for {\thresholdgreedy}. 

For {\freelancer} and {\guru}, we see that the {\networkbalance} algorithms form more teams of varying radii, but still have average radii of about 0.6. The teams formed by {\thresholdgreedy} for these datasets still have the largest team radii, with means of 0.92 and 0.96, respectively 
(See Table~\ref{tab:team-characteristics-size}).

\spara{Team densities:} We define the density $\avgteamdensity$ of a team to be the sum of degrees of all experts in that team divided by the total number of experts on that team. This measure quantifies how well-connected the output teams are.
In Table~\ref{tab:team-characteristics-connectivity}, we report the average team density $\avgteamdensity$ achieved by the different algorithms. 
While we observed that the {\networkbalance} algorithms formed slightly larger teams than {\thresholdgreedy}, we now see that the former also outputs denser teams on average. 

\spara{Team pairwise distances:} We define the mean pairwise distance $\avgteampairwise$ of a team to be the mean of all pairwise shortest paths of experts on that team. The average pairwise distance of a team gives us an indication of how well connected experts on a team are. 
In Table~\ref{tab:team-characteristics-connectivity}, we report the mean pairwise distance $\avgteampairwise$ of teams formed by the different algorithms. 
We observe that for all three {\imdb} datasets and {\bibsonomyone}, {\thresholdgreedy} has the highest team mean pairwise distance of all the algorithms. However, we observe that for {\bibsonomytwo}, {\bibsonomythree}, {\freelancer} and {\guru}, the {\networkbalance} algorithms have a higher team mean pairwise distance.

Overall, we observe that {\thresholdgreedy} forms teams with fewer experts and smaller variance in team size compared to  {\networkbalance}. On the other hand, the {\networkbalance} algorithms form more compact teams than {\thresholdgreedy} in terms of the radii of teams. Additionally, the teams formed by the {\networkbalance} algorithms are significantly denser in terms of their connections between team members.

%% file: conclusion.tex
In this paper, we introduced two new team-formation problems: 
 {\bbalance}  and the more general {\nbalance} problem; we also designed algorithms for solving them.

In  {\bbalance}  the objective is to assign experts to tasks such that the total coverage of the tasks (in terms of their skills) is maximized and the maximum workload of any expert in the assignment is minimized. We proved that  {\bbalance}  is NP-hard.
We adopted a weaker notion of approximation~\citep{harshaw19submodular, mitra21submodularplus}, tailored for our objective, and -- within this setting -- we designed a polynomial-time approximation algorithm, {\thresholdgreedy} for {\bbalance}. 

In the {\nbalance} problem, we expand our {\bbalance} formulation  to include communication costs in a social graph. We have the same objective with the added constraint that every team in the expert-task assignment, $\assignment$ must also satisfy a radius constraint $\maxradius(\assignment) \leq r$. This problem is a generalization of {\bbalance}, and thus also NP-hard. 
For {\nbalance}, we designed {\networkbalance} a practical algorithm for solving it. This algorithm follows the same high-level algorithmic ideas we used for {\thresholdgreedy}, yet
it does not come with approximation guarantees.

For both problems, we showed that we can exploit the structure of our objective function 
and design speedups that work extremely well in practice.
We also developed a more general framework where we can efficiently tune the importance of the two parts of our objective and therefore make our framework applicable to a wide set of applications. 
Finally, we demonstrated the practical utility of the algorithmic framework we proposed in a variety of real datasets and also compared the characteristics of teams formed by the {\thresholdgreedy} and {\networkbalance} algorithms.
Our experiments with a variety of datasets from various domains demonstrated the utility of our framework and the efficacy of our algorithms.

%% file: appendix-proof.tex
Consider the following problem, which is solved for every threshold $\threshold$ in each iteration of the {\networkbalance} algorithm: given a set of preformed teams $\mathcal{T}$ find an assignment of teams to 
tasks such that the sum of task coverages is maximized while every expert's load is below a pre-specified threshold $\threshold$.  We call this problem  {\ntbalance} and we formally define it as follows:

\begin{problem}[{\ntbalance}]{\label{problem:ntbalance}}
Consider a set of $\numtasks$ tasks $\tasks=\{\task_1,\ldots , \task_\numtasks \}$, a set of $\numexperts$ experts $\experts = \{\expert_1,\ldots \expert_\numexperts \}$, and a set of $t$ teams $\mathcal{T} = \{T_1, \ldots, T_{t}\}$, such that~$T_i\subseteq \experts$. 
Also assume a load threshold $\threshold$. Let $x_{kj}=1$ if team $T_k$ is assigned to task $\task_j$ and let $C_{kj}$ denote the fraction of skills required by $\task_j$ covered by team $T_k$. 
Finally, let $M(i,k)=1$ if expert $\expert_i$ is in team $T_k$, 
and $M(i,k)=0$ if $\expert_i$ is not in $T_k$.
Our problem can be written as the following integer program:
%
\begin{align*}
	\text{maximize}~~ &\sum_{k=1}^t \sum_{j=1}^\numtasks C_{kj}x_{kj}\\
	\text{such that}~~ &\sum_{k=1}^t x_{kj} \leq 1, \qquad \qquad \qquad \text{ for all }  1 \leq j \leq \numtasks ,\\
	&\sum_{k=1}^t M(i,k) \sum_{j=1}^\numtasks x_{kj} \leq \threshold \qquad \text{ for all }  1 \leq i \leq \numexperts,      \text{ and}\\
    & x_{kj}\in \{0,1\}.
\end{align*}
\end{problem}


\begin{theorem}\label{thm:ntbalance}
The {\ntbalance} problem is \NP-hard.
\end{theorem}

\begin{proof}
We provide a proof of \NP-hardness via a reduction from {\independentset}. Given an undirected graph $G = (V,E)$ and an integer $\sigma \geq 0$, the decision version  of the {\independentset} problem is to
determine if there exists a subset of nodes $V' \subseteq V$ such that no two nodes in $V'$ share an edge, and $|V'| \geq \sigma$. For the purposes of this reduction, we consider the decision version of {\ntbalance} with an integer $\sigma \geq 0$, where the objective is to determine if their exists a solution  to the above program  such that $\sum_{k=1}^t \sum_{j=1}^\numtasks C_{kj}x_{kj} \geq \sigma$, i.e., the sum of task coverages is at least $\sigma$.

Given an arbitrary undirected graph $G = (V, E)$ with $|V| = t$ vertices, we create an instance $\texttt{TM}$ 
of {\ntbalance} 
in polynomial time. The idea is to map vertices in $G$ to teams in $\texttt{TM}$, and edges between pairs of vertices in $G$ correspond to \emph{distinct} experts shared by the corresponding teams in $\texttt{TM}$. We give the details below.

\begin{itemize}
    \item For each vertex $v \in V$ in $G$, create a team with a \emph{single} expert with $1$ \emph{unique} skill that is only associated with this expert in $\texttt{TM}$. Thus, create $t$ teams $\{T_1,\ldots , T_t\}$.

    \item For every edge $(v, v') \in E$ in $G$, add \emph{one} expert with \emph{zero} skills i.e $\{\emptyset\}$, to both the teams in $\texttt{TM}$ that correspond to the vertices $v$ and $v'$ in $G$.
    \item Create $t$ tasks $\task_1, \dots, \task_t$ to have exactly the set of skills corresponding to the $t$ teams $T_1,  \ldots, T_{t}$ in $\texttt{TM}$. This way, there is an 1-1 correspondence between team $T_i$ and task $\task_i$ and we call $\task_i$ the \emph{corresponding} task for team $T_i$.
\end{itemize}

Observe that $\texttt{TM}$ has $|E| + t$ experts, $t$ tasks, and $t$ teams, such that each task has \emph{exactly} one team that covers its skills \emph{completely}. 
Additionally, each team (and each task) has exactly one unique skill. 
Consider the $k$-th team and $j$-th task: if $k = j$, the coverage $C_{kj} = 1$, and if $k \neq j$, then $C_{kj} = 0$. Then it immediately follows that, given an assignment $\assignment$ that has a total task coverage $\sigma \geq 0$, there exists a set of $\sigma$ teams such that these $\sigma$ teams were assigned to their \emph{corresponding} tasks. 

We now show that the {\ntbalance} instance (\texttt{TM}) we have created has a team-task assignment $\assignment$ with total task coverage of $\sigma$ and load $\threshold = 1$ if and only if the corresponding instance of {\independentset} has an independent set of size $\sigma$ in $G$.

\spara{{\independentset} $\to$ {\ntbalance}.} Assume that $G$ has an independent set of size $\sigma$, i.e., it is possible to select $\sigma$ vertices in $G$ such that no two vertices share a common edge. This corresponds to picking $\sigma$ teams in the {\ntbalance} instance, such that no two teams share a expert. Since by our construction of the problem instance each team has exactly one corresponding task that can be fully covered, we can assign each of the $\sigma$ teams to its corresponding task and achieve a total task coverage of $\sigma$, within the load constraint $\threshold = 1$.

\spara{{\ntbalance} $\to$ {\independentset}.} Assume we have an assignment $\assignment$ for our {\ntbalance} instance with total task coverage of $\sigma \geq 0$, and every expert is assigned to at most $\threshold=1$ tasks. Now we can find an independent set of size $\sigma$ in $G$ as follows. We know that there exists a set of $\sigma$ teams that have been assigned to their corresponding tasks, to achieve a coverage of $\sigma$. Since $\assignment$ satisfies the load constraint $\threshold=1$, each expert is assigned to at most one task, and consequently no teams share an expert. Thus, if we select the $\sigma$ vertices in the {\independentset} instance that correspond to the $\sigma$ tasks in the {\ntbalance} instance, we have an independent set of size $\sigma$, such that no two vertices share a common edge.
\end{proof}

\vspace{5mm}
\begin{corollary}\label{corr:ntbalance-reduction}
{\ntbalance} cannot be approximated to within a factor better than $\bigO(n^{1-\epsilon})$.
\end{corollary}

\begin{proof}
For any $\epsilon > 0$, {\independentset} cannot be approximated to within a factor better than $\bigO(n^{1-\epsilon})$ in polynomial time unless $\mathbf{P}=\NP$~\citep{haastad1999clique}. 

We consider the optimization versions of {\independentset} and {\ntbalance}; the goal is to find the maximum independent set in an undirected graph $G = (V,E)$ in the former, and to find a feasible assignment such that the sum of coverages is maximized with load constraint $\threshold=1$ in the latter. 
We use the subscripts \texttt{IS} and \texttt{TM} to refer to instances $x$, solutions $y$, and the objective functions $m$ of these problems respectively. $\opt$ denotes the optimal value of the solution to either problem.


The reduction in Theorem~\ref{thm:ntbalance} describes a function that maps an instance $x_\texttt{IS}$ of {\independentset} to an instance $x_\texttt{TM}$ of {\ntbalance}. 
Given a solution~$y_\texttt{TM}$, Theorem~\ref{thm:ntbalance} describes a function that maps $y_\texttt{TM}$ back into a solution $y_\texttt{IS}$ of {\independentset}. Note that $\opt(x_\texttt{IS}) = \opt(x_\texttt{TM})$ and $m_\texttt{TM}(x_\texttt{TM},y_\texttt{TM}) = m_\texttt{IS}(x_\texttt{IS},y_\texttt{IS})$, since for any {\ntbalance} instance with (optimal) coverage of $\sigma$ the corresponding instance of {\independentset} has a (maximal) independent set of size~$\sigma$.

Towards a contradiction, assume {\ntbalance} can be approximated to within a factor $\gamma$ in polynomial time, such that $\gamma > \bigO(n^{1-\epsilon})$. Thus, there exists an instance $x_\texttt{TM}$, and a solution $y_\texttt{TM}$, such that $\gamma \opt(x_\texttt{TM}) \leq  m_\texttt{TM}(x_\texttt{TM},y_\texttt{TM})$. Then, we have:

\begin{align*}
\gamma &\leq  \frac{m_\texttt{TM}(x_\texttt{TM},y_\texttt{TM})}{\opt(x_\texttt{TM})} 
       = \frac{m_\texttt{IS}(x_\texttt{IS},y_\texttt{IS})}{\opt(x_\texttt{TM})} 
       \leq \frac{\bigO(n^{1-\epsilon}) \opt(x_\texttt{IS})}{\opt(x_\texttt{TM})} 
       = \bigO(n^{1-\epsilon}),
\end{align*}
which is a contradiction.
Thus, we conclude that {\ntbalance} cannot be approximated to within a factor better than $\bigO(n^{1-\epsilon})$.
\end{proof}

%% file: paper_dmkd.bbl

\begin{thebibliography}{33}
\ifx \bisbn   \undefined \def \bisbn  #1{ISBN #1}\fi
\ifx \binits  \undefined \def \binits#1{#1}\fi
\ifx \bauthor  \undefined \def \bauthor#1{#1}\fi
\ifx \batitle  \undefined \def \batitle#1{#1}\fi
\ifx \bjtitle  \undefined \def \bjtitle#1{#1}\fi
\ifx \bvolume  \undefined \def \bvolume#1{\textbf{#1}}\fi
\ifx \byear  \undefined \def \byear#1{#1}\fi
\ifx \bissue  \undefined \def \bissue#1{#1}\fi
\ifx \bfpage  \undefined \def \bfpage#1{#1}\fi
\ifx \blpage  \undefined \def \blpage #1{#1}\fi
\ifx \burl  \undefined \def \burl#1{\textsf{#1}}\fi
\ifx \doiurl  \undefined \def \doiurl#1{\url{https://doi.org/#1}}\fi
\ifx \betal  \undefined \def \betal{\textit{et al.}}\fi
\ifx \binstitute  \undefined \def \binstitute#1{#1}\fi
\ifx \binstitutionaled  \undefined \def \binstitutionaled#1{#1}\fi
\ifx \bctitle  \undefined \def \bctitle#1{#1}\fi
\ifx \beditor  \undefined \def \beditor#1{#1}\fi
\ifx \bpublisher  \undefined \def \bpublisher#1{#1}\fi
\ifx \bbtitle  \undefined \def \bbtitle#1{#1}\fi
\ifx \bedition  \undefined \def \bedition#1{#1}\fi
\ifx \bseriesno  \undefined \def \bseriesno#1{#1}\fi
\ifx \blocation  \undefined \def \blocation#1{#1}\fi
\ifx \bsertitle  \undefined \def \bsertitle#1{#1}\fi
\ifx \bsnm \undefined \def \bsnm#1{#1}\fi
\ifx \bsuffix \undefined \def \bsuffix#1{#1}\fi
\ifx \bparticle \undefined \def \bparticle#1{#1}\fi
\ifx \barticle \undefined \def \barticle#1{#1}\fi
\bibcommenthead
\ifx \bconfdate \undefined \def \bconfdate #1{#1}\fi
\ifx \botherref \undefined \def \botherref #1{#1}\fi
\ifx \url \undefined \def \url#1{\textsf{#1}}\fi
\ifx \bchapter \undefined \def \bchapter#1{#1}\fi
\ifx \bbook \undefined \def \bbook#1{#1}\fi
\ifx \bcomment \undefined \def \bcomment#1{#1}\fi
\ifx \oauthor \undefined \def \oauthor#1{#1}\fi
\ifx \citeauthoryear \undefined \def \citeauthoryear#1{#1}\fi
\ifx \endbibitem  \undefined \def \endbibitem {}\fi
\ifx \bconflocation  \undefined \def \bconflocation#1{#1}\fi
\ifx \arxivurl  \undefined \def \arxivurl#1{\textsf{#1}}\fi
\csname PreBibitemsHook\endcsname

\bibitem[\protect\citeauthoryear{Anagnostopoulos et~al.}{2010}]{anagnostopoulos10power}
\begin{bchapter}
\bauthor{\bsnm{Anagnostopoulos}, \binits{A.}},
\bauthor{\bsnm{Becchetti}, \binits{L.}},
\bauthor{\bsnm{Castillo}, \binits{C.}},
\bauthor{\bsnm{Gionis}, \binits{A.}},
\bauthor{\bsnm{Leonardi}, \binits{S.}}:
\bctitle{Power in unity: forming teams in large-scale community systems}.
In: \bbtitle{{ACM} Conference on Information and Knowledge Management, {CIKM}},
pp. \bfpage{599}--\blpage{608}
(\byear{2010})
\end{bchapter}
\endbibitem

\bibitem[\protect\citeauthoryear{Anagnostopoulos et~al.}{2012}]{anagnostopoulos2012online}
\begin{bchapter}
\bauthor{\bsnm{Anagnostopoulos}, \binits{A.}},
\bauthor{\bsnm{Becchetti}, \binits{L.}},
\bauthor{\bsnm{Castillo}, \binits{C.}},
\bauthor{\bsnm{Gionis}, \binits{A.}},
\bauthor{\bsnm{Leonardi}, \binits{S.}}:
\bctitle{Online team formation in social networks}.
In: \bbtitle{WWW}
(\byear{2012})
\end{bchapter}
\endbibitem

\bibitem[\protect\citeauthoryear{Anagnostopoulos et~al.}{2018}]{anagnostopoulos18algorithms}
\begin{bchapter}
\bauthor{\bsnm{Anagnostopoulos}, \binits{A.}},
\bauthor{\bsnm{Castillo}, \binits{C.}},
\bauthor{\bsnm{Fazzone}, \binits{A.}},
\bauthor{\bsnm{Leonardi}, \binits{S.}},
\bauthor{\bsnm{Terzi}, \binits{E.}}:
\bctitle{Algorithms for hiring and outsourcing in the online labor market}.
In: \beditor{\bsnm{Guo}, \binits{Y.}},
\beditor{\bsnm{Farooq}, \binits{F.}} (eds.)
\bbtitle{{ACM} {SIGKDD}},
pp. \bfpage{1109}--\blpage{1118}
(\byear{2018})
\end{bchapter}
\endbibitem

\bibitem[\protect\citeauthoryear{Bhowmik et~al.}{2014}]{bhowmik2014submodularity}
\begin{bchapter}
\bauthor{\bsnm{Bhowmik}, \binits{A.}},
\bauthor{\bsnm{Borkar}, \binits{V.}},
\bauthor{\bsnm{Garg}, \binits{D.}},
\bauthor{\bsnm{Pallan}, \binits{M.}}:
\bctitle{Submodularity in team formation problem}.
In: \bbtitle{SDM}
(\byear{2014})
\end{bchapter}
\endbibitem

\bibitem[\protect\citeauthoryear{Kargar et~al.}{2013}]{kargar2013finding}
\begin{bchapter}
\bauthor{\bsnm{Kargar}, \binits{M.}},
\bauthor{\bsnm{Zihayat}, \binits{M.}},
\bauthor{\bsnm{An}, \binits{A.}}:
\bctitle{Finding affordable and collaborative teams from a network of experts}.
In: \bbtitle{SDM}
(\byear{2013})
\end{bchapter}
\endbibitem

\bibitem[\protect\citeauthoryear{Kargar and An}{2011}]{kargar2011discovering}
\begin{bchapter}
\bauthor{\bsnm{Kargar}, \binits{M.}},
\bauthor{\bsnm{An}, \binits{A.}}:
\bctitle{Discovering top-k teams of experts with/without a leader in social networks}.
In: \bbtitle{CIKM}
(\byear{2011})
\end{bchapter}
\endbibitem

\bibitem[\protect\citeauthoryear{Kargar et~al.}{2012}]{kargar2012efficient}
\begin{bchapter}
\bauthor{\bsnm{Kargar}, \binits{M.}},
\bauthor{\bsnm{An}, \binits{A.}},
\bauthor{\bsnm{Zihayat}, \binits{M.}}:
\bctitle{Efficient bi-objective team formation in social networks}.
In: \bbtitle{ECML PKDD}
(\byear{2012})
\end{bchapter}
\endbibitem

\bibitem[\protect\citeauthoryear{Lappas et~al.}{2009}]{lappas2009finding}
\begin{bchapter}
\bauthor{\bsnm{Lappas}, \binits{T.}},
\bauthor{\bsnm{Liu}, \binits{K.}},
\bauthor{\bsnm{Terzi}, \binits{E.}}:
\bctitle{Finding a team of experts in social networks}.
In: \bbtitle{KDD}
(\byear{2009})
\end{bchapter}
\endbibitem

\bibitem[\protect\citeauthoryear{Majumder et~al.}{2012}]{majumder2012capacitated}
\begin{bchapter}
\bauthor{\bsnm{Majumder}, \binits{A.}},
\bauthor{\bsnm{Datta}, \binits{S.}},
\bauthor{\bsnm{Naidu}, \binits{K.}}:
\bctitle{Capacitated team formation problem on social networks}.
In: \bbtitle{KDD}
(\byear{2012})
\end{bchapter}
\endbibitem

\bibitem[\protect\citeauthoryear{Li et~al.}{2015a}]{li2015team}
\begin{botherref}
\oauthor{\bsnm{Li}, \binits{C.-T.}},
\oauthor{\bsnm{Shan}, \binits{M.-K.}},
\oauthor{\bsnm{Lin}, \binits{S.-D.}}:
On team formation with expertise query in collaborative social networks.
KAIS
(2015)
\end{botherref}
\endbibitem

\bibitem[\protect\citeauthoryear{Li et~al.}{2015b}]{li2015replacing}
\begin{bchapter}
\bauthor{\bsnm{Li}, \binits{L.}},
\bauthor{\bsnm{Tong}, \binits{H.}},
\bauthor{\bsnm{Cao}, \binits{N.}},
\bauthor{\bsnm{Ehrlich}, \binits{K.}},
\bauthor{\bsnm{Lin}, \binits{Y.-R.}},
\bauthor{\bsnm{Buchler}, \binits{N.}}:
\bctitle{Replacing the irreplaceable: Fast algorithms for team member recommendation}.
In: \bbtitle{WWW}
(\byear{2015})
\end{bchapter}
\endbibitem

\bibitem[\protect\citeauthoryear{Li et~al.}{2017}]{li2017enhancing}
\begin{botherref}
\oauthor{\bsnm{Li}, \binits{L.}},
\oauthor{\bsnm{Tong}, \binits{H.}},
\oauthor{\bsnm{Cao}, \binits{N.}},
\oauthor{\bsnm{Ehrlich}, \binits{K.}},
\oauthor{\bsnm{Lin}, \binits{Y.-R.}},
\oauthor{\bsnm{Bucher}, \binits{N.}}:
Enhancing team composition in professional networks: Problem definitions and fast solutions.
TKDE
(2017)
\end{botherref}
\endbibitem

\bibitem[\protect\citeauthoryear{Rangapuram et~al.}{2013}]{rangapuram2013towards}
\begin{bchapter}
\bauthor{\bsnm{Rangapuram}, \binits{S.S.}},
\bauthor{\bsnm{B{\"u}hler}, \binits{T.}},
\bauthor{\bsnm{Hein}, \binits{M.}}:
\bctitle{Towards realistic team formation in social networks based on densest subgraphs}.
In: \bbtitle{WWW}
(\byear{2013})
\end{bchapter}
\endbibitem

\bibitem[\protect\citeauthoryear{Yin et~al.}{2018}]{yin2018social}
\begin{botherref}
\oauthor{\bsnm{Yin}, \binits{X.}},
\oauthor{\bsnm{Qu}, \binits{C.}},
\oauthor{\bsnm{Wang}, \binits{Q.}},
\oauthor{\bsnm{Wu}, \binits{F.}},
\oauthor{\bsnm{Liu}, \binits{B.}},
\oauthor{\bsnm{Chen}, \binits{F.}},
\oauthor{\bsnm{Chen}, \binits{X.}},
\oauthor{\bsnm{Fang}, \binits{D.}}:
Social connection aware team formation for participatory tasks.
IEEE Access
(2018)
\end{botherref}
\endbibitem

\bibitem[\protect\citeauthoryear{Harshaw et~al.}{2019}]{harshaw19submodular}
\begin{bchapter}
\bauthor{\bsnm{Harshaw}, \binits{C.}},
\bauthor{\bsnm{Feldman}, \binits{M.}},
\bauthor{\bsnm{Ward}, \binits{J.}},
\bauthor{\bsnm{Karbasi}, \binits{A.}}:
\bctitle{Submodular maximization beyond non-negativity: Guarantees, fast algorithms, and applications}.
In: \bbtitle{International Conference on Machine Learning, {ICML}},
pp. \bfpage{2634}--\blpage{2643}
(\byear{2019})
\end{bchapter}
\endbibitem

\bibitem[\protect\citeauthoryear{Mitra et~al.}{2021}]{mitra21submodularplus}
\begin{botherref}
\oauthor{\bsnm{Mitra}, \binits{S.}},
\oauthor{\bsnm{Feldman}, \binits{M.}},
\oauthor{\bsnm{Karbasi}, \binits{A.}}:
Submodular + concave.
CoRR
\textbf{abs/2106.04769}
(2021)
\end{botherref}
\endbibitem

\bibitem[\protect\citeauthoryear{Hamidi~Rad et~al.}{2023}]{hamidi2023variational}
\begin{barticle}
\bauthor{\bsnm{Hamidi~Rad}, \binits{R.}},
\bauthor{\bsnm{Fani}, \binits{H.}},
\bauthor{\bsnm{Bagheri}, \binits{E.}},
\bauthor{\bsnm{Kargar}, \binits{M.}},
\bauthor{\bsnm{Srivastava}, \binits{D.}},
\bauthor{\bsnm{Szlichta}, \binits{J.}}:
\batitle{A variational neural architecture for skill-based team formation}.
\bjtitle{ACM Transactions on Information Systems}
\bvolume{42}(\bissue{1}),
\bfpage{1}--\blpage{28}
(\byear{2023})
\end{barticle}
\endbibitem

\bibitem[\protect\citeauthoryear{Kou et~al.}{2020}]{kou2020efficient}
\begin{bchapter}
\bauthor{\bsnm{Kou}, \binits{Y.}},
\bauthor{\bsnm{Shen}, \binits{D.}},
\bauthor{\bsnm{Snell}, \binits{Q.}},
\bauthor{\bsnm{Li}, \binits{D.}},
\bauthor{\bsnm{Nie}, \binits{T.}},
\bauthor{\bsnm{Yu}, \binits{G.}},
\bauthor{\bsnm{Ma}, \binits{S.}}:
\bctitle{Efficient team formation in social networks based on constrained pattern graph}.
In: \bbtitle{2020 IEEE 36th International Conference on Data Engineering (ICDE)},
pp. \bfpage{889}--\blpage{900}
(\byear{2020}).
\bcomment{IEEE}
\end{bchapter}
\endbibitem

\bibitem[\protect\citeauthoryear{Berkta{\c{s}} and Yaman}{2021}]{berktacs2021branch}
\begin{barticle}
\bauthor{\bsnm{Berkta{\c{s}}}, \binits{N.}},
\bauthor{\bsnm{Yaman}, \binits{H.}}:
\batitle{A branch-and-bound algorithm for team formation on social networks}.
\bjtitle{INFORMS Journal on Computing}
\bvolume{33}(\bissue{3}),
\bfpage{1162}--\blpage{1176}
(\byear{2021})
\end{barticle}
\endbibitem

\bibitem[\protect\citeauthoryear{Nikolakaki et~al.}{2021}]{nikolakaki21efficient}
\begin{bchapter}
\bauthor{\bsnm{Nikolakaki}, \binits{S.M.}},
\bauthor{\bsnm{Ene}, \binits{A.}},
\bauthor{\bsnm{Terzi}, \binits{E.}}:
\bctitle{An efficient framework for balancing submodularity and cost}.
In: \bbtitle{{ACM} {SIGKDD}},
pp. \bfpage{1256}--\blpage{1266}
(\byear{2021})
\end{bchapter}
\endbibitem

\bibitem[\protect\citeauthoryear{Dorn and Dustdar}{2010}]{dorn2010composing}
\begin{bchapter}
\bauthor{\bsnm{Dorn}, \binits{C.}},
\bauthor{\bsnm{Dustdar}, \binits{S.}}:
\bctitle{Composing near-optimal expert teams: a trade-off between skills and connectivity}.
In: \bbtitle{CoopIS}
(\byear{2010})
\end{bchapter}
\endbibitem

\bibitem[\protect\citeauthoryear{Nikolakaki et~al.}{2020}]{nikolakaki20finding}
\begin{botherref}
\oauthor{\bsnm{Nikolakaki}, \binits{S.M.}},
\oauthor{\bsnm{Cai}, \binits{M.}},
\oauthor{\bsnm{Terzi}, \binits{E.}}:
Finding teams that balance expert load and task coverage.
CoRR
\textbf{abs/2011.04428}
(2020)
\end{botherref}
\endbibitem

\bibitem[\protect\citeauthoryear{Selvarajah et~al.}{2021}]{selvarajah2021unified}
\begin{barticle}
\bauthor{\bsnm{Selvarajah}, \binits{K.}},
\bauthor{\bsnm{Zadeh}, \binits{P.M.}},
\bauthor{\bsnm{Kobti}, \binits{Z.}},
\bauthor{\bsnm{Palanichamy}, \binits{Y.}},
\bauthor{\bsnm{Kargar}, \binits{M.}}:
\batitle{A unified framework for effective team formation in social networks}.
\bjtitle{Expert Systems with Applications}
\bvolume{177},
\bfpage{114886}
(\byear{2021})
\end{barticle}
\endbibitem

\bibitem[\protect\citeauthoryear{Krause and Golovin}{2014}]{krause2014submodular}
\begin{barticle}
\bauthor{\bsnm{Krause}, \binits{A.}},
\bauthor{\bsnm{Golovin}, \binits{D.}}:
\batitle{Submodular function maximization.}
\bjtitle{Tractability}
\bvolume{3}(\bissue{71-104}),
\bfpage{3}
(\byear{2014})
\end{barticle}
\endbibitem

\bibitem[\protect\citeauthoryear{Lewis}{1983}]{lewis1983michael}
\begin{barticle}
\bauthor{\bsnm{Lewis}, \binits{H.R.}}:
\batitle{Computers and intractability. a guide to the theory of np-completeness}.
\bjtitle{The Journal of Symbolic Logic}
\bvolume{48}(\bissue{2}),
\bfpage{498}--\blpage{500}
(\byear{1983})
\end{barticle}
\endbibitem

\bibitem[\protect\citeauthoryear{Nemhauser and Wolsey}{1978}]{nemhauser78best}
\begin{barticle}
\bauthor{\bsnm{Nemhauser}, \binits{G.L.}},
\bauthor{\bsnm{Wolsey}, \binits{L.A.}}:
\batitle{Best algorithms for approximating the maximum of a submodular set function}.
\bjtitle{Math. Oper. Research}
\bvolume{3}(\bissue{3}),
\bfpage{177}--\blpage{188}
(\byear{1978})
\end{barticle}
\endbibitem

\bibitem[\protect\citeauthoryear{Vazirani}{2013}]{vazirani2013approximation}
\begin{bbook}
\bauthor{\bsnm{Vazirani}, \binits{V.V.}}:
\bbtitle{Approximation Algorithms},
(\byear{2013})
\end{bbook}
\endbibitem

\bibitem[\protect\citeauthoryear{Minoux}{1978}]{minoux1978accelerated}
\begin{bchapter}
\bauthor{\bsnm{Minoux}, \binits{M.}}:
\bctitle{Accelerated greedy algorithms for maximizing submodular set functions}.
In: \bbtitle{Optimization Techniques},
pp. \bfpage{234}--\blpage{243}
(\byear{1978})
\end{bchapter}
\endbibitem

\bibitem[\protect\citeauthoryear{Papadimitriou and Steiglitz}{1998}]{papadimitriou1998combinatorial}
\begin{bbook}
\bauthor{\bsnm{Papadimitriou}, \binits{C.H.}},
\bauthor{\bsnm{Steiglitz}, \binits{K.}}:
\bbtitle{Combinatorial Optimization: Algorithms and Complexity},
(\byear{1998})
\end{bbook}
\endbibitem

\bibitem[\protect\citeauthoryear{{Gurobi Optimization, LLC}}{2023}]{gurobi}
\begin{botherref}
\oauthor{\bsnm{{Gurobi Optimization, LLC}}}:
{Gurobi Optimizer Reference Manual}
(2023).
\url{https://www.gurobi.com}
\end{botherref}
\endbibitem

\bibitem[\protect\citeauthoryear{Khan et~al.}{2016}]{khan2016efficient}
\begin{barticle}
\bauthor{\bsnm{Khan}, \binits{A.}},
\bauthor{\bsnm{Pothen}, \binits{A.}},
\bauthor{\bsnm{Mostofa Ali~Patwary}, \binits{M.}},
\bauthor{\bsnm{Satish}, \binits{N.R.}},
\bauthor{\bsnm{Sundaram}, \binits{N.}},
\bauthor{\bsnm{Manne}, \binits{F.}},
\bauthor{\bsnm{Halappanavar}, \binits{M.}},
\bauthor{\bsnm{Dubey}, \binits{P.}}:
\batitle{Efficient approximation algorithms for weighted b-matching}.
\bjtitle{SIAM Journal on Scientific Computing}
\bvolume{38}(\bissue{5}),
\bfpage{593}--\blpage{619}
(\byear{2016})
\end{barticle}
\endbibitem

\bibitem[\protect\citeauthoryear{Benz et~al.}{2010}]{benz2010social}
\begin{barticle}
\bauthor{\bsnm{Benz}, \binits{D.}},
\bauthor{\bsnm{Hotho}, \binits{A.}},
\bauthor{\bsnm{Jäschke}, \binits{R.}},
\bauthor{\bsnm{Krause}, \binits{B.}},
\bauthor{\bsnm{Mitzlaff}, \binits{F.}},
\bauthor{\bsnm{Schmitz}, \binits{C.}},
\bauthor{\bsnm{Stumme}, \binits{G.}}:
\batitle{The social bookmark and publication management system {BibSonomy}}.
\bjtitle{The VLDB Journal}
\bvolume{19}(\bissue{6}),
\bfpage{849}--\blpage{875}
(\byear{2010})
\doiurl{10.1007/s00778-010-0208-4}
\end{barticle}
\endbibitem

\bibitem[\protect\citeauthoryear{H{\aa}stad}{1999}]{haastad1999clique}
\begin{botherref}
\oauthor{\bsnm{H{\aa}stad}, \binits{J.}}:
Clique is hard to approximate within n 1- $\varepsilon$
(1999)
\end{botherref}
\endbibitem

\end{thebibliography}
